%% file: main.tex
\documentclass{lmcs}
\pdfoutput=1
\usepackage[utf8]{inputenc}

\usepackage{lastpage}
\lmcsdoi{21}{3}{25}
\lmcsheading{}{\pageref{LastPage}}{}{}%
{May~07,~2024}{Sep.~05,~2025}{}

\usepackage[T1]{fontenc}

\usepackage{amsmath}
\usepackage{amssymb}
\usepackage{amsthm, thmtools, thm-restate}
\usepackage{stmaryrd}
\usepackage{mathtools}

\usepackage{mathpartir}

\usepackage{booktabs}

\usepackage{url}

\usepackage{xcolor}
\usepackage{xspace}

\usepackage{davide}
\usepackage[]{macro}
\usepackage{commonmacros}

\usepackage[hypertexnames=false]{hyperref}

\begin{document}
\title[Extensional and Non-extensional Functions as Processes]{Extensional and Non-extensional \texorpdfstring{\\}{} Functions as Processes}

\author[K.~Sakayori]{Ken Sakayori\lmcsorcid{0000-0003-3238-9279}}[a]
\author[D.~Sangiorgi]{Davide Sangiorgi\lmcsorcid{0000-0001-5823-3235}}[b,c]

\address{The University of Tokyo, Japan}

\address{Universit\`a di Bologna, Italy}

\address{Inria, France}	%

\input{abstract}

\maketitle

\input{intro}
\input{background}

\input{EpiI}
\input{abs-enc}
\input{oenc}
\input{wires}

\input{BT-LT}
\input{Dinf}
\input{related-work}

\section*{Acknowledgment}
We thank the referees, both those of LICS 2023, and those of LMCS,  for  insightful
comments and  valuable suggestions.
This work has been supported by the MIUR-PRIN projects `Analysis of Program Analyses' (ASPRA, ID
201784YSZ5\_004),  
`Resource Awareness in Programming: Algebra, Rewriting, and
Analysis' (RAP, ID P2022HXNSC), by 
JSPS KAKENHI Grant Number JP24K20731,
and by the European Research Council (ERC) Grant DLV-818616 DIAPASoN.

\bibliographystyle{alphaurl}

\bibliography{main}

\clearpage
\appendix
\input{notation}
\clearpage
\input{oenc-appx}

\input{wires-appx}
\input{BT-LT-appx}
\end{document}

%% file: abstract.tex
\begin{abstract}
Following Milner's seminal paper,
the representation of functions as
processes  has received considerable attention.
For  pure $\lambda$-calculus, the process representations yield (at
best)
 \emph{non-extensional}
 $\lambda$-theories (i.e., $\beta$ rule holds, whereas $\eta$
 does not).

In the paper, we study how  to obtain  \emph{extensional}
representations, and how to move  
 between 
extensional and
 non-extensional representations. 
Using Internal $\pi$, \piI\ (a subset of the $\pi$-calculus in which
all outputs are bound),
we develop a refinement of Milner's original encoding of functions as
processes that is \emph{parametric} on certain abstract components called
\emph{wires}. These are, intuitively, processes whose task is to
connect two end-point channels.
We show that when  a few algebraic properties
of wires hold,
 the encoding yields a
$\lambda$-theory.
Exploiting the symmetries and dualities of \piI, we
isolate three main classes of wires.  The first two have a sequential
behaviour and are dual of each
other; the third has a parallel behaviour and is the dual of
itself.
We show the  adoption of the parallel  wires yields an
extensional $\lambda$-theory; in fact, it yields an equality that
coincides with that of
Böhm trees with infinite $\eta$.
In contrast,  the other two classes of wires yield
non-extensional $\lambda$-theories whose
 equalities 
are those  of the  Lévy-Longo   and Böhm trees.

\end{abstract}

%% file: intro.tex
\section{Introduction}
\label{s:intro}

 Milner's  work~\cite{Milner90,Milner92}  on the  encoding of
the pure \( \lambda \)-calculus  into the \( \pi
\)-calculus is  generally considered a landmark paper  in the area of
semantics and programming languages. 
The encoding of the $\lambda$-calculus is a significant test of
expressiveness for the $\pi$-calculus. 
The encoding  also gives  an interactive semantics to the \( \lambda
\)-calculus, which allows one to analyse it using
the instruments available in the \( \pi \)-calculus.
After Milner's
seminal work, a number of  encoding variants have been
put forward (e.g.~\cite{SangiorgiWalker01} and references therein)
 by modifying the target language (often to a
subcalculus of the \( \pi \)-calculus) or the encoding itself.  The
correctness of these encodings is usually supported by the operational
correspondence against a certain evaluation strategy of the \( \lambda
\)-calculus and by the validity of the \( \beta \)-rule, $(\lambda x.
M) N = M \sub N x$. (In this paper, by validity of a
$\lambda$-calculus rule with respect to a certain process encoding
$\enco{\,\cdot\, }$, we mean that $\enco M \wb \enco N$ for all
instances $M =N$ of (the congruence closure of) the rule, where $\wb$
is a basic behavioural equivalence for the pure processes, such as
  ordinary bisimilarity.)

The equality on \( \lambda \)-terms induced by the encoding has also been investigated; in
this equality two  $\lambda$-terms $M$ and $N$ are equal when their images are
behaviourally equivalent processes.  
For   Milner's original (call-by-name) encoding, such an equality  coincides with
the Lévy-Longo tree (\qLT{}) equality~\cite{Sangiorgi93,Sangiorgi00} %
(the result is by large independent of the behavioural
 equivalence adopted for the processes~\cite{SangiorgiXu18}).
It has also been shown how to recover 
the Böhm tree (\qBT{}) equality~\cite{SangiorgiWalker01}, by
 modifying Milner's encoding~--- allowing reductions underneath a  $\lambda$-abstraction~---
 and selecting divergence-sensitive behavioural equivalences on processes such as 
 must-testing.

Tree structures play a pivotal role in the $\lambda$-calculus.
For instance,   trees allow one to  unveil the computational
content hidden in a $\lambda$-term, with respect to some relevant
minimal information. 
In \qBTs{} the  information is the head normal forms, whereas in
\qLTs{} it is
the weak head normal forms.
\qBT{} and \qLT{} equalities coincide with the local structures of well-known models  of the
$\lambda$-calculus, such as Plotkin and Scott's $P_\omega$~\cite{Plo72,Scott76}, and
the \emph{free lazy Plotkin-Scott-Engeler  models}~\cite{Levy76,Eng81,Lon83}.

In \qBT{}s and \qLT{}s, the  computational content of a $\lambda$-term is unveiled using
 the $\beta$-rule alone. Such
structures are sometimes called \emph{non-extensional}, as opposed to
the \emph{extensional} structures, in which the $\beta$-rule   is coupled
with the  $\eta$-rule, 
$M = \lambda x. \app M x$  (for $x$ not free in $M$).
In extensional theories two functions are equated
if, whenever applied to the same argument,  they yield equal results.
A well-known
extensional tree-structure are \qBTs{} with infinite $\eta$, shortly
 \qBTinfs{}.
The equality of \qBTinf{}s, also known as the equational theory \( \mathcal{H}^*\), coincides with
that  of 
  Scott's $\Dinf$ model~\cite{Scott76}, historically the first model
of the untyped $\lambda$-calculus.  
A seminal result by Wadsworth~\cite{Wadsworth76} shows that the \qBTinfs{} are intimately related to the
head normal forms, as 
the \qBTinf{} equality coincides with contextual equivalence in which
the  head normal forms are the
observables.

In  representations of functions as processes, extensionality and
the $\eta$-rule, even in their most
basic 
form,  have always appeared out of reach. For instance, in Milner's
 encoding, 
 $x$ and
$\lambda y . \app x y$ have quite different behaviours: the former process
is a single output particle, whereas the latter has an
infinite behaviour and, moreover, the initial action is an input.

The general goal of this paper is to study extensionality in the
representation of functions as processes. In particular, we wish to
understand if and how
one can derive extensional representations, and the difference between
extensional and non-extensional representations from a process perspective.

We outline the main technical contributions.
We develop a refinement of Milner's original encoding of functions, using
Internal $\pi$ (\piI), a subcalculus of the $\pi$-calculus in which
only bound names may be exported. 
The encoding makes use of  certain abstract components called
\emph{wires}. These are, intuitively, processes whose task is to
connect two end-point channels; and when one of the two end-points is
restricted, the wires  behave as substitutions.
In the encoding, wires are called `abstract' because their definitions are not made
explicit. 
We show that assuming
a few basic algebraic properties
of wires
(having to do with transitivity of wires and substitution) is
sufficient to obtain a $\lambda$-theory, i.e.~the validity of the $\beta$-rule.

We then delve into the  impact of the concrete definition of the wires, notably on the
equivalence on $\lambda$-terms induced by the encoding. 
In the $\pi$-calculus literature, the
most common form of wire between two channels $a$ and $b$ is written
$!a(u). \out b u$ (or $a(u). \out b u$, if only needed once), and
sometimes called  a \emph{forwarder}~\cite{HoYo95,Mer00thesis}.
 In \piI, free outputs
are forbidden and such a wire
becomes a  recursively-defined process. 
We call this kind of wires \emph{\IOwires{}}, because of their
`input before output' behaviour. 
Exploiting the properties of \piI, e.g., its symmetries and dualities, we identify two
other main kinds of wires: the \emph{\OIwires{}},  with an `output before input' behaviour
and 
which are  thus the 
 dual of the  \IOwires{};    and the \emph{\Pwires{}}, or \emph{parallel wires},
 where input and output can fire concurrently (hence such wires are behaviourally the same as their dual).

We show that  moving among these three kinds of wire 
corresponds to moving among the three above-mentioned tree structures of the
 $\lambda$-calculus, namely \qBT{}s, \qLT{}s, \qBTinf{}s.
Precisely, we obtain \qBT{}s
when adopting the ordinary \IOwires{};
\qLT{}s when adopting the \OIwires{};
and  \qBTinf{}s when adopting the \Pwires{}.
This also implies that  \Pwires{} allow us  to validate
the $\eta$-rule (in fact both $\eta$ and   infinite $\eta$).
The results are summarised in Table~\ref{table:instances-of-abs-enc}, where $\QencoX$ is the concrete encoding in which the \texttt{X} wires are used.
\par
\begin{table}[tb]
  \centering
  \caption{Instances of the abstract encoding}
  \begin{tabular}{c c c}
    \toprule
    Encoding & Parameter (wires)   & Characterises  \\
    \midrule
    \( \QencoIO \) &\IOwires&  \( \qBT \)   \\
    \hline
    \( \QencoP \) & \Pwires  &\( \qBTinf \)   \\
    \hline
    \( \QencoOI \) & \OIwires & \( \qLT \)   \\
    \bottomrule
  \end{tabular}
  \label{table:instances-of-abs-enc}
\end{table}
We are not aware of results in
 the literature that produce an \emph{extensional} $\lambda$-theory from a
 processes model, let alone that derive the 
\qBTinf{} equality.
We should also stress that the choice of the  wire
    is the  \emph{only} modification needed  for switching among the three tree structures: 
 the encoding of the $\lambda$-calculus is otherwise
 the same, nor does it change the underlying calculus and its behavioural equivalence
 (namely, \piI\  and bisimilarity).  

There are various reasons for using \piI{} in our study. The first
and most important
 reason has to do with the symmetries
 and dualities of \piI, as hinted above.
 The second reason is proof techniques: in the paper  we  use
 a wealth of proof techniques, ranging from algebraic
laws to  forms of `up-to bisimulation' and  to  unique solutions of equations; not all of
them are available in the ordinary $\pi$-calculus.
The third reason has to do with $\eta$-rule. In  studies of the expressiveness of \piI{} in the literature~\cite{Boreale98}
the encoding  of the free-output construct
 into \piI{} %
resembles an (infinite) \( \eta \)-expansion.
The essence of the encoding is the following transformation (which needs to be recursively applied to eliminate all free outputs):
\begin{equation}
  \label{eq:intro:fout-out-capability}
  \out a p \mapsto \res q \left( \out a q \mid \iproc q {\seq y} {\out p {\seq y}} \right).
\end{equation}
A free output of \( p \) is replaced by a bound output, that is, an
output of a freshly created name \( q \) (for simplicity, we assume
  that $p$ is meant to be used only once by the recipient).
The transformation requires 
 \emph{localised} calculi~\cite{MerroSangiorgi04}, in which
the recipient of a name 
may only  use it in output, 
and
 resembles an $\eta$-expansion of a variable of the
$\lambda$-calculus in that, intuitively,  direct access to the name  
 \( p \) is replaced by  access
   to the function \( \lambda \seq y . \out p {\seq y} \).

A possible 
 connection between \piI{} and \( \eta \)-expansion
may also be found in papers such as~\cite{CairesPT16}, where
 \( \eta \)-expanded proofs (proofs in which the identity rule
 is only applied to  atomic formulas) are related to 
(session-typed) processes  with   bound outputs only.
Yet, the technical link with our works appears weak because the wires
that we  use to  achieve extensionality (the \Pwires{} of
Table~\ref{table:instances-of-abs-enc}) are behaviourally quite different from  the process
structures mentioned above.

We derive the encoding into \piI\
  in two steps. The first step consists,
intuitively, in
  transplanting Milner's encoding into \piI, by 
 replacing  free outputs  with  bound outputs plus 
wires, following the idea in  (\ref{eq:intro:fout-out-capability})
above.  However,
 (\ref{eq:intro:fout-out-capability}) is only valid in  
localised calculi, whereas 
  Milner's encoding also requires
 the \emph{input}  capability of names to be transmitted.
Therefore
we have to modify the wire in
(\ref{eq:intro:fout-out-capability}), essentially inverting the two
names $p$ and $q$. The correctness of the resulting
transformation relies on properties about the usage  of names that are specific to the
representation of functions.
The second step adopted
to derive the encoding  consists of allowing reductions underneath a $ \lambda $
abstraction; that is, implementing a \emph{strong} reduction strategy.
This transformation is necessary in order to mimic the computation
required to obtain head normal forms. 

Encodings of strong
reduction strategies have appeared in the literature; they rely on the possibility of
encoding non-blocking prefixes (sometimes called \emph{delayed} in the literature)~\cite{Abramsky94,BellinScott94,Fu97,PaVi97b,Mer00thesis,MerroSangiorgi04,SangiorgiWalker01}, i.e.,
prefixes $\mu\!:\! P$ in which actions from $P$ may fire before $\mu$, as long as $\mu$
does not bind names of the action.
The encodings of non-blocking prefixes in the literature
require the names bound in $\mu$ to be localised.
Here again, the difficulty was to adapt the schema to
non-localised names. 
Similar issues arise within  wires, as their definition also requires
certain prefixes to be non-blocking.

\noindent \emph{Structure of the paper}
Section~\ref{sec:background} recalls background material on $\lambda$-calculus and
\piI{}.
Section~\ref{sec:EpiI} introduces wires and permeable prefixes.
In Section~\ref{sec:abs-enc}, we present
the abstract encoding, using the abstract wires, and the assumptions we
make on wires; we then verify that  such assumptions are sufficient to
obtain a $\lambda$-theory.
Section~\ref{sec:oenc} defines an optimised abstract encoding,
which will be useful for later proofs. 
In Section~\ref{sec:wires}, we introduce the three classes of concrete wires,
and show that they satisfy the required assumptions for wires. 
In Section~\ref{sec:BT-LT}, we pick the \IOwires and \OIwires, and  prove full abstraction for
  \qLT{}s and \qBT{}s.
In Section~\ref{sec:Dinf}, we do the same for the \Pwires{}
and   \qBTinf{}s.
Section~\ref{sec:related-work} discusses further related work and possible  future developments.
For readability, some proofs are only given in the Appendices.

%% file: background.tex
\section{Background}
\label{sec:background}

A tilde  represents  a tuple.
The $i$-th element of a tuple $\til P$ is  referred to as $P_i$. 
All notations are extended to tuples componentwise.
\input{lambda}
\input{piI}

%% file: lambda.tex
\subsection{The \texorpdfstring{\( \lambda \)}{lambda}-calculus}
\label{sec:lambda}

We let $x$ and $y$ range over the set of $\lambda$-calculus  variables.
The set   $\Lao$  of $\lambda$-terms is  
 defined by the  grammar  
\begin{center}
$ M \Coloneqq \;   x    \midd   \lambda   x. M  \midd
  \app{M_1} {M_2} \, .$
\end{center}
Free variables, closed terms, substitution, $\alpha$-conversion etc.\ are  defined as usual~\cite{Barendregt84}; the set of free variables of $M$  is  $\fv M$.
Here and in the rest of the paper (including when reasoning about
 processes), we adopt the usual
`Barendregt convention'.
This will allow us to assume freshness
of bound variables and names  whenever needed.
 We group brackets on the left; therefore $M N L $ is $(M N ) L$.
 We  abbreviate $\lambda  x_1. \cdots. \lambda  x_n.M $   as 
 $\lambda  x_1 \cdots x_n.M $, or $\lambda  \tilde{x}. M$.

A number  of reduction relations are mentioned in this paper.
The (standard) \emph{\( \beta \)-reduction} relation \( M \red N \) is the relation on \( \lambda \)-terms induced by the following rules:
{\small \[
\begin{array}{ll}
[\beta]\; \infer%
{ }{\app {(\lambda x.M)} N \red M \sub N x}
&\qquad [ \mu ] \; \infer{N \red N'}{\app M N \red \app M {N'}} \\
{[ \nu ]}\;  \infer{M \red M'}{\app M N \red \app {M'} N}
&\quad [ \xi ] \; \infer{M \red M'}{\lambda x. M \red \lambda x. M'}
\end{array}
\]
}
The (weak) \emph{call-by-name reduction} relation uses only
the \( \beta \) and \( \nu \) rules,  
whereas  \emph{strong call-by-name}, written \( \snred \), also has \( \xi \);
the \emph{head reduction}, written \( \hred \), is a deterministic  variant of \( \snred \) in which
the redex contracted
 is the head one, i.e., \( \app {(\lambda y. M_0)} { M_1} \) of \( \lambda \seq x . \app {(\lambda y. M_0)} {M_1 \cdots M_n} \).
\emph{Head normal forms} are of the form  \( \lambda \seq x . \app y \seq M
\).
As usual, we use a double arrow to indicate the reflexive and
transitive closure of a reduction relation, as in \( \wred \)
and $\Hred$.
A term \( M \) \emph{has a head normal form} \( N \) if \( M \Hred N \) and \( N \) is the (unique) head normal form.
Terms that do not have a head normal form are called \emph{unsolvable}.
An unsolvable \( M \) has an \emph{order of unsolvability \( n
  \)}, if \( n \) is the largest natural number such that \( M \Hred
\lambda x_1\ldots x_n. M' \), for $n \geq 0$, and some
$x_1,\ldots,x_n,M'$.
If there is no such largest number, then \( M \) is \emph{of order \( \omega \)}.
For instance,  if
  \( \Delta\) is  \(\lambda x. \app x  x \), and  \( \Omega\) is \( \app
  \Delta \Delta \), and \( \ogre \) is \( \app {(\lambda x y. \app x x)}
  {(\lambda x y. \app x x)} \), 
then we have $  \Omega
 \hred \Omega$,  therefore \( \Omega \) is an unsolvable of order \( 0 \), and 
 \( \lambda x. \Omega \)  is of order \( 1 \); 
whereas 
\( \ogre   \hred \lambda  y.   \ogre \), indeed for any $n$ we have 
\( \ogre
 \Hred
\lambda y_1\ldots y_n. \ogre\), therefore  
\( \ogre \) is an unsolvable of order $\omega$.

We recall the (informal) definitions of  Lévy-Longo trees and Böhm trees, and of
Böhm trees up-to infinite \texorpdfstring{\( \eta \)}{eta}-expansion.
The \emph{Lévy-Longo tree} of \( M \) is the labelled tree, \( \LT M \), defined coinductively as follows:
\begin{enumerate}
\item
 \( \LT M = \top \) if \( M \) is an unsolvable of order \( \omega \);
\item
 \( \LT M = \lambda x_1 \ldots x_n. \bot \) if \( M \) is an unsolvable of order \( n < \omega \);
\item
\( \LT M  \) is the tree with $\lambda \seq x.y$ as the root and \( \LT {M_1} \ldots \LT {M_n} \) as the children, if \( M \)  has head normal form \( \lambda \seq x.\app y {M_1 \cdots  M_n} \) with \( n \geq 0 \).
\end{enumerate}
The definition of Böhm trees (\qBT{}s) is obtained from that of
\qLT{}s using \qBT{} in place of \qLT{}, and
demanding that \( \BT M = \bot \) whenever \( M \) is unsolvable (in
place of clauses  (1) and (2)). 

An \( \eta \)-expansion of a \qBT{}, whose root is \( \lambda \seq{x}. y \) and children are \( \BT{M_1}, \ldots \BT{M_n}\), is given by a tree whose root is \( \lambda \seq {x} z. y \) and children are \( \BT{M_1}, \ldots \BT{M_n}, z \).
Intuitively, an infinite \( \eta \)-expansion of a \qBT{} is obtained by allowing this
expansion at each step of the clause~(3). 
Thus, an  (informal) coinductive definition of the \emph{Böhm trees up-to infinite \texorpdfstring{\( \eta \)}{eta}-expansion} of
$M$,
\( \BTinf M \), is given as follows:
\begin{enumerate}
  \item \( \BTinf M = \bot \) if \( M \) is unsolvable, and
  \item if \( M \Hred \lambda x_1 \ldots x_m. \app y {M_1 \cdots M_n} \) for \( m \ge 0 \) and \( n \ge 0 \), then
        \begin{enumerate}
          \item \( \BTinf M = \BTinf {\lambda x_1 \ldots x_{m - 1}. \app y {M_1 \cdots
                M_{n - 1}}} \), if \( x_m \notin \fv{\app y {M_1 \cdots M_{n - 1}}}\) and
            \( \BTinf {M_n} = \BTinf {x_m} \);
and otherwise
          \item \( \BTinf M = \)
                \begin{tikzpicture}[level distance=10mm,sibling distance=4mm,baseline={([yshift={-\ht\strutbox}]current bounding box.north)}]
                  \node {$\lambda x_1 \ldots x_m.y$} [grow=down]
                  child {node{\( \BTinf {M_1}\)}}
                  child[missing]   {node{2}}
                  child[missing]   {node{3}}
                  child { node {$\cdots$} edge from parent[draw=none]  node {}}
                  child[missing] {node{5}}
                  child[missing] { node {6} }
                  child {node{\( \BTinf {M_n} \)}};
                \end{tikzpicture}
        \end{enumerate}
\end{enumerate}

In the
equality induced by the above  trees, two terms  are related if their trees
are the same (as usual modulo $\alpha$-conversion). These equalities may be defined
coinductively as forms of bisimilarity on $\lambda$-terms (\cite{Sangiorgi93}), in the expected
manner.
We only present   \emph{\qBTinf{}-bisimilarity}, because it is the most delicate one and also because it will be used in a few important proofs.
In such a bisimilarity,   
first introduced by Lassen~\cite{Lassen99},  a (finite) \( \eta \)-expansion is allowed at each step of the bisimulation game.

\begin{defiC}[\cite{Lassen99}]
\label{def:lassen-bisim}
 A relation \( \relR \) on \( \lambda \)-terms is a
 \emph{\qBTinf{}-bisimulation} if, whenever \( M \relR N \), either
 one of the following holds: 
 \begin{enumerate}
   \item \( M \) and \( N \) are unsolvable
   \item \( M \Hred \lambda x_1 \ldots x_{l+m}. \app y {M_1 \cdots M_{n + m} } \) and \( N \Hred \lambda x_1 \ldots x_l. \app y {N_1 \cdots N_n} \),
     where the variables \( x_{l + 1}, \ldots, x_{l + m }  \) are not free in \( \app y {N_1 \cdots N_n} \), and
 \( M_i \relR N_i\) for \( 1 \le i \le n \),  and also \( M_{n + j} \relR x_{l + j } \) for \( 1 \le j \le m\)
    \item the symmetric case, where \( N \) reduces to a head normal
    form with more leading \( \lambda \)s. 
 \end{enumerate}
 The largest \qBTinf{}-bisimulation is called \emph{\qBTinf-bisimilarity}.
 We also write \( \BTinf M = \BTinf N \) when \( M \) and \( N \) are \(\qBTinf \)-bisimilar.
\end{defiC}
Two terms are indeed  \qBTinf{}-bisimilar precisely when
they have the same Böhm tree up-to infinite \texorpdfstring{\( \eta \)}{eta}-expansion.

\begin{exa}
\label{ex:BT-LT}
We provide some examples of Böhm trees and Lévy-Longo trees, using the terms 
  \( \Delta\),  \( \Omega\), and \( \ogre\) introduced earlier. 
\[
  \begin{array}{llll}
  \BT \Delta =
  \begin{tikzpicture}[baseline={([yshift=-.8ex]current bounding box.center)}, level distance=7mm]
  \node {\( \lambda x . x \)}
  child {node{\( x \)}};
  \end{tikzpicture}
  &
  \BT \Omega = \bot
  &
  \BT {\lambda x. \Omega} = \bot
  &
  \BT \ogre = \bot
  \\
  \LT \Delta =
  \begin{tikzpicture}[baseline={([yshift=-.8ex]current bounding box.center)}, level distance=7mm]
  \node {\( \lambda x . x \)}
  child {node{\( x \)}};
  \end{tikzpicture}
  &
  \LT \Omega = \bot
  &
  \LT {\lambda x. \Omega} = \lambda x. \bot
  &
  \LT \ogre = \top
  \end{array}
  \]

  For all the terms \( M \) used in this example, we have \( \BT M = \BTinf M \).
\end{exa}

\begin{exa}
\label{ex:inf-eta}
  Let \( \Einf \) be a term such that \( \app \Einf z \Hred \lambda y. \app z (\Einf y)
  \); for instance, we can take $\Einf$ to be  \( \app Y {(\lambda f x y. \app x {(\app f y)})}\) using the fixpoint combinator \( Y = \lambda f. \app {(\lambda x . \app f {(\app x x)})}  {(\lambda x . \app f {(\app x x)})} \).
  Intuitively, the term \( \app \Einf z \) can be considered as the `limit of the sequence of \( \eta \)-expansions'
  \[
    z \red_{\eta} \lambda z_1. \app z {z_1} \red_{\eta} \lambda z_1. \app z (\lambda z_2. \app {z_1} {z_2}) \red_{\eta} \cdots
  \]
  In other words, \( \app \Einf z \) is a term that represents an `infinite \( \eta \)-expansion' of \( z \).
  The terms \( z \) and \( \app \Einf z \)  have different Böhm trees:
  \[
  \BT {\app \Einf z} =
  \begin{tikzpicture}[baseline={([yshift=-.8ex]current bounding box.center)}, level distance=8.5mm]
  \node {\( \lambda z_1 . z \)}
  child {node{\( \lambda z_2. z_1 \)}
    child {node{\( \lambda z_3 . z_2 \)}
      child {node {\( \vdots \)}}
    }
  };
  \end{tikzpicture}
  \qquad
  \BT z = z.
  \]

  However, \( \BTinf{\app \Einf z} = \BTinf{z} \) as the
  two terms can be equated using an infinite form of $\eta$-expansion.
  To show that they  are  \qBTinf{}-bisimilar, 
let \( \relR \) be the set of pairs of the form \( (x, \app \Einf x) \) and \( \BEbisim \) be  \qBTinf{}-bisimilarity.
 Then  \( {\relR} \cup {\BEbisim} \) is a \qBTinf{}-bisimulation.
To see this, assume that \( M ({\relR} \cup {\BEbisim}) N \).
  We only consider the case where \( M \) is a variable, say \( x \), and \( N \) is \( \app \Einf x \).
  We have \( x \Hred x \) and \( \Einf x \Hred \lambda y. \app x (\app \Einf y) \).
  This  closes the case because \( x \BEbisim x \) and \( \app \Einf y \relR y \).

\end{exa}

%% file: piI.tex
\subsection{Internal \texorpdfstring{\( \pi \)}{pi}-calculus}
\label{sec:piI}

In all  encodings we consider, the encoding of a  $\lambda$-term
 is parametric on  a name, 
i.e., it is a %
 function from names  to $\pi$-calculus  
 processes. We also need parametric processes (over one or several names) for writing recursive process definitions
 and equations. 
We call such  parametric processes {\em abstractions}.
The actual instantiation of the parameters of an abstraction $F$ is done
via the {\em application} construct $\appI F \tila$.
We use $P,Q$ for processes, $F$ for abstractions.
Processes and abstractions  form the set  of {\em $\pi$-agents} (or simply \emph{agents}), ranged
over by $A$. 
Small letters 
 $a,b, \ldots, x,y, \ldots$  
range  over the infinite set of names.
The grammar of the internal \( \pi \)-calculus (\piI{}) is thus:
$$ \begin{array}{ccll}
A & \Coloneqq & P \midd F & \mbox{(agents)}\\[\mypt]
P & \Coloneqq & \nil    \midd    \inp a \tilb . P    \midd    \bout a \tilb . P
   \midd    \res a P
& \mbox{(processes)}  \\[\myptSmall]
& &
   \midd     P_1 |  P_2   \midd  ! \inp a \tilb . P   
\midd \appI F \tila
\\[\mypt]
F & \Coloneqq & \bind \tila P \midd \cnst & \mbox{(abstractions)}
   \end{array}
 $$
The operators used  have the usual meanings.
Here, \( \bout a \tilb \) is a \emph{bound output}, which may be thought of as an abbreviation \( \res {\til b} \bar a \langle \til b \rangle \) using the syntax of the standard \( \pi \)-caclulus.
Having only bound output constructs is a key characteristic of \piI{}, which poses a symmetry between input and output constructs.
In  prefixes $\inp a \tilb$ and $\bout a \tilb$, we call 
 $a$ the {\em subject}.
We  
 often abbreviate 
$\res a \res b P$ as $\resb{
a,b} P$.
Prefixes,  restriction, and  abstraction   are binders and give
rise in the expected way  to the definition of {\em free names},
{\em
bound names}, and {\em names} of an agent, 
respectively indicated with $\fn -$,  $\bn -$, and
$\n -$, as well as  that of  $\alpha$-conversion.  
An agent is \emph{name-closed}  if it does not contain free names. 
In the grammar, \( \cnst \) is  a \emph{constant},  used to write  recursive definitions. Each
constant \( \cnst \) has a defining equation of the form
\( \cnst \defeq  \bind{\tilx} P \), where $\bind{\tilx} P  $ is
name-closed;
$\tilx$
 are the formal parameters of
the constant (replaced by the actual parameters whenever the constant
is used).
Replication could be avoided in the syntax since it can be encoded
 with recursion. However its semantics is simple, and it is 
a useful
 construct for encodings; thus we chose to include it in  the grammar.

For convenience, we set some conventions.
An {application redex} $\appI{((\til x) P)}{\til a}$ can be normalised as $P \sub {\til x}{\til a}$. An agent is
\emph{normalised} if all such application redexes have been
contracted. %
When reasoning
 on behaviours
it is useful to assume that all expressions are normalised, in the above sense.    Thus
in the remainder of the paper \emph{we identify an agent with its normalised expression}.
 As in the $\lambda$-calculus, following the 
usual Barendregt convention
we   identify
processes or actions which only  differ on the choice of the 
bound names.
The symbol $=$
indicates syntactic  identity; in the case of terms or objects of the calculi,
it  will
 mean  `syntactic identity modulo $\alpha$-conversion'.

Since the calculus is polyadic, we assume a \emph{sorting system}~\cite{Milner93} to avoid disagreements  in the arities of the tuples of names carried by a given name and in applications of abstractions.
In Milner's encoding (written in the polyadic $\pi$-calculus)
as well as in all encodings in the paper, there are only two sorts of
names: \emph{location names}, and    \emph{variable names}.
Location names carry a pair of a variable name and
a  location name; variable names carry a single location name. 
Using $p,q,r,...$ for location names, and 
$x,y,z,...$ for variable names, the forms of the possible prefixes 
are:
\[ \inp p{x,q}. P \midd \inp x{p}.P \midd 
\bout p{x,q}. P \midd \bout x{p}.P 
\] 
This sorting will be maintained throughout the paper. Hence
process transformations and algebraic laws will be given with
reference to such a  sorting.

The operational semantics of \piI{} is standard~\cite[Section 7.5]{SangiorgiWalker01}, and reported in Figure~\ref{fig:LTS}.
Transitions are of the form \( \proc \arr \act \proc' \), where the set of \emph{actions} is given by
\[
\act \Coloneqq \bin a {\seq b} \midd \bout a {\seq b} \midd \tau
\]
and bound names of \( \act \) are \emph{fresh}, i.e., they  do not appear free in \( \proc \).
The meaning of  actions is the  usual one: \( \bin a {\seq b} \) and \( \bout a {\seq b}\) are bound input and outputs, respectively, and \( \tau \) is the internal action.
The co-action \( \bar \act \) of \( \act \) used in the rule \LTS{com} is defined by \( \overline {\bin a {\seq b}} \defeq \bout a {\seq b} \) and \( \overline {\bout a {\seq b}} \defeq \bin a {\seq b} \).
There is only one rule for communication and one rule for restriction, which makes the LTS simpler than that of the \( \pi \)-calculus.
This is because we do not need to care about the difference between free and bound outputs as all the outputs in \piI{} are bound.
Since we are using the Barendregt convention, in \LTS{par}, we have an implicit side condition \( \bn \act \cap \fn Q = \emptyset \).
The \( \abs{\til y} Q \equival \abs{\til x} P \) in the premise of rule \LTS{com} merely means that the two agents are equal up-to \( \alpha \)-conversion.
As usual, we write \( \wtrans{\ }\) for the reflexive transitive closure of \( \trans \tau \), \( \wtrans{\act} \) for \( \wtrans{\  } \trans \act \wtrans{\ } \),
and  \( \wtrans {\hat \act}\) is \( \wtrans {\act} \) for \( \act  \neq \tau \) and \( \wtrans{\ } \) otherwise.

\begin{figure*}[!t]
\begin{center}
\begin{tabular}{rlclrrcl}
[\LTS{pre}]&
\infer{ }{\act.P \arr\act {P}}
&[\LTS{par}]
\infer{P \arr\act  P'}
{ P | Q   \arr\act P'| Q }
\\[\mysp]
[\LTS{res}]& \infer{{  P} \arr{\act}{P'} \qquad  x \not \in \n\act}
{\res x P   \arr{ \act} \res x P'}
&
\multicolumn{4}{c}{
    [\LTS{com}] \; \infer{P \arr{ \act } P' \qquad Q \arr{\outC\act} Q' \qquad \act\neq \tau,   \; \tilx = \bn\act}
  {P |  Q \arr{ \tau} \res{  \tilx} (P'|  Q') }
}
 \\[\mysp]
    [\LTS{rep}] &
\infer{ }
 { !\act. P  \arr{ \act}  P' | !\act. P }
&
\multicolumn{4}{c}{
    [\LTS{con}] \;
\infer{P \arr \act P' \qquad \cnst  \defeq \abs{\til y} Q \qquad \abs{\til y} Q \equival \abs{\til x} P}
 {\appI {\rmmm K}{\til x}  \arr\act P'}
}
\end{tabular}
\end{center}
\caption{The  standard LTS  for \piI{}.}
\label{fig:LTS}
\end{figure*}

The reference behavioural equivalence for \piI\ is (weak)
bisimilarity.
It is known that, in \piI{}, bisimilarity coincides with barbed
congruence  (for processes that are image-finite up to weak bisimilarity; all the agents obtained as encodings of \( \lambda \)-terms will be image-finite up to weak bisimilarity).

\begin{defi}[Bisimilarity]
\label{def:bisim}
  A symmetric relation \( \relR \) over processes is a \emph{(weak) bisimulation}, if, whenever \( \proc \relR \proctwo \),
  \begin{itemize}
    \item \( \proc \trans \act \proc' \) implies \( \proctwo \wtrans {\hat \act} \proctwo' \) for some \( \proctwo' \)  such that \( \proc' \relR \proctwo' \).
  \end{itemize}
  Processes \( \proc \) and \( \proctwo \) are \emph{(weakly) bisimilar}, written \( \proc \bisim \proctwo \) if there exists a bisimulation \( \relR \) such that \( \proc \relR \proctwo \).

In a few places we will also use
   \emph{strong bisimilarity}, written \( \sbisim \); this is
 defined analogous to \( \bisim \), but \( \proctwo \) must respond
 with a (strong) transition \( \proctwo \trans \act \proctwo' \).
\end{defi}

We also use the \emph{expansion} preorder, written \( \lexp \), an asymmetric  variant of \( \bisim \)
in which, intuitively, \( \proc \lexp \proctwo \) holds if  $P \wb
Q$ but also $Q $ has at least as many $\tau$-moves as $P$.
\begin{defi}[Expansion]
\label{def:expansion}
A relation \( \relR \) over processes is an \emph{expansion} if \( P \relR Q \) implies:
\begin{enumerate}
\item  Whenever
  $P \trans \act P'$,  there exists  $Q' $  such that \ $Q \Arr\act Q'$ and $P'  \relR  Q'$;
\item whenever $Q \arr\act Q'$, there exists   $P' $  such that \ $P \arcap\mu P'$ and $P'  \RR  Q'$.
\end{enumerate}
Here \( \arcap \act \) is the strong version of \( \wtrans {\hat \act }\), that is, \( \arcap \act \) is \( \arr \act \) if \( \act\neq \tau \) and is \( = \) or \( \trans \tau  \) if \( \act = \tau \).
Recall that \( \wtrans \tau \) means at least one \( \tau \)-action.
Therefore, when \( P \trans \tau P' \), \( Q \) must respond using at least one \( \tau \)-action  whereas, when \( Q \trans \tau Q' \), \( P \) must respond using at most one \( \tau \)-action.
We say that $Q$  \emph{expands} $P$, written
$P \lexp  Q$, if $P  \relR  Q$,  for some expansion \( \relR \).
\end{defi}

Finally, we sometimes use \emph{structural congruence}, when we need to emphasize that two
processes are the same modulo some structural rewriting of their syntax.

\begin{defi}[Structural congruence]
\label{def:scong}
The \emph{structural congruence} \( \scong \) is the smallest congruence relation over processes that includes the \( \alpha \)-equivalence and the following equivalence:
\begin{gather*}
  \proc \mid \nil \scong \proc \quad \proc \mid \proctwo \scong \proctwo \mid \proc \quad (\proc \mid \proctwo) \mid \procthree \scong \proc \mid (\proctwo \mid \procthree) \\
  \res x \res y \proc \scong \res y \res x \proc \quad (x \neq y) \\
  \res x (\proc \mid \proctwo) \scong \res x \proc \mid \proctwo \quad (x \notin \fv \proctwo).
\end{gather*}
\end{defi}

We summarise the inclusion order between the process relations that we use:
\[
  {\scong} \subsetneq {\sbisim} \subsetneq {\lexp} \subsetneq {\bisim}.
\]
In \piI{}, \emph{all the relations above are (pre)congruences}~\cite{Sangiorgi96a}.
All behavioural relations are extended to abstractions by requiring ground
instantiation of the parameters.
For instance,  $(x)P  \wb (x)Q$ if $P \wb Q$.
Appendix~\ref{a:nota} summarises notations used in the paper (term relations,
encodings, etc.).

\subsubsection{Proof techniques}
\label{sec:proof-techniques}
Here we review known proof-techniques for \Ipi{} that we will use in this paper.
These are well-known algebraic laws, notably laws for private replications, up-techniques for bisimilarity, and unique solutions of equations.

\subsubsection*{Algebraic laws}
We will often use a group of laws about
private replication.  That is, laws for processes of the form \( \res a
(\proc \mid \riproc a {\seq b} \proctwo) \), where \( a \)
appears free only in output position of \( \proc \) and \( \proctwo \).
Since the process \( \riproc a {\seq b} \proctwo \) is replicated and
private it can be distributed or pushed inside prefixes. These laws
are known as the replication theorems~\cite[Section 2.2.3]{SangiorgiWalker01}.
\begin{lem}
\label{lem:rep-thm}
  Suppose \( x \) occurs free only in output subject position of \( \proc \), \( \proctwo \) and \( \procthree \).
  Then we have
  \begin{enumerate}
    \item \( \res x (\proctwo \mid \procthree \mid \riproc x p \proc ) \sbisim  \res x (\proctwo \mid \riproc x p \proc) \mid  \res x (\procthree \mid \riproc x p \proc) \);
          \label{it:lem:rep-thm:par}
    \item \( \res x ( \pi.\proctwo \mid  \riproc x p \proc ) \sbisim  \pi. \res x (\proctwo \mid  \riproc x p \proc ) \), if \( \pi \) is a non-replicated prefix, subject of \( \pi \) is not \( x \) and \( \bn \pi \cap \fn {\riproc x p \proc} = \emptyset \);
    \item \( \res x (\riproc y q \proctwo \mid \riproc x p \proc) \sbisim  \riproc y q \res x(\proctwo \mid \riproc x p \proc) \);
    \item \( \res x (\proctwo \mid \riproc x p \proc) \sbisim \proctwo \), if \( x \notin \fn \proctwo \);
          \label{it:lem:rep-thm:gc}
    \item \( \res x (\bout x p . \proctwo \mid \riproc x p \proc) \gexp \res x \res p(\proctwo \mid \proc \mid \riproc x p \proc) \).
          \label{it:lem:rep-thm:comm}
  \end{enumerate}
\end{lem}
We often call the law~\eqref{it:lem:rep-thm:gc} the `garbage collection law'.

\subsubsection*{Up-to techniques}
Our main up-to technique will be \emph{up-to context and expansion}~\cite{Sangiorgi96b}, which
admits the  use of contexts and of behavioural equivalences such as expansion to achieve the closure of a relation in the bisimulation game.
\begin{defi}[Bisimulation up-to context and \( \lexp \)]
\label{def:up-to-ctx}
A symmetric relation \(  {\relR}  \) on \Ipi{}-processes is a \emph{bisimulation up-to context and up-to \( \lexp \) if \( \proc \relR \proctwo \) and \( \proc \arr \act   \proc'' \)
imply that there are a (possibly multi-hole) context $\qct$ and processes  \( \seq \proc' \) and \( \seq \proctwo' \) such that  \( \proc'' \gexp \ct {\seq \proc'} \), \( \proctwo \Arcap\act  \contr  \ct{\seq \proctwo'} \) and \( \seq \proc'  \relR  \seq \proctwo' \).
Here \( \seq \proc'  \relR  \seq \proctwo' \) means \( \proc'_i \relR \proctwo'_i \) for each component.}
\end{defi}
A special instance of this technique is the \emph{up-to expansion} technique where the common context \( C \) is taken as the empty-context.

\begin{thm}
  If \( \relR \) is a bisimulation up-to context and expansion then \( {\relR} \subseteq {\bisim } \).
\end{thm}
We refer the readers to~\cite{Sangiorgi96b} for the proof.
(Adapting the proof to \piI{} and 
arbitrary
 multi-hole contexts is straightforward.)
We also use the technique of \emph{expansion} up-to context and \( \lexp \), which is defined analogously to Definition~\ref{def:up-to-ctx}, as a proof technique  to prove expansion results
 (the main difference is that one now
requires  \( \proc'' \lexp \ct {\seq \proc'} \)).

\subsubsection*{Unique solution of equations}

We  briefly recall the `unique solution of equations' technique~\cite{DurierHS19}.
\emph{Equation variables} \( X, Y, Z \) are used to write equations.
The body of an equation is a name-closed abstraction
possibly containing equation  variables
(that is,  applications can also be of the form $\appI X\tila$).
We use $E$ to range over %
expression bodies; and
 $\EE$ to range over systems of equations, defined as follows.
 In all the definitions, the indexing set $I$ can be infinite.

\begin{defi}
Assume that, for each $i$ of 
 a countable indexing set $I$, we have a variable $X_i$, and an expression
$E_i$, possibly containing   variables. 
Then 
$\{  X_i = E_i\}_{i\in I}$
(sometimes written $\til X = \til E$)
is 
  a \emph{system of equations}. (There is one equation for each
  variable $X_i$.)
  A system of equations is \emph{guarded} if each
  occurrence of a variable in the body of an equation is underneath a
  prefix.
\end{defi}

We write $E[\til F]$ for the abstraction obtained  by
replacing in $E$ each occurrence of the variable $X_i$ 
with the abstraction $F_i$. 
This is a  syntactic replacement, with instantiation of the parameters:
e.g.,  
replacing $X$ with $(\til x)P$ in $X\param{\til a}$ amounts to %
replacing $X\param{\til a}$ with the process $P\sub{\til a}{\til x}$.

\begin{defi}\label{d:un_sol}
Suppose  $\{  X_i = E_i\}_{i\in I}$ is a system of equations. We say that:
\begin{itemize}
\item
 $\til F$ is a \emph{solution of the 
system of equations  for $\wb$} 
if for each $i$ it holds
that $F_i \wb E_i [\til F]$.
\item The  system has 
\emph{a unique solution for $\wb$}  if whenever 
$\til F$ and $\til G$ are both solutions for $\wb$, we have %
$\til F \wb \til G$. \end{itemize}
 \end{defi}

 \begin{defi}[Syntactic solutions]
The    \emph{syntactic solutions}
of a system of equations  $\{  X_i = E_i\}_{i\in I}$
   are the  
   recursively defined constants $\KEi E \defeq E_i[\KE]$, for
   $i\in I$. %
\end{defi}
The syntactic solutions
of a system of equations %
are indeed solutions of
it. 
The  unique-solution technique   relies on an analysis of
  divergences. 
A process $P$ \emph{diverges} if it can perform an
infinite sequence of internal moves, possibly after some visible ones
(i.e., actions different from $\tau$). Formally, this holds if there
are processes $P_i$, $i\geq 0$, and some $n$ such that
$P=P_0\arr{\mu_0} P_1 \arr{\mu_1} P_2 \arr{\mu_2}\dots$ and for all
$i>n$, $\mu_i=\tau$. We call a \emph{divergence of $P$} the sequence
of transitions $\big(P_i\arr{\mu_i}P_{i+1}\big)_{i\geq 0}$.  
An abstraction $F$
 has a divergence 
if the process $\appI F\tila$ has a divergence,  where $\tila$ are fresh
names. 

\begin{thmC}[\cite{DurierHS22}]\label{thm:usol}
A guarded system of equations %
whose   syntactic 
solutions are agents with no divergences  has a unique solution for $\wb$. 
\end{thmC}

%% file: EpiI.tex
\section{Wires and Permeable Prefixes}
\label{sec:EpiI}

We introduce the abstract notion of \emph{wire}
process; and, as a syntactic sugar, the process constructs for 
 \emph{permeable prefixes}.  Wires and permeable prefixes will play a
 central role in the technical development in the following
 sections. 

\subsubsection*{Wires}
We use the notation \( \glink a b \) for an  \emph{abstract wire};
this is, intuitively, a special process whose purpose is
 to connect the output-end of $a$ with the input-end of $b$
(thus \( \glink a b \) itself will use $a$ in input and $b$ in
output). We call such wires `abstract' because
we will not give a  definition for them. 
We  only state (Section~\ref{sec:abs-enc})
 some behavioural properties that are expected to hold, and that 
 have mainly to do with substitutions; approximately:
\begin{enumerate}
\item
 if $P$ uses $b $  only in input, then $ \res b ( \glink a b | P  )\gexp  P \sub ab $

\item
 dually, if $P$ uses $a $  only in output, then $ \res a ( \glink a b | P  )\gexp  P \sub ba$

\end{enumerate}
Further conditions will however be needed on $P$ for such properties
 to hold (e.g., in (1), the input at $b$ in $P$ should be `at the
top-level', and in (2), the outputs at $a$ in $P$ should be
`asynchronous'.)
Special cases of  (1) and (2) are  forms of transitivity for wires, with
the common name restricted:
\begin{enumerate}
\item[(3)]
 \( \res b (\glink ab \mid \glink bc) \gexp  (\glink bc) \sub ab =  (\glink ab) \sub cb  = \glink ac \).
\end{enumerate}
When (1) holds we say that $P$ is 
\emph{I-respectful with respect to \( \glink ab \)}; similarly when
(2) holds we say that $P$ is 
\emph{O-respectful with respect to \( \glink ab \)}. 
When (3) holds, for any $a,b,c$ of the same sort, we say that 
\emph{wires are transitive}.

As we have two sorts of names in the paper (location
names and variable names), we will correspondingly deal with two sorts
of wires,  \emph{location wires} and  
\emph{variable wires}.
In fact, location wires will be the key
structures.
Thus when, in Section~\ref{sec:wires}, we consider \emph{concrete} instantiations of the location wires,
the corresponding  definitions of the variable wires will be adjusted so to maintain the expected properties of the location wires.

\subsubsection*{Permeable prefixes}
We write
  \( \diprfx a {\seq b} \proc \) and \( \oprfx a {\seq b} \proc
\) for  a \emph{permeable input} and a \emph{permeable  output}. 
Intuitively, a permeable prefix only blocks actions
involving the bound names of the prefix.
For instance, a permeable
input $\diprfx a {\tilb} \proc$, as an ordinary input, is capable of
producing   action $\bin a {\tilb}$  thus yielding  the
derivative $P$. However, in contrast with the ordinary input, in 
$\diprfx a {\tilb} \proc$
  the 
process  $P$ is active and running, 
and can thus interact with the outside environment as well as with the prefix
  \( \diprfx a {\til b} \proc \) (though such a form of  interaction will not occur in processes encoding of \( \lambda
  \)-terms); the only
constraint is that  the actions
from $P$ involving the bound names $\tilb$ cannot fire for as long as
 the top
 prefix $a(\tilb)$ is not consumed.

Given the two sorts of names that will be used in the
paper, the possible forms of permeable prefixes are:
\[
\diprfx p {x,q} \proc 
\midd \oprfx p {x,q}   \proc  \midd
\diprfx x {p} \proc 
\midd \oprfx x {p}   \proc
\]
Moreover, it will always be the case that, in a prefix
$\diprfx p {x,q} \proc$, process $P$ uses  $x$ only in output and
$q$ only once in input; and conversely, \( P \) uses \( x \) only in output and \( q \) only once in output in  $\oprfx p {x,q} \proc$.
Similarly,  in
$\diprfx x {p} \proc$, 
 process $P$  uses  $p$ only once in input; and conversely, in $\oprfx x {p} \proc$, name \( p \) is used only once in output position in \( P \).

We stress that permeable prefixes 
 should be taken as syntactic sugar; formally they are defined from
 ordinary prefixes and wires as follows 
\begin{align*}
  \diprfx p {x, q} \proc &\defeq \resb {x, q} \left(\iproc p {x', q'}{\left(\glink x {x'} \mid \glink {q'} q \right)} \mid \proc \right) \\
  \diprfx x p \proc &\defeq \res p \left(\iproc x {p'} {\glink {p'} p} \mid  \proc \right) \\
  \oprfx p {x, q} \proc &\defeq \resb {x, q} \left(\bout p {x', q'}. \left(\glink {x'} x \mid \glink q {q'} \right) \mid \proc \right) \\
  \oprfx x p \proc &\defeq \res p \left(\bout x {p'}. \glink p {p'} \mid  \proc \right)
\end{align*}
Such definitions thus depend on the concrete forms of wires adopted.
The definitions behave as intended only when the processes underneath the permeable prefixes are respectful.
For example, we have
\begin{align*}
  &\diprfx p {x, q} \proc \trans{p(x, q)} \scong
  \resb{x', q'}\left( \glink {x'} {x} \mid \glink {q} {q'} \mid \proc \sub {x', q'} {x, q}  \right) \gexp \proc
\end{align*}
if \( \proc \) is I-respectful with respect to \( \glink {q'} q \) and
O-respectful with respect to
 \( \glink x {x'} \), for fresh \( q' \) and \( x' \).

Later, when  the abstract wires will be instantiated to
{concrete} wires,  we will study properties of
I-respectfulness and O-respectfulness, as well as, correspondingly,
properties of permeable prefixes,   
in the setting of encodings of  functions.

We end this section by introducing some algebraic laws for permeable prefixes.
These laws allow us to avoid desugaring while proving the properties of the encodings.

\begin{lem}
\label{lem:permeable-comm}
Suppose that wires, which are used to define the permeable prefixes, are transitive.
  \begin{enumerate}
    \item  \( \res p (\diprfx p{x, q} \proc \mid \oprfx p {x, q} \proctwo ) \gexp \resb {x,  q}(\proc \mid \proctwo) \) if either \( \proc \) is I-respectful with \( \glink q {q'} \) and O-respectful with \( \glink {x'} x \) or \( \proc \) is O-respectful with \( \glink q {q'} \)  and I-respectful with \( \glink {x'} x \).
    \item  \( \res x (\riproc x p \proc \mid \oprfx x p \proctwo  \mid \procthree) \gexp \res x (\riproc x p \proc \mid \res p (\proc \mid \proctwo)  \mid \procthree)  \) if \( p \notin \fn \procthree \), \( x \) does not appear in an input subject position of \( \proc, \proctwo, \procthree\) and either \( \proc \) is I-respectful with \( \glink p {p'} \) or \( \proctwo \) is O-respectful with \( \glink {p'} p \).
  \end{enumerate}
\end{lem}

\begin{lem}
\label{lem:permeable-scong}
The following structural laws hold.
  \begin{enumerate}
    \item If \( a \neq b\) and \( a \notin \seq c \) then \( \res a \diprfx b {\seq c} \proc \scong  \diprfx b {\seq c} \res a \proc  \) and \( \res a \oprfx b {\seq c} \proc \scong  \oprfx b {\seq c} \res a \proc  \)
    \item If \( \seq b \cap \fn{\proctwo} = \emptyset \) then \( \proctwo  \mid    \diprfx a {\seq b} \proc
 \scong \diprfx a {\seq b} (\proctwo \mid \proc) \) and \( \proctwo \mid \oprfx a {\seq b} \proc  \scong \oprfx a {\seq b} (\proctwo \mid \proc) \).
  \end{enumerate}
\end{lem}

Lemma~\ref{lem:permeable-comm} expresses a property of a restricted interaction consuming  permeable prefixes.  Lemma~\ref{lem:permeable-scong} shows two structural laws for permeable prefixes, one concerning restriction, the other
concerning parallel composition. These lemmas simply follow from the syntactic definition of the permeable prefixes.
The assumption about the transitivity of wires, in Lemma~\ref{lem:permeable-comm}, is harmless as all the wires we use in this paper are transitive.
In what follows, we will use these structural rules without explicitly mentioning Lemma~\ref{lem:permeable-scong}.

%% file: abs-enc.tex
\section{Abstract Encoding}
\label{sec:abs-enc}
\renewcommand*{\enco}{\encoN}

This section introduces the abstract encoding of \( \lambda \)-terms
into \( \piI \)-processes. We call the encoding `abstract' because it
uses  the abstract wires
discussed in the previous section.  In other words,  the encoding
is \emph{parametric} with respect to the concrete
 definition of the wires.
We then
 prove that,
whenever  the wires satisfy a few basic  laws,
 the
encoding yields a  \(
\lambda \)-theory.

\subsection{Definition of the abstract encoding}
\label{sec:abs-enc-def}

We begin by recalling Milner's original encoding $\QencoM$ of (call-by-name)
$\lambda$-calculus into the  $\pi$-calculus~\cite{Milner92,Milner93}:
\begin{align*}
  \QencoM \enco x p &\defeq \out x p \\
  \QencoM \enco {\lambda x . M} p &\defeq  \iproc p {x, q} {\QencoM \enco M q} \\
  \QencoM \enco {\app M N} p &\defeq \resb {q, x} \left(\QencoM \enco M q \mid \out q {x, p} \mid \riproc x r {\QencoM \enco N r} \right) \\
\end{align*}

The encoding of a
$\lambda$-term  $M$ is parametric over a port $p$, which
  can be thought of as the
\emph{location} of $M$, for $p$ represents  the
unique port along which $M$ may be called
by its environment, thus receiving two names:  (a trigger  for) its argument and  the  location to be
used for the next interaction.
Hence,  \( \QencoM \) (as well as  all the encodings in the paper) is
a function from \( \lambda \)-terms to abstractions of the form \(
\abs p \proc \). We write \( \QencoM \enco M p \) as a shorthand
for \(
\appI{\QencoM \enco M {}} p \).
A function application of the $\lambda$-calculus becomes, in the
$\pi$-calculus,
 a particular form of
parallel combination of two agents, the function and its argument; $\beta$-reduction is then modelled as process interaction.
An argument of an application is translated as a replicated server,
that can be used as many times as needed, each time providing
 a name to be used as location for the following computation.

\begin{figure}[t]
\begin{align*}
  \Qencoabs \enco x p &\defeq \oprfx x {p'} \glink p {p'} \\
  \Qencoabs \enco {\lambda x . M} p &\defeq  \diprfx p {x, q} \Qencoabs \enco M q \\
  \Qencoabs \enco {\app M N} p &\defeq \res q (
  \Qencoabs \enco M q
  \mid \oprfx q {x, p'} (\riproc x r {\Qencoabs \enco N r}  \mid  \glink p {p'}))
\end{align*}
\caption{The abstract encoding $\QencoA$.}
\label{fig:abs-enc}
\end{figure}

In Figure~\ref{fig:abs-enc} we report the
abstract encoding  \( \Qencoabs \).
There are two main modifications from Milner's encoding \( \QencoM \):
\begin{enumerate}
\item The encoding uses \piI, rather than $\pi$-calculus; for this, all
  free outputs are replaced  by a combination of
 bound outputs and wires (as hinted in Introduction).
\item A permeable input is used, in place of an ordinary input, in
  the translation of abstraction so to allow reductions underneath a
  $\lambda$-abstraction. (We thus implement a \emph{strong}
  call-by-name  strategy.)
\end{enumerate}

We  report a few basic conditions that will be required on wires.
The main ones concern the behaviour of wires as substitutions and
transitivity of wires.

\begin{defi}[Wires]
  \label{ax:glink-piI}
  As a convention, we assume that names \( a, b, c \) are of the same sort, either location names or variable names.
  \emph{Wires} \( \glink a b \) are processes that satisfy the following properties:
  \begin{enumerate}
    \item The free names of \( \glink a b \) are \( a \) and \( b \).
          Furthermore,  \( \glink a b \) only uses \( a \) in input and \( b \) in output.
          \label{it:ax:glink-piI:fn}
    \item If \( \glink a b \trans \act \proc \) for some \( \proc \), then \( \act \neq \tau \).
    \label{it:ax:glink-piI:no-tau}
    \item \( \res b (\glink a b \mid \glink b c) \gexp \glink a c \).
          \label{it:ax:glink-piI:trans}
    \item \( \res q (\glink p q \mid \diprfx q {x, r} \proc) \gexp \diprfx p {x, r} \proc \), provided that \( (\res {x, r})(\glink x {x'} \mid \glink {r'} r \mid \proc ) \gexp \proc \sub {x' , r'} {x, r} \), where \( x', r'\) are fresh names.
          \label{it:ax:glink-piI:subst-din}
    \item \( \res p \left(\glink p q \mid \oprfx p {x, r} \proc \right) \gexp \oprfx q {x, r} \proc \), provided that \( (\res {x, r})(\glink {x'} x \mid \glink r {r'} \mid \proc ) \gexp \proc \sub {x' , r'} {x, r} \),  where \( x', r'\) are fresh names.
          \label{it:ax:glink-piI:subst-dout}
    \item \( \res y (\glink x y \mid \riproc y p {\proc}) \gexp \riproc {x} {p} \proc \), provided that \( y \notin \fn \proc \) and \( \res p (\glink {p'} p \mid \proc ) \gexp \proc \sub {p'} p \), where \( p' \) is fresh.
          \label{it:ax:glink-piI:subst-rep}
    \item \( \res x (\glink x y \mid \oprfx x {p} \proc ) \gexp \oprfx y {p} \proc \), provided that \( x \notin \fn \proc \) and \( \res p (\glink p {p'} \mid \proc ) \gexp \proc \sub {p'} p \), where \( p' \) is fresh.
          \label{it:ax:glink-piI:subst-dout-var}
    \item \( \glink x y \) is a replicated input process at \( x \), i.e. \( \glink x y = \riproc x p \proc \) for some \( \proc \).
    \label{it:ax:glink-piI:var-link-rep}
  \end{enumerate}
\end{defi}
Condition~\ref{it:ax:glink-piI:fn} is a simple syntactic requirement.
Condition~\ref{it:ax:glink-piI:no-tau} says that wires are
`optimised' in that they cannot do any immediate internal interaction
(this requirement, while not mandatory, facilitates a few proofs).
 Law~\ref{it:ax:glink-piI:trans} is about the transitivity of wires.
Laws~\ref{it:ax:glink-piI:subst-din}-\ref{it:ax:glink-piI:subst-dout-var}
show that wires act as substitutions for permeable inputs, permeable
outputs and replicated input prefixes.
We do not require similar laws for ordinary prefixes, e.g., as in
\[
\res p \left(\glink p q \mid \bout p {x, r}. \proc \right)
\gexp \bout q {x, r}. \proc
\]
because wires break the strict sequentiality imposed by such prefixes
(essentially transforming an ordinary prefix into a permeable one:
only for the process on the right any action from
 $P$ is blocked until the environment accepts an interaction at $q$).
Condition~\ref{it:ax:glink-piI:var-link-rep} requires \( \glink x
y \) to be an input replicated processes, and is useful so to be able
to use the replication laws (Lemma~\ref{lem:rep-thm}).
Location wires, on the other hand,  are linear and hence should not be duplicated.
Moreover, we do not even require them to be an input process because we shall consider location wires with different control flows.

\emph{Hereafter we assume that \( \glink p q \) and \( \glink x y \)
are  indeed  wires, i.e., processes that
 satisfy the requirements of Definition~\ref{ax:glink-piI}}.
We can therefore exploit such requirements
to derive properties  of the abstract encoding.

Lemma~\ref{lem:glink-subst} shows that the processes encoding
functions are I-respectful with respect to the  location wires, and
 O-respectful with respect to the variable wires.

\begin{lem}
  \label{lem:glink-subst}
  \hfill
  \begin{enumerate}
    \item \( \res p (\glink qp \mid \Qencoabs \enco M p) \gexp \Qencoabs \enco M q \)
          \label{it:lem:glink-subst:cont}
    \item \( \res x (\glink x y \mid \Qencoabs \enco M p ) \gexp \Qencoabs \enco {M \sub y x} p \)
          \label{it:lem:glink-subst:var}
      \end{enumerate}
\end{lem}
\begin{proof}
  We prove~\ref{it:lem:glink-subst:cont} and~\ref{it:lem:glink-subst:var} simultaneously by induction on the structure of \( M \).

  \mylabel{Case \( M = x \)}
  We begin with the case of location names.
  \begin{align*}
    \res p (\enco x p \mid \glink q p)
    &= \res p (\oprfx x {p'} \glink {p} {p'} \mid \glink q p ) \\
    &\scong \oprfx x {p'} \res p (\glink p {p'} \mid \glink q p)  \\
    &\gexp \oprfx x {p'} \glink q {p'} \tag{\ref{it:ax:glink-piI:trans} of Definition~\ref{ax:glink-piI}} \\
    &= \enco{x} {q}
  \end{align*}

  Next we show that \( \res x (\enco M p \mid \glink x y) = \enco {M \sub y x} p \) holds when \( M \) is a   variable.
  First, we consider the case where \( M = z \neq x \).
  Since \( \glink x y \) is of the form \( \riproc x p \proc \) by~\ref{it:ax:glink-piI:var-link-rep} of Definition~\ref{ax:glink-piI}, it follows that \( \res x (\glink x y) \sbisim  \nil \), and this concludes this case.
  If \( M = x \), then
  \begin{align*}
    \res x (\enco x p \mid \glink x y)
    &=  \res x (\oprfx x {p'} \glink p {p'} \mid \glink x y) \\
    &\gexp \oprfx y {p'} \glink p {p'} \tag{\ref{it:ax:glink-piI:subst-dout-var} of Definition~\ref{ax:glink-piI}}  \\
    &= \enco y p
  \end{align*}
  The premise of~\ref{it:ax:glink-piI:subst-dout-var} of Definition~\ref{ax:glink-piI} is satisfied because of the transitivity of wires (\ref{it:ax:glink-piI:trans} of Definition~\ref{ax:glink-piI}).

  \mylabel{Case \( M = \lambda x . N \)}
  We first show that~\ref{it:lem:glink-subst:cont} holds.
  This is a direct consequence of~\ref{it:ax:glink-piI:subst-din} of Definition~\ref{ax:glink-piI}.
  \begin{align*}
    \res p (\enco {\lambda x. N} p \mid \glink q p)
    &= \res p (\diprfx p {x, r} \enco N r \mid \glink q p) \\
    &\gexp \diprfx q {x, r} \enco N r \tag{\ref{it:ax:glink-piI:subst-din} of Definition~\ref{ax:glink-piI} together with the i.h.} \\
    &= \enco {\lambda x. N} q
  \end{align*}
  The proof for~\ref{it:lem:glink-subst:var} is also a direct consequence of the induction hypothesis.
  That is,
  \begin{align*}
    \res y (\enco {\lambda x.N} p \mid \glink y z)
    &= \res y (\diprfx p {x, q} \enco {N} q \mid \glink y z) \\
    &\scong \diprfx p {x, q} \res y (\enco {N} q \mid \glink y z) \\
    &\gexp \diprfx p {x, q} (\enco {N \sub z y} q ) \tag{i.h.} \\
    &= \enco {(\lambda x.N) \sub z y} p
  \end{align*}
  Here we assumed that \( x \), \( y \) and \( z \) are pairwise distinct; the general case can be proved using \( \alpha \)-conversion.

  \mylabel{Case \( M = \app N L \)}
  The case for location names follows from the transitivity of wires.
  Observe that
  \begin{align*}
    &\res p \left(\oprfx r {x, p'} (\riproc x {r'} {\enco L {r'}} \mid  \glink p {p'}) \mid \glink q p \right) \\
    &\scong (\res {p, x, p'})
        (\bout r {x', p''}. (\glink {x'} x \mid \glink {p'} {p''}) \mid \riproc x {r'}
      {\enco L {r'}} \mid  \glink p {p'} \mid \glink q p )
    \tag{definition of permeable prefix}
  \\
    &\gexp (\res {x, p'})(\bout r {x', p''}. (\glink {x'} x \mid \glink {p'} {p''}) \mid \riproc x {r'} {\enco L {r'}} \mid  \glink q {p'})
      \tag{\ref{it:ax:glink-piI:trans} of Definition~\ref{ax:glink-piI}} \\
    &= \oprfx r {x, p'} (\riproc x {r'} {\enco L {r'}} \mid  \glink q {p'}) \tag{definition of permeable prefix}
  \end{align*}
  Using this, we get
  \begin{align*}
    \res p \left( \enco {\app N L} p \mid \glink qp \right)
    &\scong  \res r \left(\enco N r \mid \res p \left(\oprfx r {x, p'} (\riproc x {r'} {\enco L {r'}} \mid  \glink p {p'}) \mid \glink q p \right) \right) \\
    &\gexp  \res r \left(\enco N r \mid \oprfx r {x, p'} (\riproc x {r'} {\enco L {r'}} \mid  \glink q {p'}) \right) \\
    &= \enco {\app N L} q
  \end{align*}
  The proof for \( \glink x y \) follows from the replication theorem and the induction hypothesis.
  First, observe that \( \glink x y \) must be of the form \( \riproc x p \proc \) because of~\ref{it:ax:glink-piI:var-link-rep} of Definition~\ref{ax:glink-piI}.
  We, therefore, can apply the replication theorem to \( \glink x y \).
  Hence, we have
  \begin{align*}
    \res x (\riproc z p  \enco L p \mid \glink x y)
    &\sbisim \riproc z p {\res x(\enco L p \mid \glink x y)} \tag{replication theorem}  \\
    &\gexp  \riproc z p {\enco {L \sub y x} p}  \tag{i.h.}
  \end{align*}
  The claim follows by applying this expansion relation, the replication theorem for parallel composition and the induction hypothesis:
  \begin{align*}
    &\res x (\enco{\app N L} p  \mid \glink x y) \\
    &\scong (\res {x, r}) \left(\enco N r \mid \oprfx r {z, p'} (\riproc z {r'} {\enco L {r'}} \mid  \glink p {p'}) \mid \glink x y \right) \\
    &\sbisim \res r
      \begin{aligned}[t]
      (&\res x \left(\enco N r \mid \glink x y \right) \\
      &\mid \oprfx r {z, p'} (\res x \left(\riproc z {r'} {\enco L {r'}} \mid \glink x y \right) \mid  \glink p {p'}))
      \end{aligned} \tag{replication theorem} \\
    &\gexp \res r \left(\enco {N \sub y x} r \mid \oprfx r {z, p'} (\riproc z {r'} {\enco {L \sub y x} {r'}} \mid  \glink p {p'}) \right) \tag{i.h. and the above expansion relation} \\
    &= \enco {\app {N \sub y x} {L \sub y x}} p \\
    &= \enco {(\app N L) \sub y x} p. \tag*{\qedhere}
  \end{align*}
\end{proof}

\subsection{Validity of  \texorpdfstring{\( \beta \)}{beta}-reduction}
The abstract encoding \( \Qencoabs \) validates \( \beta \)-reduction with respect to the expansion relation.
This is proved by showing that substitution of a \( \lambda \)-term \( M \) is implemented as a communication to a replicated server that owns  \( M \).
The proof is similar to that for Milner's encoding, except that we exploit Lemma~\ref{lem:glink-subst}.
We recall  that \( M \red N \) is the \emph{full} \( \beta \)-reduction.

\begin{restatable}[Substitution]{lem}{gencSubst}
  \label{lem:genc-subst}
  If \( x \notin \fv N \), then \( \res x ( \Qencoabs \enco M p \mid \riproc x q {\Qencoabs \enco N q} ) \gexp \Qencoabs \enco {M \sub N x} p \).
\end{restatable}
\begin{proof}
  By induction on the structure of \( M \) using the replication theorem.

  For the base case, we consider the case \( M = x \); if \( x \notin \fv M \), then we just need to apply the garbage collection law.
  \begin{align*}
    \res x ( \enco x p \mid  \riproc x q {\enco {N} q} )
    &= \res x ( \oprfx x {p'} \glink p {p'}  \mid \riproc x {q} {\enco {N} q} ) \\
    &\gexp \res {p'} (\glink p {p'} \mid \encoN {N} {p'}) \tag{Lemma~\ref{lem:permeable-comm}, and garbage collection on \( x \)} \\
    &\gexp \encoN {N} p \tag{Lemma~\ref{lem:glink-subst}}.
  \end{align*}

  The case  \( M = \lambda y . M' \) is a straightforward consequence of the induction hypothesis:
  \begin{align*}
    \res x ( \enco {\lambda y. M'} p \mid  \riproc x q {\enco N q})
    &= \res x ( \diprfx p {y, r} {\enco {M'} r}  \mid \riproc x {q} {\enco N q})\\
    &\equiv \diprfx p {y, r} \res x ({\enco {M'} r}  \mid \riproc x {q} {\enco N q}) \tag{Lemma~\ref{lem:permeable-scong}} \\
    &\gexp \diprfx p {y, r} \enco {M' \sub N x} r \tag{i.h.} \\
    &= \enco {(\lambda y. M') \sub N x}p.
  \end{align*}

  It is the case where \( M = \app {M_1} {M_2} \) that needs the replication theorem.
  If \( M = \app {M_1} {M_2}\), we have
  \begin{align*}
    &\res x \left( \enco {\app {M_1} {M_2}} p \mid \riproc x q {\enco N q} \right) \\
    &\scong \resb {x,q}
        (\enco {M_1} q \mid \oprfx q {y, p'} \left( \riproc y r {\enco {M_2} r} \mid \glink p {p'}\right) \mid \riproc x q {\enco N q}) \\
    &\sbisim  \res q
      \begin{aligned}[t]
        (&\res x \left(\enco {M_1} q \mid \riproc x q {\enco N q}  \right) \mid\\
        &\oprfx q {y, p'} \left( \res x \left(\riproc y r {\enco {M_2} r} \mid \riproc x q {\enco N q} \right) \mid \glink p {p'}\right))
      \end{aligned}  \tag{replication theorem for parallel composition} \\
    &\sbisim  \res q
      \begin{aligned}[t]
        (&\res x \left(\enco {M_1} q \mid \riproc x q {\enco N q}  \right) \mid \\
        &\oprfx q {y, p'} \left( \riproc y r {\res x \left(\enco {M_2} r \mid \riproc x q {\enco N q} \right)} \mid \glink p {p'}\right))
      \end{aligned}  \tag{replication theorem for replicated input}  \\
    &\gexp \res q \left(\enco {M_1 \sub N x} q \mid \oprfx q {y, p'} \left( \riproc y r {\enco {M_2 \sub N x} r} \mid \glink p {p'}\right) \right) \tag{i.h.} \\
    &= \enco {(\app {M_1} {M_2}) \sub N x} p \tag*{\qedhere}
  \end{align*}
\end{proof}

\begin{thm}
  \label{thm:genc-validates-beta}
  If \( M \red N \), then \( \Qencoabs \enco M p \gexp \Qencoabs \enco N p \).
\end{thm}
\begin{proof}
  It suffices to show that \( \enco {\app {(\lambda x. M)} N} p \gexp \enco {M \sub N x} p \) because the other cases follow from the precongruence of \( \gexp \).
  We have
  \begin{align*}
    \enco {\app {(\lambda x. M)} N} p
    &= \res q \left(\diprfx q {x, r} \enco M r \mid \oprfx q {y, p'} \left( \riproc y {r'} {\enco N {r'}}  \mid \glink p {p'} \right) \right) \qquad (\text{\( q \) is fresh}) \\
    &\gexp  (\res {y, p'}) \left(\enco {M \sub y x} {p'} \mid \riproc y {r'} {\enco N {r'}}  \mid \glink p {p'} \right) \tag{Lemma~\ref{lem:permeable-comm} and~\ref{lem:glink-subst}} \\
    &\gexp \res y \left(\enco {M \sub y x} p \mid \riproc y {r'} {\enco N {r'}} \right) \tag{Lemma~\ref{lem:glink-subst}} \\
    &\gexp \enco {M \sub N x} p \tag{Lemma~\ref{lem:genc-subst}}
  \end{align*}
  as desired.
\end{proof}

Since  bisimilarity  is a congruence in \piI{} and our encoding is compositional, the validity of \( \beta \)-reduction implies that the equivalence induced by the encoding is a \( \lambda \)-theory.
\begin{cor}
\label{cor:lambda-theory}
  Let \( {=_\pi} \defeq \{ (M, N) \mid \Qencoabs \enco M {} \bisim \Qencoabs \enco N {} \} \).
  Then \( =_\pi \) is a \( \lambda \)-theory, that is, a congruence on \( \lambda \)-terms that contains \( \beta \)-equivalence.
\end{cor}

\begin{rem}
From a \( \lambda \)-theory,  a
$\lambda$-model can be extracted~\cite{Barendregt84}, hence Corollary~\ref{cor:lambda-theory} implies that we can
 construct a
$\lambda$-model out of the process terms.
The domain of the model would be
 the processes that are in the image of the encoding, quotiented with bisimilarity.
We could not define  the domain of the model  out  of all process
terms (as in~\cite{Sangiorgi00}, as opposed to the processes in the image of the encoding)
 because our proofs rely on Lemma~\ref{lem:glink-subst}, and such a
 lemma cannot be extended to the set of all processes.
\end{rem}

%% file: oenc.tex
\section{Optimised Encoding}
\label{sec:oenc}
We introduce an optimised version of the
abstract encoding, which 
removes certain `administrative steps' on the process terms.
This will allow us to have a sharper operational correspondence
between  $\lambda$-terms and processes, which  
 will be needed in  proofs in later sections.
As in the previous section,
we work with abstract wires, only assuming 
the requirements in Definition~\ref{ax:glink-piI}.

To motivate the need of the optimised encoding, let us consider the encoding of a term \( \app {(\app x M)} N \):
\[
 \resb {p_0, p_1}(
\begin{aligned}[t]
 &\oprfx x {p_0'} \glink {p_0'} {p_0}  \\
&\mid\oprfx {p_0} {x_1. p_1'}
 \left( \riproc{x_1}{r_1} {\enco M} {r_1} \mid \glink {p_1} {p_1'}\right)  \\
 &\mid \oprfx {p_1}{x_2, p_2}\left(\riproc{x_2}{r_2}
   {\enco N} {r_2}  \mid \glink p {p_2}\right))
\end{aligned}
\]
This process has, potentially (i.e., depending on the concrete
instantiations of the wires),  some initial administrative reductions.
For instance, the output at \( p_1 \) may interact with the input end
of the wire   \( \glink {p_1} {p_1'} \).
\begin{figure*}[!t]
  \begin{align*}
  \oenco x p &\defeq \oprfx x {p'} \glink p {p'} \\
  \oenco {\lambda x . M} p &\defeq  \diprfx p {x, q} \oenco M q \\
  \oenco {\app x {M_1 \cdots M_n}} p &\defeq \oprfx x {p_0}  \oencoN {p_0} p {\oenco {M_1} {} \cdots \oenco {M_n}{}} \\
  \oenco {\app {(\lambda x. M_0)} {M_1 \cdots M_n}} p &\defeq \res {p_0} \left(\diprfx {p_0}{x, q} \oenco {M_0} q \mid  \oencoN {p_0} p {\oenco {M_1} {} \cdots \oenco {M_n}{}} \right) \\
  \oencoN {p_0} p {\oenco {M_1} {} \cdots \oenco {M_n}{}}  &\defeq
\begin{array}[t]{l}
 \oprfx {p_0}{x_1,  p_1} \cdots \oprfx {p_{n - 1}}{x_n, p_n} 
\\
\left( \riproc {x_1} {r_1} {\oenco {M_1} {r_1}} \mid \cdots \mid \riproc {x_n} {r_n} {\oenco {M_n} {r_n}}  \mid \glink p {p_n} \right)
\end{array}
   \end{align*}
\caption{Optimised encoding. (The number \( n \) is greater than \( 0 \) in the last three cases.)}
\label{fig:oenc}
\end{figure*}

In the optimised encoding  \( \Qoenco \),   in
Figure~\ref{fig:oenc},  
 any initial reduction of a process
has a  direct
 correspondence with a (strong call-by-name) reduction of the source
 $\lambda$-term.
With respect to
the unoptimised encoding $\QencoA$,
the novelties are in the clauses for
application, where the case of a head normal form 
$\app x {M_1 \cdots M_n}$ (for $n\geq 1$) and of an application 
$\app {(\lambda x. M_0)} {M_1 \cdots M_n}$
with a head redex are distinguished.
In both cases, 
$  \oencoN {p_0} p {\oenco {M_1} {} \cdots \oenco {M_n}{}}$ is used
for a
compact representation of the encoding of the trailing arguments
$M_1,\ldots,M_n$, as a sequence of nested permeable prefixes and a
bunch of replications embracing the terms $M_i$.

Analogous properties to those in Section~\ref{sec:abs-enc} for the unoptimised encoding $\QencoA$  hold for \( \Qoenco \).
For instance, \( \Qoenco \) validates \( \beta \)-reduction (Lemma~\ref{lem:gen-oenc-validates-beta}).
Using such properties, and reasoning by induction of the structure of
a $\lambda$-term, we
can  prove that \( \Qoenco \) is indeed an optimisation.
\begin{restatable}{lem}{genOencValidatesBeta}
  \label{lem:gen-oenc-validates-beta}
  If \( M \red N \), then \( \oenco M p \gexp \oenco N p \).
\end{restatable}

\begin{restatable}{lem}{genOencIsOptimisation}
  \label{lem:gen-oenc-is-optimisation}
  \( \Qencoabs \enco M p \gexp \oenco M p\).
\end{restatable}
The details about the operational behaviour of
\( \oenco M p
\), and its operational correspondence with \( M \), are described in Appendix~\ref{sec:oenc-appx}.
We only report here the statements of a few important lemmas.

\begin{restatable}{lem}{genOencTauHasMatchingRed}
  \label{lem:gen-oenc-tau-has-mathcing-red}
  If \( \oenco M p \trans \tau \proc \) then there exists \( N \) such
  that \( M \snred N \) and \( \proc \gexp \oenco N p \).
\end{restatable}

\begin{restatable}{lem}{genOencInputMustBeAtp}
  \label{lem:gen-oenc-input-must-be-at-p}
  If \( \oenco M p \trans{\act} \proc \) and \( \act \) is an input action, then \( \act \) is an input at \( p \).
\end{restatable}

If \( M = \lambda x. M' \), then its process encoding \( \oenco {M} p \) can always do an input at \( p \),
as the encoding of a abstraction begins with an input  at its location name.
Lemma~\ref{lem:gen-oenc-input-must-be-at-p} says that such an input at $p$ is indeed the only possible input;
that is,  we cannot observe an inner \( \lambda \)-abstraction (i.e., a $\lambda$-abstraction in \( M' \)).

Later, when we relate our encoding to trees of the  \( \lambda \)-calculus,
the notions of head normal form and (un)solvable term  will be important. 
Hence some of our operational correspondence results concern them.
\begin{restatable}{lem}{genOencOutputImpliesSolvable}
\label{lem:gen-oenc-output-implies-solvable}
  Let \( M \) be a \( \lambda \)-term.
  If \( \oenco M p \trans {\bout x q} \proc \) for some \( \proc \),
  then \( M \)  has a head normal form
\( \lambda \seq y . \app x \seq M \), for some (possibly
  empty) sequence of terms \( \seq M \) and variables \( \seq y \)
  with
  \( x \notin \seq y \).
\end{restatable}
\begin{restatable}{lem}{genOencUnsolvableNoOutput}
  \label{lem:gen-oenc-unsolvable-no-output}
  Let \( M \) be an unsolvable term.
  Then there does not exist an output action \( \act \) such that \(
  \oenco M p \wtrans \act \proc \) for some \( \proc \).
\end{restatable}

\begin{cor}
\label{c:solv}
If $M$ is solvable then there are input actions $\mu_1, \ldots,
\mu_n$ ($n \geq 0$) and an output action $\mu$ such that 
 \(
  \oenco M p \wtrans {\mu_1} \ldots \wtrans {\mu_n}\wtrans {\mu} P
  \), for some $P$.
\end{cor}

By Lemma~\ref{lem:gen-oenc-is-optimisation}, Corollary~\ref{c:solv} also holds for $\QencoA$. 
The converse of Corollary~\ref{c:solv} will also hold, in all three
concrete encodings that will be studied in the next section. 
However we believe  the result cannot be derived from the assumptions on
wires in Definition~\ref{ax:glink-piI}.

We conclude  by looking, as an example, at
 the unsolvable term \( \Omega \defeq \app {(\lambda x. \app x x)} \app {(\lambda x. \app x x)} \).
\begin{exa}
\label{ex:oenc-Omega}
  The process \( \oenco \Omega p \) is
  \[
    \res {p_0} ( \begin{aligned}[t]
      &\diprfx {p_0} {x, q} \oprfx x {q_0} \oprfx {q_0}
      {y_1, q_1}  \left( \riproc {y_1} {r_1} {\oenco x {r_1}} \mid
        \glink q {q_1} \right)\\
 &\mid \oprfx {p_0}{x_1, p_1} \left(
        \riproc {x_1} {r_1} {\oenco {\lambda x. \app x x} {r_1}} \mid
        \glink p {p_1} \right) )
\end{aligned}
   \]
  The only action \( \oenco \Omega p \)
 can do is a \( \tau \)-action or an input at \( p \).
  Whether the input can be performed or not
 will depend on the concrete definition of the wire \( \glink p q \).
The possibility of an 
 input action from an unsolvable of order $0$ such as $\Omega$
 is a major difference between our encoding and encodings in
 the literature, where
 the encoding of such unsolvables are usually purely divergent
 processes.
\end{exa}

%% file: wires.tex
\section{Concrete Wires}
\label{sec:wires}
We now examine concrete instantiations of the abstract wires in the
encoding $\QencoA$ (and its optimisation \( \Qoenco \)).
In each case we have to define the wires for  location and variable names.
 The location wires are the important ones: the definition of the variable wires will follow from them, with the goal of guaranteeing their expected properties.
 We consider three concrete  wires: \emph{\IOwires{}},
\emph{\OIwires{}}, and \emph{\Pwires{}}.
The main difference among them is in the order in which
the input and output of the location wires are performed.

\subsection{Definition of the concrete wires}
Location and variable wires will be defined by means of recursion.
In contrast with the variable wires, the location wires are non-replicated processes, reflecting the linear  usage of
such names.
We recall that the choice of a  certain kind of concrete
 wires  (\IOwires, \OIwires, or \Pwires) also
affects the definition of the permeable prefixes (as it refers to the wires), including   the permeable
prefixes that may be used within the wires themselves.
We will also show de-sugared definitions of the concrete wires, i.e., without reference to permeable prefixes.
We add a subscript ($\mathtt{IO}$, $\mathtt{OI}$,
$\mathtt{P}$) to indicate a concrete wire (as opposed to an abstract one).
For readability, in the definitions of the concrete wires 
the name parameters are instantiated (e.g., writing 
$ \link ab  \defeq  P$ rather than
$ \link*  \defeq \abs{a,b} P$).

\subsubsection*{\IOwires{}}
In the \IOwires{}, the input of a wire precedes the output.
 \begin{align*}
   \link p q &\defeq \bin p {y, p_1}. \oprfx q {x, q_1} ( \link {p_1} {q_1} \mid \link x y) \\
   \link x y &\defeq \riproc x {p} {\oprfx y {q}  \link {p} {q}}
 \end{align*}
Inlining the abbreviations for permeable prefixes (as they are, 
 in turn, defined using wires, in this specific case,
 the I-O wires),  
we obtain: 
 \begin{align*}
   \link p q &\defeq \bin p {y, p_1}.
               \begin{aligned}[t]
                 \resb{x, q_1}(&\bout q {x', q_1'}.(\link{q_1}{q_1'} \mid \link {x'}{x})  \\
                 &\mid \link {p_1} {q_1} \mid \link x y)
                 \end{aligned} \\
   \link x y &\defeq \riproc x {p} {\res q (\bout y {q'}.\link q
   {q'}  \mid \link {p} {q} )}
 \end{align*}
\IOwires{}, beginning with an input and proceeding with an output,  are similar to the
ordinary wires  in the literature, 
sometimes called  \emph{forwarders}, and used to prove properties
about asynchronous and localised $\pi$-calculi (or encodings of
them)~\cite{HoYo95,Mer00thesis,MerroSangiorgi04,Boreale98}.
An important technical difference, within location wires, is the appearance
 of a permeable prefix, in place of an ordinary prefix,
and the inner wire
$ \link {p_1} {q_1}$ that, in a forwarder, would have 
$p_1$ and $q_1$ swapped.
 The  reason for these  differences is that location wires are
used with processes that are not localised (the recipient of a location name will use it
in input, rather than in output). 
The difference also shows up in the 
semantic properties: 
forwarders in the literature are normally
used to obtain properties of O-respectfulness (Section~\ref{sec:EpiI}), with the input-end of
the wire  restricted;
in contrast, \IOwires{} will be used to obtain properties of I-respectfulness,
 with the output-end of
the wire  restricted.

In the definition above of location wires, the permeable prefix cannot be replaced by an
ordinary prefix: the transitivity  of the wires (property~\ref{it:ax:glink-piI:trans} in Definitions~\ref{ax:glink-piI})  would  be lost.

\subsubsection*{\OIwires{}}
The symmetry of \piI{} enable us to consider the dual form of (location) wire,
 with the opposite control flow, namely `from output to input':  
  \begin{align*}
    \linkO qp  &\defeq 
 {\bout q {x, q_1}. \diprfx p {y, p_1}} ( { \linkO {q_1} {p_1}} \mid {
\vlinkOI xy}) \\
    \vlinkOI x y &\defeq  \riproc x p {\oprfx y q \linkO q p}
 \end{align*}

\begin{rem}[Duality]
\label{rem:OI-wires-prefix}
If duality is taken to
mean the exchange  between input and output prefixes,
then  the set of location
\IOwires  is  the dual of the set of \OIwires. 
Indeed, the location \OIwires{} are obtained from the corresponding 
location \IOwires{} by swapping     input and   output particles; variable
wires, in contrast are left unchanged. 
This means that we obtain an \OIwire from an  \IOwire{} if
 any input $p(\tilb)$ is made into an
output $\outC p {(\tilb)}$, and conversely (moreover, accordingly, the
parameters 
of the variable wires are swapped).
\end{rem}

\subsubsection*{\Pwires{}}
In the third form of wire, the sequentiality in location wires is broken: input and output execute concurrently.
This is achieved by using, in  the definition of location wires, only  permeable prefixes.
\begin{align*}
  \ilink p q &\defeq \diprfx p {y, p_1} \oprfx q {x, q_1} ( \ilink {p_1} {q_1} \mid \olink x y) \\
  \olink x y &\defeq \riproc x p {\oprfx y q \ilink p q}
\end{align*}
Without the syntactic sugar of permeable prefixes, the  definition of the location and
variable \Pwires{}  are thus:
\begin{align*}
  \ilink p q &\defeq \resb {p_1,q_1 x, y}
               \begin{aligned}[t]
                 (&\iproc p {y', p_1'} {(\ilink {p_1'} {p_1} \mid \olink y {y'})} \mid \\
                 &\bout q {q_1', x'}.(\ilink {q_1} {q_1'} \mid  \olink {x'} x) \mid \\
                 &\ilink {p_1} {q_1} \mid \olink x y)
              \end{aligned} \\
  \olink x y &\defeq \riproc x {p} {\res q (\bout y {q'}. \ilink q {q'}  \mid  \ilink p q)}
\end{align*}

The wire \( \ilink p q \) is dual of itself: due to the use of
permeable prefixes, 
swapping input and output prefixes has no behavioural effect.

\subsection{Transitivity and well-definedness of the concrete wires}
All three wires introduced in the previous subsection are well defined:
\begin{lem}
\label{lem:all-wires-satisfy-ax}
The \IOwires{}, \OIwires{} and \Pwires{} are indeed wires; that is, they
satisfy the laws of Definition~\ref{ax:glink-piI}.
\end{lem}

For each kind of wires, the proof is carried out in two steps.
First, we show that the wires are transitive, using up-to techniques
for bisimilarity.
Then, the other laws are proved by algebraic reasoning (including the use of transitivity of wires).
The algebraic reasoning is done in a similar manner for all the three wires.

Here, \emph{we shall only give the proof of the transitivity of \Pwires{}} since it is the most delicate one, because of
the concurrency allowed by permeable prefixes and because permeable prefixes are defined in terms of the wires themselves.
The proofs of the transitivity of the other two wires, and the proofs for the other laws are reported in Appendix~\ref{sec:wires-appx}.

To illustrate the difficulty of \Pwires{}, let us consider the wires \( \res q (\ilink p q \mid \ilink q r) \).
This process can immediately reduce at the internal name \( q \).
Moreover, the derivative
\[
  \diprfx p {p_1, y} \oprfx {r} {q_1, x}
  \begin{aligned}[t]
    (&\resb{z_1,z_2}(\olink x {z_1} \mid \olink {z_1} {z_2}  \mid \olink {z_2} y) \mid \\
    &\resb{s_1, s_2}(\ilink {p_1}{s_1} \mid \ilink {s_1} {s_2} \mid \ilink {s_2} {q_1}))
\end{aligned}
\]
shows that the reduction has made the chain of wires longer.
(Further reductions are then possible; indeed infinitely-many reductions may
  thus be produced.)
To deal with these cases, we strengthen the claim and show the transitivity of chains of wires.
In doing so, we crucially rely on up-to  proof techniques for \piI{}, notably `expansion up-to expansion and context', and several algebraic laws (cf.~Section~\ref{sec:proof-techniques}).
The `up-to context' is used to cut out common contexts such as \( \diprfx p {p_1, y} \oprfx r {q_1, x}(\hole \mid \hole ) \).
The algebraic laws are used to transform the processes so to be able to  apply the `up-to context', mainly by performing  internal
interactions.
It is unclear how the proof could be carried out without such  proof techniques.

As we need to consider chains of \Pwires{} of arbitrary length, we introduce a
notation for them.
\newcommand*{\chain}[4]{\mathrm{chain}_{#1}^{#2}(#3, #4)}
We thus set:
  \begin{align*}
    \chain {\mathtt {P}} 1 p q &\defeq \link p q &\chain {\mathtt{P}} {n + 1} p q &\defeq \res r (\chain {\mathtt{P}} n p r \mid \link r q) \\
    \chain {\mathtt{P}} 1 x y &\defeq \link x y & \chain {\mathtt{P}} {n + 1} x y &\defeq \res z (\chain {\mathtt{P}} n x z \mid \link z y)
  \end{align*}
where \( r \notin \{p, q\} \) and \( z \notin \{ x, y\} \).

Now we are ready to show the transitivity of \Pwires{}.
\begin{lem}
  \label{lem:ilink-olink-trans}
  The \Pwires{} \( \ilink p q \) and \( \olink x y \) are transitive, that is,
 \( \res  q (\ilink p q \mid \ilink q r) \gexp \ilink p r \) and \( \res y (\olink x y \mid \olink y z) \gexp \olink x z \).
\end{lem}
\begin{proof}
  \newcommand*{\ichain}[3]{\mathrm{chain}_{\mathtt{P}}^{#1}(#2, #3)}
  \newcommand*{\ochain}[3]{\mathrm{chain}_{\mathtt{P}}^{#1}(#2, #3)}
For the proof,
we strengthen the statement and prove transitivity for chains of wires of arbitrary length \( n \).
  Let
  \begin{align*}
    {\relR_1} &\defeq \left\{ \left(\ilink {p_0} {p_n}, \ichain n {p_0} {p_n} \right) \; \midbar \; n \ge 2 \right\} \\
    {\relR_2} &\defeq \left\{ \left(\olink {x_0} {x_n}, \ochain n {x_0} {x_n} \right) \; \midbar \; n \ge 2 \right\}
  \end{align*}
  We show that \( {\relR_1} \cup {\relR_2}\) is an expansion up-to \( \lexp \) and context.

  Before considering how processes in the relation can match each other's transition,
we present some useful observations that will be used throughout the proof.
  Recall that \( \ilink {p_i} {p_{i+1}} \) is of the form
  \begin{align*}
    &(\res {x_i^+, q_i^+, x_{i + 1}^-, q_{i + 1}^-})\\
    (&\iproc {p_i}{x_i, q_i} {(\olink {x_i^+} {x_i} \mid \ilink {q_i} {q_i^+} )} \\
    &\mid \bout {p_{i+1}} {x_{i+1}, q_{i+1}}.(\olink {x_{i+1}} {x_{i+1}^-} \mid \ilink{q_{i+1}^-}{q_{i+1}} ) \\
    &\mid \olink {x_{i + 1}^-} {x_i^+} \mid \ilink {q_i^+} {q_{i+1}^-})
  \end{align*}
  Therefore, given
  \[
    \ichain n {p_0} {p_n} \equiv
    \resb {p_1, \ldots, p_{n-1}} (\ilink {p_0} {p_1} \mid \cdots \mid \ilink {p_{n - 1}}{p_n}) \, ,
  \]
 reducing the process by executing all the (internal) interactions at the \( p_i \)'s,
 gives us a process of the form
  \begin{align*}
    &(\res {x_0^+, x_1^-, x_1 ,x_1^+, \ldots, x_{n-1}^-, x_{n-1} ,x_{n-1}^+, x_n^-}) \\
    &(\res {q_0^+, q_1^-, q_1 ,q_1^+, \ldots, q_{n - 1}^-, q_{n -1} ,q_{n - 1}^+, q_n^-}) \\
    &(\iproc {p_0}{x_0, q_0} {(\olink {x_0^+} {x_0} \mid \ilink {q_0} {q_0^+} )} \\
    &\mid \bout {p_n} {x_n, q_n}.(\olink {x_n} {x_n^-} \mid \ilink{q_n^-}{q_n} )\\
    &\mid \ilink{q_0^+}{q_1^-} \mid \ilink {q_1^-} {q_1} \mid \ilink {q_1}{q_1^+} \mid \cdots \\
    &\mid \ilink {q_{n - 1}^-} {q_{n - 1}} \mid \ilink {q_{n - 1}}{q_{n - 1}^+} \mid \ilink {q_{n - 1}^+} {q_n^-}\\
    &\mid \olink {x_n^-} {x_{n - 1}^+} \mid \olink {x_{n - 1}^+} {x_{n - 1}} \mid \olink {x_{n - 1}} {x_{n - 1}^-} \mid \cdots \\
    &\mid \olink {x_1^+} {x_1} \mid \olink {x_1} {x_1^-} \mid \olink {x_1^-} {x_0^+})
  \end{align*}
  Up to structural congruence, the above process can  be written as
  \begin{align*}
    &(\res {x_0^+, x_n^-}) (\res {q_0^+,  q_n^-}) \\
    &(\iproc {p_0}{x_0, q_0} {(\olink {x_0^+} {x_0} \mid \ilink {q_0} {q_0^+} )} \\
    &\mid \bout {p_n} {x_n, q_n}.(\olink {x_n} {x_n^-} \mid \ilink{q_n^-}{q_n} ) \\
    &\mid \ichain {3n + 1} {q_0^+} {q_{n + 1}^-} \mid \ochain {3n + 1} {x_{n + 1}^-} {x_0^+} ) \\
    &\scong \diprfx {p_0} {x_0, q_0} \oprfx {p_{n+1}}{x_{n+1}, q_{n+1}} (\ochain {3n - 2}{x_n}{x_0} \mid \ichain {3n - 2} {q_0} {q_n})
  \end{align*}

  Moreover, since the reductions performed  are all at restricted linear names,
  in each reduction the initial process is in the relation $\gexp$ with the derivative
  process; that is,
  for any \( n \ge 2\), we have
  \begin{align}
    \label{e:ichain}
    \ichain n {p_0} {p_n} &\gexp  \diprfx {p_0} {x_0, q_0} \oprfx {p_n}{x_n, q_n} (\ochain {3n - 2}{x_n}{x_0} \mid \ichain {3n - 2} {q_0} {q_n})
  \end{align}
  Exploiting this property,  we can  prove that \( {\relR_1} \cup {\relR_2}\) is an expansion up-to expansion and context.
  We first consider the case where \( \ilink{p_0} {p_n} \relR_1 \ichain n {p_0} {p_n} \).
  We only consider the case where the process on the right-hand side makes the challenge; the opposite direction can be proved similarly.
  There are three possible actions that the process on the right-hand side can make: (1) \( \tau \)-action, (2) input at \( p_0 \), and (3) output at \( p_{n + 1} \).
  We start by proving the first case.
  If
  \[ (\res {p_1, \ldots, p_{n - 1}}) (\ilink {p_0} {p_1} \mid \cdots \mid \ilink {p_{n - 1}}{p_n} ) \trans \tau \proc \, , \]
  then the action must have been caused by an interaction at \( p_i \), for some \( i \) such that \( 1 \le i \le n - 1 \).
  We can execute the interactions at the remaining \( p_i \)'s,
  and then, using the property (\ref{e:ichain}) above, we have
  \begin{align*}
      \proc \gexp \diprfx {p_0} {x_0, q_0} \oprfx {p_{n+1}}{x_n, q_n} (\ochain {3n - 2}{x_n}{x_0} \mid \ichain {3n - 2} {q_0} {q_n}).
  \end{align*}
  For the matching transition we take the \( 0 \)-step transition, i.e.~the identity relation.
  Since
    \begin{align*}
      &\ilink {p_0} {p_n} =  \diprfx {p_0} {x_0, q_0} \oprfx {p_n} {x_{n}, q_{n}} (\olink
        {x_n} {x_0} \mid \ilink {q_n}{q_0}) \;,  \\
      & \olink {x_n} {x_0} \relR_2 \ochain {3n - 2}{x_n}{x_0} \;, \\
      & \ilink {q_0} {q_n} \relR_1 \ichain {3n - 2} {q_0} {q_n}
    \end{align*}
    we can conclude this case using the up-to expansion and context technique.
    Similarly, if
    \[
      \ichain n {p_0}{p_n} \trans {p_0(x_0, q_0)} \proc,
    \]
    then we can show that
    \[
      \proc \gexp (\res {x_0^+, q_0^+}) (\olink {x_0^+} {x_0} \mid \ilink {q_0} {q_0^+} \mid  \oprfx {p_n}{x_n, q_n} (\ochain {3n - 2}{x_n}{x_0^+}  \mid \ichain {3n - 2} {q_0^+} {q_n})).
    \]
    We can match this transition with
    \begin{align*}
      \ilink {p_0} {p_n} \trans {p_0(x_0, q_0)} (\res {x_0^+, q_0^+}) (\olink {x_0^+} {x_0} \mid \ilink {q_0} {q_0^+} \mid \oprfx {p_n}{x_n, q_n} (\olink {x_n} {x_0^+} \mid \ilink {q_0^+} {q_n})).
    \end{align*}
    Once again, since
     \begin{gather*}
      \olink {x_n} {x_0} \relR_2 \;, \\
      \ilink {q_0} {q_n} \relR_1  \ichain {3n - 2} {q_0} {q_n}
    \end{gather*}
    we can appeal to the up-to expansion and context technique to finish the case.

    The remaining case (the case where process makes an output at \( p_n \)) can be proved by the same reasoning.

    Now we sketch the case for variable names.
    Take
    \[
     \olink {x_0} {x_n} \relR_2  \ochain n {x_0} {x_n}.
    \]
    There is only one possible action that the process on the right-hand side can make: an input at \( x_0 \).
    By a reasoning similar to that of the location wires, we can show that
    if
    \[
      \ochain n {x_0} {x_n} \trans {x_0(p_0)} \proc
    \]
    then
    \[
      \proc \gexp
      \begin{aligned}[t]
        & \ochain n {x_0} {x_n} \mid  \oprfx {x_n}{p_n} \ichain {2n -1}{p_0}{p_n}
      \end{aligned}
    \]
    by executing the interactions at the \( x_i \)'s.
    The above transition can be matched by the transition \( \olink {x_0} {x_n} \trans{x_0(p_0)} \olink {x_0} {x_n} \mid \oprfx{x_n} {p_n} \ilink {p_0}{p_n} \), and we can conclude by using the up-to context technique.
\end{proof}

\begin{rem}
The duality between \IOwires{} and \OIwires{} also shows up in proofs.
For instance,  for \IOwires{} the proof of law~\ref{it:ax:glink-piI:subst-din} of Definitions~\ref{ax:glink-piI} does not use the
premise of the law (i.e., the respectfulness of \( \proc \)), whereas the proof of the dual law~\ref{it:ax:glink-piI:subst-dout} does.
In the case of \OIwires{}, the opposite happens: the proof of law~\ref{it:ax:glink-piI:subst-dout}  uses the premise,
whereas  that of law~\ref{it:ax:glink-piI:subst-din} does not.
\end{rem}

In  the following sections we examine the concrete encodings obtained
by instantiating the wires of the abstract
 encoding \( \Qencoabs \) (Figure~\ref{fig:abs-enc}) with 
 the \IOwires, \OIwires, and \Pwires{}. %
We denote the resulting (concrete) encodings as 
 \( \QencoIO \), \( \QencoOI \), and \( \QencoP \), respectively.
Similarly
 \( \QoencoIO \), \( \QoencoOI \), and \( \QoencoP \) are the
 instantiations of the abstract optimised encoding \( \Qoenco \).
For instance, in
 \( \QencoIO \) and  \( \QoencoIO \) an 
abstract wire $\glink ab$ is instantiated with the 
 corresponding concrete wire 
$   \link ab$; and similarly for \OIwires\ and  \Pwires. 

Having shown %
that all  the wires satisfy the requirements of Axiom~\ref{ax:glink-piI},
 we can  use, in the proofs about 
 all concrete encodings (optimised and not)
 the results  in Sections~\ref{sec:abs-enc} and~\ref{sec:oenc}
 for the abstract encoding 
and  its abstract optimisation.

%% file: BT-LT.tex
\section{Full Abstraction for \qLTs{} and \qBTs{}}
\label{sec:BT-LT}

In this section we consider \( \QencoIO \) and \( \QencoOI \), and prove full abstraction with respect to the \qBTs{} and \qLTs{}, respectively.
The main proof is given in Section~\ref{sec:BT-LT:full-abst}.
Before that, in Section~\ref{sec:BT-LT:unsolv}, we discuss the difference between \( \QencoIO \) and \( \QencoOI \) on the encoding of unsolvable terms because it highlights the difference between the two encodings.
Some lemmas about the encoding of unsolvable terms given in Section~\ref{sec:BT-LT:unsolv} will also play key roles in the main proof.
In the proofs we exploit the optimised encodings $\QoencoOI$ and $\QoencoIO$.

\subsection{Encoding of unsolvable terms}
\label{sec:BT-LT:unsolv}
We recall that
the differences between \qBTs{} and \qLTs{} are due to the treatment of unsolvable terms (cf.~Section~\ref{sec:lambda}).
 \qBTs{}  equate all the unsolvable terms,
 whereas \qLTs{} 
distinguish  unsolvables of different order, such as 
 \( \Omega \) and \( \lambda x. \Omega \).
We begin, as an example,
with  the terms \( \Omega \) and \( \lambda x . \Omega \).
As we have seen in Example~\ref{ex:oenc-Omega}, 
in the abstract optimised encoding 
$\Qoenco$ process
\( \oenco \Omega p \) is:
\begin{align*}
  \res {p_0} ( &\diprfx {p_0} {x, q} \oprfx x {q_0} \oprfx {q_0} {y_1, q_1} \left( \riproc {y_1} {r_1} {\oenco x {r_1}} \mid \glink q {q_1} \right) \\
  &\mid \oprfx {p_0}{x_1, p_1} \left( \riproc {x_1} {r_1} {\oenco {\lambda x. \app x x} {r_1}} \mid \glink p {p_1} \right)).
\end{align*}
Its instantiation with \OIwires,
$ \oencoOI \Omega p$,
 cannot do any input action:
as 
 \( \glink p {p_1} \)  becomes the \OIwire{}
\( \linkO {p_1} p \), 
the input occurrence of the free name
 \( p \) is guarded by \( p_1 \),  which in turn is bound by  the
(permeable) prefix at $p_0$. 
Indeed, the only action that
$ \oencoOI \Omega p$
 can perform is (up-to expansion)
$ \oencoOI \Omega p
\trans \tau  \gexp 
 \oencoOI \Omega p $, which 
 corresponds to the %
 reduction \( \Omega \red \Omega \).
Hence,
  $ \oencoO
 {\Omega} p $ cannot 
 match the input action
  \( \oencoO
 {\lambda x.\Omega} p \trans{p(x, q)} \oencoO \Omega q \), and the two
 processes are distinguished.

In contrast, with \IOwires,   processes
$ \oencoI
 {\lambda x.\Omega} p $ and 
 $\oencoI \Omega p $ are indistinguishable. 
As before,  the former process can exhibit an input transition
  \( \oencoI
 {\lambda x.\Omega} p \trans{p(x, q)} \oencoI \Omega q \).
However, now
 \( \oencoI \Omega p \) has a  matching input transition,
 because when \( \glink p {p_1} \) is the \IOwire \( \link p {p_1} \),  the
 input at \( p \) is not guarded.
 The derivative of the input \( p(y, q) \) is
\begin{align*}
  &
    \begin{aligned}
      \res {p_0} (&\diprfx {p_0} {x, q} \oencoI {\app x x} q \\
      &\mid \oprfx {p_0}{x_1, p_1} ( \riproc {x_1} {r_1} {\oencoI {\lambda x. \app x x} {r_1}} \\
      &\mid \oprfx {p_1} {x_2, p_2} (\link {x_2} y \mid \link q {p_2} ) ))
    \end{aligned}  \\
  &\begin{aligned}
      =\res {p_0} (&\diprfx {p_0} {x, q} \oencoI {\app x x} q \\
      &\mid \oprfx {p_0}{x_1, p_1} ( \riproc {x_1} {r_1} {\oencoI {\lambda x. \app x x} {r_1}} \\
      &\mid  \oprfx {p_1} {x_2, p_2} (\riproc {x_2} {r_2} {\oencoI{y} {r_2}} \mid \link q {p_2} )) )
    \end{aligned} \\
&\scong \oencoI {\app \Omega y} q
\end{align*}
(exploiting the  definitions of \(\oencoI {y} {r_2}\) and \( \link {x_2} y \)).
In a similar manner, one then shows that
 \( \oencoI{\Omega} q \) and \(  \oencoI{\app \Omega y }{q} \) 
can  match each other's transitions, and iteratively so, on the
resulting derivatives.
These observations are not limited to \( \Omega \) and \( \lambda x . \Omega\), but can be said against general unsolvable terms.
Below we state these properties as lemmas.
Some of the proofs of the lemmas are given in Appendix~\ref{sec:BT-LT-appx}.

As indicated in the preceding example, in \( \QoencoOI\), a term \( \oencoO{M} {p} \)  can
perform an input transition if and only if $M$ is, or may reduce to,  a function, say
$ M = \lambda x. M' $,
and the input
action intuitively corresponds to consuming the outermost 
`$ \lambda x $'.

\begin{restatable}{lem}{absUnsolvnIn}
\label{lem:abs-unsolv-n-in}
Let \( M \) be an unsolvable term of order \( n \), where \( 0 < n \le \omega \).
Then \( \oencoO M p \) can do a weak input transition at \( p \).
Moreover, if \( \oencoO M p \wtrans{p(x, q)} \proc \), then there exists \( N \) such that \( \proc \gexp \oencoO N q \) and \( N \) is an unsolvable of order \( n - 1 \) (under the assumption \( \omega - 1 = \omega \)).
\end{restatable}

In addition, a  process  \( \oencoO{M} {p} \) is bisimilar to \( \nil \) iff $M$ is an unsolvable of order $0$.
\begin{restatable}{lem}{oencoOZeroTau}
\label{lem:oencoO-order-zero-only-tau}
Let \( M \) be an unsolvable term of order \( 0 \).
Then the only action \( \oencoO M p \) can do is a \( \tau \)-action.
\end{restatable}
Therefore, we have:
\begin{restatable}{lem}{oencoOEquateUnsolvSameOrder}
\label{lem:oencoO-equate-unsolv-same-order}
  Let \( M \) and \( N \) be unsolvables of order \( m \) and \( n \) respectively, where \( 0 \le m, n \le \omega \).
Then    \( \oencoO M p  \bisim \oencoO N p  \) 
iff \( m = n \).
\end{restatable}
\begin{proof}
The only if direction is proved by contraposition.
To see that unsolvable terms \( M \) and \( N \) with different orders are distinguished, we just need to count the number of consecutive inputs that  \( \oencoO M p \) and \( \oencoO N p \) can do.
By Lemmas~\ref{lem:abs-unsolv-n-in} and~\ref{lem:oencoO-order-zero-only-tau}, it follows that \( \oencoO M p \) can do \( n \) consecutive weak input transitions if the order of \( M \) is \( n \); if \( n = \omega \), the number of consecutive inputs that \( \oencoO M p \) can do is unbounded.

We now prove the if direction.
We first prove the case for \( n = m = 0\), and use that result to give the proof for arbitrary \( n \).

For the case \( n = m = 0 \), we show that the relation \( \relR \) defined as
\[
  \bigcup_p \{ (\oencoO M p, \oencoO N p) \mid \text{ \( M, N \)  unsolvables of order \( 0 \)}\}
\]
is a bisimulation up-to expansion.
Suppose that \( \oencoO M p \relR \oencoO N p \).
If \( \oencoO M p \) makes a transition \( \oencoO M p \trans \act \proc \), then \( \act = \tau \) by Lemma~\ref{lem:oencoO-order-zero-only-tau}.
Hence, we have \( \proc \gexp \oenco {M'} p \) with \( M \red M' \) by Lemma~\ref{lem:gen-oenc-tau-has-mathcing-red}.
Since a term obtained by reducing an unsolvable term of order \( 0 \) must also be an unsolvable of order \( 0 \), we have \( \proc \gexp \oencoO {M'} p \relR \oencoO N p \).
In other words, we can take \( \oencoO N p \wtrans {} \oencoO N p \) as the matching transition.

To conclude we show that the relation \( \relR \) defined as
\[
  \bigcup_p \left\{ (\oencoO M p, \oencoO N p) \; \midbar \;
     \text{ \( M, N \) are unsolvables of  the same order}
    \right\}
\]
is a bisimulation up-to expansion.
Suppose that \( \oencoO M p \relR \oencoO N p \).
The order \( 0 \) case is exactly what we proved above, so let us assume that
\( M \) and \( N \) are unsolvable terms whose order is \( n \neq 0 \) (where \( n \) may be \( \omega \)).
Assume that \( \oencoO M p \) makes a transition \( \oencoO M p \trans \act \proc \).
If \( \act = \tau \), then we can reason as we did for the order 0 case.
The only other possibility is the case where \( \act = p(x, q) \) with \( x, q \) being fresh.
Then, by Lemma~\ref{lem:abs-unsolv-n-in}, there exists \( M' \) such that \( \proc \gexp \oencoO {M'} q \) and \( M' \) is an unsolvable of order \( n - 1 \).
By Lemma~\ref{lem:abs-unsolv-n-in}, we have \( \oencoO N p \wtrans{p(x,q)} \proctwo \gexp \oencoO {N'} q \) for some unsolvable term whose order coincides with that of \( M' \).
Since \( M' \) and \( N' \) are unsolvables of the same order, we have \( \oencoO {M'} q \relR \oencoO {N'} q \).
\end{proof}

We have discussed above why, in contrast, 
\(
\QoencoIO$
 equates $\lambda x. \Omega$ and $\Omega$. 
Similarly, \(
\QoencoIO$
 equates all
 the unsolvable terms.

\begin{restatable}{lem}{oencoIUnsolvInput}
\label{lem:oencoI-unsolv-input}
\hfill
\begin{enumerate}
  \item If \( M \) is an unsolvable of order \( 0 \), then \( \oencoI M p \) can do an input at \( p \).
        Moreover, if \( \oencoI M p \trans {p(x, q)} \proc \), then \( \proc \gexp \oencoI {\app M x} q \).
        \label{it:lem:oencoI-unsolv-input:zero}
  \item If \( M \) is unsolvable, then \( \oencoI M p \) can do an input at \( p \).
        Moreover, if \( \oencoI M p \trans {p(x, q)} \proc \), then there exists an unsolvable term \( M' \) such that \( \proc \gexp \oencoI {M'} q \).
        \label{it:lem:oencoI-unsolv-input:any-ord}
\end{enumerate}
\end{restatable}

The case \eqref{it:lem:oencoI-unsolv-input:zero} of Lemma~\ref{lem:oencoI-unsolv-input} is used within the proof of case \eqref{it:lem:oencoI-unsolv-input:any-ord}; the order \( 0 \) case is spelled out since it is the most interesting case in which an input action does not merely correspond to consuming a \( \lambda \).

\begin{restatable}{lem}{oencoIEquateUnsolv}
 \label{lem:oencoI-equate-unsolv}
 For any unsolvable term \( M \), we have \( \oencoI M p \bisim \oencoI \Omega p \).
\end{restatable}

\begin{proof}
  We show that the relation \( \relR \) defined as
  \[
    \bigcup_p \left\{ (\oencoI M p, \oencoI N p) \mid \text{\( M \) and \( N \) are unsolvables} \right\}
  \]
  is a bisimulation up-to expansion.

  Suppose that \( \oencoI M p  \relR \oencoI N p \).
  We consider the case where \( \oencoI M p \) makes the challenge; we omit the opposite case as it is symmetrical.
  By Lemmas~\ref{lem:gen-oenc-unsolvable-no-output} and~\ref{lem:gen-oenc-input-must-be-at-p}, the only actions \( \oencoI M p \) can do is either a \( \tau \)-action or an input at \( p \).

  If \( \oencoI M p \trans {\tau} \proc \) then we can take \( \oencoI N p \wtrans{} \oencoI N p \) as the matching transition because we have \( \proc \gexp \oencoI {M'} p \relR \oencoI N p \) for some unsolvable term \( M' \) such that \( M \red M'\) by Lemma~\ref{lem:gen-oenc-tau-has-mathcing-red}.

  Assume that \( \oencoI M p \trans{p(x,q)} \proc \).
  Then, thanks to~\ref{it:lem:oencoI-unsolv-input:any-ord} of Lemma~\ref{lem:oencoI-unsolv-input}, there exists an unsolvable term \( M' \) such that \( \proc \gexp \oencoI {M'} q \).
  Similarly, by~\ref{it:lem:oencoI-unsolv-input:any-ord} of Lemma~\ref{lem:oencoI-unsolv-input}, we have \( \oencoI N p \trans{p(x, q)} \gexp \oencoI {N'} q \) for some unsolvable term \( N' \).
  The claim follows because \( \oencoI {M'} q \relR \oencoI {N'} q \).
\end{proof}

\subsection{Full abstraction proofs}
\label{sec:BT-LT:full-abst}
The goal of this subsection is to prove the following two theorems:
\begin{restatable}[Full abstraction for \qLT{}]{thm}{fullabstLT} 
\label{thm:full-abst-LT}
   \( \LT M = \LT N \) if and only if \( \encoOI M {} \bisim \encoOI N {} \).
\end{restatable}
\begin{restatable}[Full abstraction for \qBT{}]{thm}{fullabstBT}
\label{thm:full-abst-BT}
   \( \BT M \!=\! \BT N \) if and only if \( \encoI M {} \bisim \encoI N {} \).
\end{restatable}

The proofs exploit  work by Sangiorgi and Xu~\cite{SangiorgiXu18}, which sets  conditions
for obtaining full abstraction with respect to \qLTs{} and \qBTs{} in an encoding of the  $\lambda$-calculus into a process calculus.
Our proofs go through each such condition, showing that it is satisfied.

The rest of this sections is organised as follows.
We first review the general conditions given in~\cite{SangiorgiXu18} (Section~\ref{sec:conditions}).
Then we show that \( \QencoOI \) satisfies the conditions for \qLTs{} and \( \QencoIO \) satisfies that for \qBTs{} (Section~\ref{sec:checking-conditions}).

\subsubsection{General conditions}
\label{sec:conditions}

Here we present a simplified version of the conditions given
in~\cite{SangiorgiXu18}, tailored to our needs,
where the calculus is  \piI{}, and the relations to be considered are bisimilarity and the
expansion relation for \piI{} (the conditions in ~\cite{SangiorgiXu18} are parametric with
respect to the calculus and its behavioural relations).
We also show that some conditions can be proved at the level of the abstract (optimised) encoding \( \Qoenco{} \).

We begin with reviewing some needed terminology.
An \emph{abstraction context}
of an encoding \( \Qenco \) is the context obtained by encoding
the $\lambda$-calculus context \( \lambda x. \hole \), that is, \( \encoE{\lambda x. \hole}{} \) (assuming that $\lambda$-calculus holes are mapped onto identical process holes);
similarly, a
\emph{variable context} of \( \Qenco \) is the encoding of the
$\lambda$-calculus  \(
n \)-hole context \( \app x {\hole_1 \cdots \hole_n} \).  A context is \emph{guarded} if each hole appears
underneath some proper (i.e., non-permeable) prefix.
A \( n \)-hole context \( C \)
\emph{has an inverse with respect to \( \lexp \)} if, for every
\( i \in \{1, \ldots, n \} \), there exists a \( \pi \)-context \( D_i
\) such that \( D_i[C[\seq \agent]] \gexp \bout a {\seq b}.\iproc b z
{\app \agent_i \langle z \rangle} \) for fresh names \( a, b, z\) such
that \( b \in \seq b \).
Intuitively, $D_i$ allows us to bring up the content of the $i-$th hole of the
  context $C$, which can be reached and activated using the prefixes at $a $ and $b$.

\begin{thmC}[\cite{SangiorgiXu18}]
\label{thm:SangiorgiXu}
\let\oldenco\enco
\renewcommand*{\enco}[2]{\Qenco \oldenco{#1}{#2}}
Let \( \Qenco \) be an encoding of the $\lambda$-calculus into \piI.
Suppose that the encoding satisfies the following conditions.
\begin{enumerate}
\item The variable contexts of \( \Qenco \) are guarded;
\item The abstraction and variable contexts of \( \Qenco \) have
    an inverse with respect to  \( \lexp \), provided that the every abstraction \( F \) that fills the holes of the context satisfies \( F = \enco {M} {} \) for some \( \lambda \)-term \( M \);
    \label{it:thm:SangiorgiXu:inv-ctx}
\item \( \Qenco \) and \( \lexp \) validate the \( \beta \) rule;
\item If \( M \) is an unsolvable of order \( 0 \) then \( \enco M {} \bisim \enco \Omega {} \);
\item The terms \( \enco {\Omega} {} \), \( \enco {\app x {\seq  M}} {} \), \( \enco {\app x {\seq  M'}} {} \), and \( \enco {\app y {\seq   M''}} {} \)
are pairwise unrelated by \( \bisim \), assuming that \( x \neq y \)  and  that  tuples \( \seq  M \) and \( \seq  {M'} \) have different  lengths.
\end{enumerate}
Then we have:
\begin{description}
\item[(LT)] if
\label{it:thm:SangiorgiXu:LT}
\begin{enumerate}
\item \( M \), \( N \)  unsolvable of order \( \omega \) implies that \( \enco M {} \bisim  \enco N {} \)
\label{it:thm:SangiorgiXu:LT-unsolv-omega}
and
\item for any \( M \) the term \( \enco {\lambda x.M}{} \) is unrelated by \( \bisim \) to \( \enco \Omega {} \) and to any term  of the form \( \enco {\app x {\seq  M}} {} \),
        \label{it:thm:SangiorgiXu:LT-Omega}
\end{enumerate}
then \( \Qenco \) and \( \bisim \) are fully abstract for \qLT{} equality;
\item[(BT)] if
\label{it:thm:SangiorgiXu:BT}
\begin{enumerate}
  \item \( M \) solvable implies that the term \( \enco {\lambda x.M} {} \) is unrelated by \( \bisim \) to \( \enco \Omega {} \) and to any term  of the form \( \enco {\app x {\seq  M}} {} \), and
  \label{it:thm:SangiorgiXu:BT-lambda}
  \item \( \enco M {} \bisim \enco \Omega {} \) whenever \( M \) is unsolvable of order \( \omega \),
  \label{it:thm:SangiorgiXu:BT-unsolv}
\end{enumerate}
then \( \Qenco \) and \( \bisim \) are fully abstract for \qBT{} equality.
\end{description}
\end{thmC}

\begin{rem}
  The condition~(\ref{it:thm:SangiorgiXu:inv-ctx}) is weaker than the original condition used in~\cite{SangiorgiXu18} (i.e., the new condition does not imply the old one).
  The original condition did not have the assumption that `abstractions that fills the hole must be encodings of \( \lambda \)-terms'.
  We need this condition because the encodings of abstraction and variable use permeable prefixes, and these only behave well with I-/O-respectful processes, such as those resulting from the encoding of \( \lambda \)-terms (Lemma~\ref{lem:glink-subst}).
  However, this does not cause a problem because,
  in~\cite{SangiorgiXu18}, whenever this condition about the inverse
  context is used, the holes are indeed filled with encodings of \( \lambda \)-terms.
\end{rem}

The existence of the inverse context can be proved at the abstract
level (i.e., for \( \Qoenco \) and \( \Qencoabs \)), and hence, it  need not be proved independently for \( \QencoOI \) and \( \QencoIO  \).
\begin{restatable}{lem}{inverseContext}
  \label{lem:gen-oenc-inverse-ctx}
  The abstraction and variable contexts of \( \Qoenco \) have inverse with respect to \( \lexp \), \emph{under the assumption that the every abstraction \( F \) that fills the context satisfies \( F = \oenco {M} {} \) for some \( \lambda \)-term \( M \)}.
\end{restatable}
The proof is by simple algebraic reasonings such as I-{} and O-respectfulness; details are given in Appendix~\ref{sec:BT-LT-full-abst-main}.

\subsubsection{Checking conditions for \qLT{} and \qBT{}}
Thanks to the general conditions we described, to prove the full-abstraction results, we only need to show that \( \QencoOI \) and \( \QencoIO \) satisfy the required conditions.
We first show that \( \QencoOI \) indeed satisfies the conditions for \qLT{}.
\label{sec:checking-conditions}
\fullabstLT*
\begin{proof}
  We check the conditions given in Theorem~\ref{thm:SangiorgiXu}.
  It suffices to give the proof for the optimised encoding \( \QoencoOI \).

  Some conditions are trivial or have already been checked.
  The variable contexts of \( \QoencoOI \) is clearly guarded, the condition about the inverse context is Lemma~\ref{lem:gen-oenc-inverse-ctx}, the validity of \( \beta \) is Lemma~\ref{lem:gen-oenc-validates-beta} and we have proved that unsolvable terms of order \( 0 \) are equated (Lemma~\ref{lem:oencoO-equate-unsolv-same-order}).

  We first see that \( \oencoO \Omega {} \), \( \oencoO {\app x {\seq M}} {} \), \( \oencoO {\app x {\seq M'}} {} \) and \( \oencoO {\app y {\seq M''}} {} \) are pairwise unrelated under the assumption that \( x \neq y \) and \( \seq M \) and \( \seq M' \) have different lengths.
  Since \( \oencoO \Omega p \) cannot do an output action (Lemma~\ref{lem:gen-oenc-unsolvable-no-output}), this process is different from the rest of the processes.
  It is also obvious that \( \oenco {\app y {\seq M''}} p \) is different from \( \oencoO {\app x {\seq M}} {} \) and \( \oencoO {\app x {\seq M'}} {} \).
  We are left with checking that \( \oencoO {\app x {\seq M}} {} \) and \( \oencoO {\app x {\seq M'}} {} \) are not bisimilar.
  Let \( m \defeq \len{\seq M} \) and \( n \defeq \len{\seq M'} \), and without loss of generality, assume that \( m < n \).
  Then \( \oencoO {\app x \seq {M}} p \) can do an input at \( p \) after \( m + 2 \) outputs.
  More concretely, we have \( \oencoO {\app x {\seq M}} p \trans{\bout x {p_0}} \trans{\bout {p_0}{x_1, p_1}} \cdots \trans{\bout {p_n}{x_{n+1}, p_{n+1}}} \trans {p(y,q)} \proc \) for some \( \proc \).
  However, \( \oencoO {\app x \seq {M'}} p \) cannot do an input at \( p \) after \( m + 2 \) outputs, and thus these two processes are not bisimilar.

  Now we look at the conditions in \textbf{(LT)} of Theorem~\ref{thm:SangiorgiXu}.
  The condition~\ref{it:thm:SangiorgiXu:LT-unsolv-omega} holds because we have proved that unsolvable terms of order \( \omega \) are equated (Lemma~\ref{lem:oencoO-equate-unsolv-same-order}).
  It remains to show the condition~\eqref{it:thm:SangiorgiXu:LT-Omega}.
  Clearly, \( \oencoI {\lambda x. M} p \) is not bisimilar to \( \oencoI \Omega p \) or to the encoding of any term of the form \( \app x {\seq M}\) because \( \oencoI {\lambda x. M} p \) can do an input at \( p \), but the others cannot.
\end{proof}

We conclude by proving that \( \QencoIO \) is fully abstract with respect to \qBTs{}.
\fullabstBT*
\begin{proof}
Conditions 1-5 of Theorem~\ref{thm:SangiorgiXu}  are checked similarly
as in the proof of the previous theorem.
  One difference is that for the encoding \( \QoencoIO \) we use
  Lemma~\ref{lem:oencoI-equate-unsolv} to say that unsolvable terms of
  order \( 0 \) are equated.
  Another (minor) difference is how to show that \( \app x {\seq M}\) and \( \app x {\seq N} \) are different when the length of \( \seq {M} \), say \( m \), and the length of \( \seq {N} \), say \( n \), are different.
  Without loss of generality suppose that \( m < n \).
  Then \( \oencoI {\app x \seq {M}} p \) can only do \( m + 1 \) consecutive outputs:
  \begin{align}
    &\oencoI {\app x {\seq M}} p \trans{\bout x {p_0}} \trans{\bout {p_0}{x_1, p_1}} \cdots \trans{\bout {p_n}{x_{n+1}, p_{n+1}}} \nonumber \\
    &\gexp \riproc x {r_1} {\oencoI {M_1} {r_1}} \mid \cdots \mid \riproc x {r_n} {\oencoI {M_n} {r_n}} \mid \link p {p_n}
    \label{eq:thm:full-abst-BT:count-out}
  \end{align}
  because \( \link p p_n \) cannot make an output action.
  On the other hand, \( \oencoI {\app x \seq {N}} p \) can do \( n + 1 \) consecutive outputs.

  Now we check the conditions in~\textbf{(BT)} of Theorem~\ref{thm:SangiorgiXu}.
  Condition~\eqref{it:thm:SangiorgiXu:BT-unsolv} holds because we proved that all the unsolvable terms are equated (Lemma~\ref{lem:oencoI-equate-unsolv}).
  So we are left with checking condition~\eqref{it:thm:SangiorgiXu:BT-lambda}.
  Let \( M \) be a solvable term.
  We need to check that \( \lambda x. M \) is unrelated to \( \Omega \) and any term of the form \( \app w {\seq N} \).
  Since \( M \) is solvable, we have \( \lambda x . M \wred \lambda x. \lambda \seq{y} . \app z {\seq M} \), where \( \seq y \) and \( \seq M \) are possibly empty.
  Since the encoding is valid with respect to \( \beta \)-reduction (Lemma~\ref{lem:gen-oenc-validates-beta}), it suffices to show that the encoding of \( \lambda x. \lambda \seq{y} . \app {z {\seq M}} \) is not bisimilar with the encoding of \( \Omega \) and \( \app w {\seq N} \).
  The process \( \oencoI {\lambda x. \lambda \seq{y} . \app {z {\seq M}}} p \) can do an output on \( z \) after doing some inputs that correspond to the leading \( \lambda \)s.
  However, \( \oencoI \Omega p \) cannot do an output action even
  after some \( \tau \) or input actions .
  Hence, the encoding of \( \lambda x. M \) and \( \Omega \) are not bisimilar.

  Finally, we show that the encodings of \( \lambda x. \lambda \seq{y} . \app z {\seq M} \) and \( \app w {\seq N} \) are not related by \( \bisim \).
  First, note that \( z \) needs to be a free variable and equal to \( w \), otherwise \( \oencoI {\lambda x. \lambda \seq{y} . \app z {\seq M}} p \) is not bisimilar with \( \oencoI {\app w {\seq N}} p \).
  Second, \( \seq{N} \) and \( \seq {M} \) must have the same length.
  Otherwise we can show that the two processes are not bisimilar by
  investigating the number of outputs that the two process can do as we did in~\eqref{eq:thm:full-abst-BT:count-out}; the fact that we have a leading \( \lambda \) in \( \lambda x . \lambda \seq y . \app w {\seq M} \) does not change the argument, as \( w \) must be a free variable and \( \lambda x \) is encoded as a permeable input.
  So, we need to show that \( \oencoI {\lambda x \seq{y}. \app w {\seq M}} p \) and \( \oencoI {\app w {\seq N}} p \) are not equated when \( \seq M \) and \( \seq N \) have the same length.
  Observe that \( \oencoI {\lambda x \seq y. \app w {\seq M}} p \trans {p(x, q)} \gexp \oencoI {\lambda \seq{y} . \app w \seq M} q \).
  The only matching transition \( \oencoI {\app w {\seq N}} p \) can do is
  \begin{align*}
    \oencoI {\app w {\seq N}} p \trans {p(x, q)}
    &\scong \oprfx w {p_0} \oencoN {p_0} q {\oencoIO{N_1} {}, \ldots, \oencoIO{N_n}{}, \oencoIO {x}{}} \\
    &= \oencoIO {\app w {\seq N \; x}} q
  \end{align*}
  where \( \seq N = N_1, \ldots, N_n \) because
  \begin{align*}
    \link p {p_n}
    &\trans {p(x, q)} \oprfx {p_n} {x_{n + 1}, p_{n + 1}} (\link q {p_{n + 1}} \mid \link {x_{n+1}} x) \\
    &=\oprfx {p_n} {x_{n + 1}, p_{n + 1}} (\link q {p_{n + 1}} \mid \riproc {x_{n+1}} {r_{n + 1}} \oencoIO x {r_{n + 1}}).
  \end{align*}
  Once again, we look at the number of consecutive outputs these process can do.
  The process \( \oencoI {\lambda \seq {y} . \app w \seq M} q \) can do \( n + 1 \) outputs whereas \( \oencoIO {\app w {\seq M \; x}} q \) can do \( n + 2 \) outputs.
  We therefore conclude that \( \lambda x. \lambda \seq{y} . \app z {\seq M} \) can never be equated to a term of the form \( \app w {\seq N} \).
\end{proof}

\begin{rem}
Neither \( \QencoIO \) nor \( \QencoOI \) validates the $\eta$-rule (i.e., the $\lambda$-theories induced are not extensional);
this follows from  Theorems~\ref{thm:full-abst-LT} and~\ref{thm:full-abst-BT}
(specifically, conditions~\ref{it:thm:SangiorgiXu:LT-Omega} of \textbf{(LT)} and~\ref{it:thm:SangiorgiXu:BT-lambda} of \textbf{(BT)} mentioned in their proofs, respectively).
\end{rem}

%% file: Dinf.tex
\section{Full Abstraction for \texorpdfstring{\qBTinfs{}}{BTs up to infinite eta}}
\label{sec:Dinf}

In this section we show that the encoding $\QencoP$ obtained by
instantiating the wires of the abstract encoding $\QencoA$ of Section~\ref{sec:abs-enc}
  with  the parallel
wires (the \Pwires)
 yields an encoding that is fully abstract with respect to 
\qBTinfs{}
(Böhm trees with  infinite \( \eta \)-expansion). 

We begin  by showing that
 \( \QencoP \) induces an \emph{extensional} \( \lambda \)-theory.
The result is a useful stepping stone towards the full abstraction result with respect to \qBTinf{},  and may help the readers to understand the differences with respect to the other two concrete encodings examined
in Section~\ref{sec:BT-LT}.
As we know (Section~\ref{sec:wires}) that
 \( \QencoP {} {}\) induces a \( \lambda \)-theory, 
 we remain to check the validity of \( \eta \)-expansion.
\begin{thm}
\label{t:encoP-validates-eta}
  For every \( M \) and \( x \notin \fv M \), we have \( \encoP M p \lexp \encoP {\lambda x. \app M x} p \).
\end{thm}
 \begin{proof}
The process $\encoP {\lambda x. \app M x} p$ is
\[\diprfx p {x, q} \res r (\encoP M r \mid \oprfx r {x',
        q'} (\riproc {x'} {r'} {\encoP x {r'}} \mid \ilink {q}
        {q'} )) .
\]
 As $\riproc {x'} {r'} {\encoP x {r'}} \mathrel{=} \olink {x'} x$, we have:
 \begin{align*}
   \encoP {\lambda x. \app M x} p &\scong \res r (\encoP M r \mid  \diprfx p {x, q} \oprfx r {x',
        q'} (\olink {x'} x \mid \ilink {q} {q'}) ) \\
    &= \res r (\encoP M r \mid \ilink p r ) \gexp \encoP M p %
  \end{align*}
using Lemma~\ref{lem:glink-subst}.
\end{proof}
The above result relies on the 
use of permeable prefixes, both in the encoding of
$\lambda$-abstraction, and within the 
\Pwires. 

\begin{cor}
  Let \( {=_\pi} \defeq \{ (M, N) \mid \encoP M {} \bisim \encoP N {} \} \).
  Then \( =_\pi \) is an extensional \( \lambda \)-theory;
  that is, a congruence on \( \lambda \)-terms that contains \( \beta \) and \( \eta\)-equivalence.
\end{cor}

We are now ready to 
  prove that \( \QencoP \) is fully abstract with respect to   \qBTinfs{}.
We focus on completeness: if \( \BTinf M = \BTinf N \) then \( \encoP M {} \bisim \encoP N {} \).
Soundness will then be essentially derived from completeness, as \qBTinf{} equality is the maximal consistent sensible \( \lambda
 \)-theory (see e.g.~\cite{Barendregt84}).
To show completeness, we rely on the `unique solution of equation'
technique, reviewed in Section~\ref{sec:proof-techniques}.

\begin{rem}[Unique solutions versus up-to techniques]
\label{r:upto}
Results about encodings of $\lambda$-calculus into process calculi, in
previous sections of this paper and in the literature,  usually employ
up-to  techniques for bisimilarity, notably 
\emph{up-to context and expansion}.
In the techniques, expansion is used
to manipulate the derivatives of two transitions so to bring up a
common context.
Such techniques do not seem powerful enough for \qBTinf{}. 
The reason is that 
some of the required transformations would violate expansion (i.e., they
would require to replace a term by  a `less efficient' one), for
instance
 `$\eta$-expanding' a term  
\( \encoP z p \) into {\( \encoP
{\app \lambda y. \app z y} p \)}.
A similar problem has been observed in the case of Milner's
call-by-value encoding~\cite{DurierHS22}.  
\end{rem}

Suppose $\R$ is a \qBTinf{}-bisimulation  (Definition~\ref{def:lassen-bisim}).
We
define a (possibly infinite) system of equations \( \Eqs_{\relR} \), solutions of which will be 
obtained from the encodings of the pairs in $\R$.
There is one equation for each pair \( (M, N) \in {\relR} \).
We describe how each equation is defined, following the clauses of \qBTinf{}-bisimulation.
Take \( (M, N) \in {\relR} \) and assume \( \seq y = \fv{M, N} \).
\begin{enumerate}
  \item \label{it:equation-unsolv}
        If \( M \) and \( N \) are unsolvable, then,
for the right-hand side of the equation,
 we pick a
  non-divergent process that is 
 bisimilar to the encoding of  \(
  \Omega \):
        \[
        \app {X_{M, N}} {\seq y} = \unsolv
        \]
        For instance, we may choose \( \unsolv  \defeq  \bind p \! \diprfx p {x, q} \appI \unsolv q \).
  \item \label{it:equation-hnf}
        If \( M \Hred \lambda x_1 \ldots x_{l + m}. \app z {M_1
      \cdots M_{n + m} } \) and \( N \Hred \lambda x_1 \ldots
    x_{l}. \app z {N_1 \cdots N_n} \), then  the equation is:
        \[
\begin{array}{l}
        \app {X_{M, N}} {\seq y \; p} \defeq \;\\
        \begin{aligned}[t]
          &\diprfx p{x_1, p_1} \cdots \diprfx {p_{l + m - 1}} {x_{l + m}, p_{l + m}}  \oprfx z {w, q} \\
          & \mathcal{O}_{\mathtt{P}}^{n+m} \left\langle
        \begin{aligned}
          & q, p_{l+m}, \appI {X_{M_1, N_1}} {\seq y_1}, \ldots, \appI {X_{M_n, N_n}} {\seq y_n}, \\
          &\appI {X_{M_{n + 1}, x_{l + 1}}}  {\seq y_{n + 1}}, \ldots, \appI {X_{M_{n + m}, x_{l + m}}} {\seq y_{n + m}}
        \end{aligned}
      \right\rangle
       \end{aligned}
\end{array}
         \]
where
\begin{tabular}[t]{l}
 \( \seq y_i = \fv{M_i, N_i} \) for \( 1 \le i \le n \), \\
 \( \seq y_i = \fv{M_i, x_{i - n + l}} \) for \( n + 1 \le  i \le n + m \).
\end{tabular}
and where $\mathcal{O}_{\mathtt{P}}^{r}$ is the instantiation with \Pwires{} of $\mathcal{O}^{r}$ in Figure~\ref{fig:oenc}.
      \item For the case symmetric to~\eqref{it:equation-hnf}, where \( N \) reduces to a
        head normal form with more leading \( \lambda \)-abstractions,
        the equation is defined similarly to~\eqref{it:equation-hnf}.
\end{enumerate}
In (\ref{it:equation-unsolv}), the use of a divergent-free term
 \( \unsolv \) allows us to meet the condition about divergence  of the
 unique-solution technique.
 The right-hand side of (\ref{it:equation-hnf}) intuitively amounts to
 having, as a  body of the equation, the process \( \oencoP {\lambda x_1 \ldots x_{l +m}. \app z {X_{M_1, N_1} \cdots X_{M_{n +m}, x_{l +m}}}} {} \).

We show that the system  \( \Eqs_{\relR} \) of equations has the desired solutions.
\begin{lem}
  \label{lem:oencoP-equates-unsolvables}
 For any \( M \) unsolvable, we have:  \( \oencoP M p \bisim  \oencoP \Omega p   \bisim \appI \unsolv p \).
\end{lem}
\begin{proof}
The reasoning is similar to that in Lemma~\ref{lem:oencoI-equate-unsolv} (whose proof is given in Appendix~\ref{sec:BT-LT-appx:unsolv}).
\end{proof}

\begin{lem}
\label{lem:solutions-of-equations}
  Let \( \relR \) be a \qBTinf{}-bisimulation and \( \Eqs_{\relR} \)
  be the system of equations defined from \( \relR \) as above.
  For each \( (M, N) \in {\relR} \), we define \( F_{M, N} \defeq \bind {\seq x, p} {\oencoP M p} \) and \( G_{M, N} \defeq \bind {\seq x, p} {\oencoP N p} \), where \( \seq x = \fv{M, N} \).
  Then \( \{ F_{M, N } \}_{(M, N) \in {\relR}} \) and \( \{ G_{M, N } \}_{(M, N) \in {\relR}} \) are solutions of \( \Eqs_{\relR} \).
\end{lem}
\begin{proof}
  Take \( (M, N) \in  {\relR} \).
  There are three cases to consider following  Definition~\ref{def:lassen-bisim}.

  If \( M \) and \( N \) are unsolvables then we have
  \[
    F_{M, N} \bisim \bind {\seq y, p} {\oencoP \Omega p} \bisim G_{M, N}
  \]
  by Lemma~\ref{lem:oencoP-equates-unsolvables}.

  If the second clause of Definition~\ref{def:lassen-bisim} holds, then we have
  \begin{align*}
    &\app {F_{M, N}} {\seq y \; p} \\
    &= \oencoP M p \\
    &\gexp \oencoP {\lambda x_1 \ldots x_{l + m}. \app y {M_1 \cdots M_{n + m}}} p \tag{Lemma~\ref{lem:gen-oenc-validates-beta}} \\
    &=
      \begin{aligned}[t]
        &\diprfx p {x_1, p_1} \cdots \diprfx {p_{l + m -1}} {x_{l + m}, p_{l + m}} \oprfx y {w, q} \\
        &\oencoNP[n + m] q  {p_{l + m}} {\oencoP {M_1} {}, \ldots, \oencoP{M_{n + m}} {} }
      \end{aligned} \\
    &=
      \begin{aligned}[t]
        &\diprfx p {x_1, p_1} \cdots \diprfx {p_{l + m -1}} {x_{l + m}, p_{l + m}} \oprfx y {w, q} \\
        &\oencoNP[n + m] q  {p_{l + m}} {F_{M_1, N_1} \langle \seq y_1 \rangle, \ldots, F_{M_{n + m}, x_{l + m}} \langle \seq y_{n + m} \rangle}
      \end{aligned}  \tag{by def. of \( F_{M, N} \)} \\
    &=
      \begin{aligned}[t]
        &\diprfx p {x_1, p_1} \cdots \diprfx {p_{l + m -1}} {x_{l + m}, p_{l + m}} \oprfx y {w, q} \\
        &\oencoNP[n + m] q  {p_{l + m}} {X_{M_1, N_1} \langle \seq y_1 \rangle, \ldots, X_{M_{n + m}, x_{l + m}} \langle \seq y_{n + m} \rangle}
      \end{aligned} \\
    &\quad \{F_{M_1, N_1} / X_{M_1, N_1}, \ldots, F_{M_{n + m}, x_{l + m}} / X_{M_{n + m}, x_{l + m} }\}
  \end{align*}
  as desired.
  The proof for \( G_{M, N} \) is similar, but we need validity of \( \eta \)-expansion (Theorem~\ref{t:encoP-validates-eta}).
  In detail, we have
    \begin{align*}
    &\app {G_{M, N}} {\seq y \; p} \\
    &= \oencoP N p \\
    &\gexp \oencoP {\lambda x_1 \ldots x_l . \app y {N_1 \cdots N_n}} p \tag{Lemma~\ref{lem:gen-oenc-validates-beta}} \\
    &= \diprfx p {x_1, p_1} \cdots \diprfx {p_{l - 1}} {x_l, p_l} \oprfx y {w, q} \oencoNP q  {p_l} {\oencoP {N_1} {}, \ldots, \oencoP{N_n} {} } \\
    &\bisim
      \begin{aligned}[t]
        &\diprfx p {x_1, p_1} \cdots \diprfx {p_{l - 1}} {x_l, p_l} \\
        &\diprfx {p_l}{x_{l + 1}, p_{l + 1}} \cdots \diprfx {p_{l + m -1}} {x_{l + m}, p_{l + m}}  \oprfx y {w, q} \\
        &\oencoNP[n + m] q  {p_{l + m}} {\oencoP {N_1} {}, \ldots, \oencoP{N_n} {}, \oencoP {x_{l + 1}} {}, \ldots, \oencoP {x_{l + m}} {} }
      \end{aligned} \tag{Theorem~\ref{t:encoP-validates-eta}} \\
    &=
      \begin{aligned}[t]
        &\diprfx p {x_1, p_1} \cdots \diprfx {p_{l + m -1}} {x_{l + m}, p_{l + m}} \oprfx y {w, q} \\
        &\oencoNP[n + m] q  {p_{l + m}} {G_{M_1, N_1} \langle \seq y_1 \rangle, \ldots, G_{M_{n + m}, x_{l + m}} \langle \seq y_{n + m} \rangle}
      \end{aligned} \tag{by def. of \( G_{M, N} \)} \\
    &=
      \begin{aligned}[t]
        (&\diprfx p {x_1, p_1} \cdots \diprfx {p_{l + m -1}} {x_{l + m}, p_{l + m}} \oprfx y {w, q} \\
        &\oencoNP[n + m] q  {p_{l + m}} {X_{M_1, N_1} \langle \seq y_1 \rangle, \ldots, X_{M_{n + m}, x_{l + m}} \langle \seq y_{n + m} \rangle}) \\
        &\{G_{M_1, N_1} / X_{M_1, N_1}, \ldots, G_{M_{n + m}, x_{l + m}} / X_{M_{n + m}, x_{l + m} }\}.
      \end{aligned}
  \end{align*}

  The case where the third clause of Definition~\ref{def:lassen-bisim} holds can be proved similarly.
\end{proof}

We also have to show that the system \( \Eqs_{\relR} \) of equations we defined has a \emph{unique} solution.
\begin{lem}
\label{lem:uniqueness-of-solutions}
  The system of equations \( \Eqs_{\relR} \) is guarded and the syntactic solution of \( \Eqs_{\relR} \) is divergence-free.
  Therefore, \( \Eqs_{\relR} \) has a unique solution.
\end{lem}
\begin{proof}
The system  \( \Eqs_{\relR} \) is guarded because all the occurrences of a variable in the right-hand side of an equation are underneath a replicated input prefixing.
  Divergence-freedom follows from the fact that the use of each name
  (bound or free) is strictly polarised in the sense that a name is
  either used as an input or as an output.
  In a strictly polarised setting, no \( \tau \)-transitions can be
  performed even after some visible actions because in \piI{}
only fresh names may  be exchanged.
\end{proof}

\begin{thm}[Completeness for \qBTinf{}]
\label{thm:BE-completeness}
If \( \BTinf M = \BTinf N \) then \( \encoP M {} \bisim \encoP N {} \).
\end{thm}
\begin{proof}
  Consider a \qBTinf{}-bisimulation \( \relR \) that equates \( M \) and \( N \).
Take the system of equations \( \Eqs_{\relR} \) corresponding to \( \relR \) as
  defined above.
  By Lemma~\ref{lem:solutions-of-equations}, \( \oencoP M {} \) and \( \oencoP N {} \) are (components) of the solutions of \( \Eqs_{\relR} \).
  Since the solution is unique (Lemma~\ref{lem:uniqueness-of-solutions}), we derive \( \oencoP M {} \bisim \oencoP N {} \).
  We also have \( \encoP M {} \bisim \encoP N {} \) (equivalence on the non-optimised encodings) because of Lemma~\ref{lem:gen-oenc-is-optimisation}.
\end{proof}

\begin{thm}[Soundness for \qBTinf{}]
\label{thm:BE-soundness}
  If \( \encoP M {} \bisim \encoP N {} \) then \(\BTinf M = \BTinf N \).
\end{thm}
\begin{proof}
  \newcommand*{\EQ}{=_{\pi}}
  Let \( \EQ \) be the equivalence induced by \( \QencoP \) and 
 \piI\ bisimilarity.
  The equivalence \( \EQ \)~is a \emph{sensible} \( \lambda \)-theory by Corollary~\ref{cor:lambda-theory} and Lemma~\ref{lem:oencoP-equates-unsolvables}.
This theory is consistent: for example we have \( \encoP x p \not \bisim \encoP \Omega p \).
By completeness (Theorem~\ref{thm:BE-completeness}),
it contains \qBTinf equality.  Then it must be equal to \qBTinf\ equality 
because  %
the latter is %
the maximal consistent sensible \( \lambda \)-theory~\cite{Barendregt84}.
\end{proof}

%% file: related-work.tex
\section{Concluding Remarks}
\label{sec:related-work}

In the paper we have  presented  
a refinement of Milner’s original encoding of functions as processes %
that is parametric on certain
abstract components called wires. 
Whenever wires satisfy a few algebraic properties, 
 the encoding yields
a 
$ \lambda $-theory. 
We have studied instantiations of the abstract wires with three kinds of concrete wires, 
that differ  on the direction and/or  sequentiality of the control flow produced. 
We have shown that such
 instantiations allow us to obtain 
full abstraction results for 
\qLT{}s, \qBT{}s, and  \qBTinf{}s,
(and hence for $\lambda$-models such as 
$P_\omega$,    
free lazy Plotkin-Scott-Engeler models and $\Dinf$).
In the case of
\qBTinf{}, this implies that  the encoding validates the $\eta$-rule, i.e., it yields an extensional
$\lambda$-theory. 

Following Milner's seminal paper~\cite{Milner90}, the topic of functions as processes
has produced a rich bibliography. 
Below we comment on the works that seem closest to ours.
We have mentioned, in the introduction, related work concerning \qLTs{}
and \qBTs{}. 
We are unaware of results about validity of the $\eta$-rule, let alone \qBTinf{},  in encodings of functions as
processes. 
The only exception is~\cite{BHY01}, where a type system for the
$\pi$-calculus is introduced
so to derive 
  full abstraction for an encoding of PCF (which implies that
 \( \eta \)-expansion for PCF is valid). 
However, in~\cite{BHY01}, types are used to 
 constrain process behaviours, so to remain with processes
that  represent `sequential functional computations'.
Accordingly, the behavioural  equivalence for processes is a  typed contextual equivalence
in which the legal contexts must respect the typing discipline and are
therefore `sequential'. 
In contrast, in  our work $\eta$ is validated under
ordinary (unconstrained) process equivalence in which, for instance, equalities
are preserved by arbitrary process contexts. (We still admit polyadic
communications and hence a sorting system, for readability~---
we  believe  that the same results  hold in a monadic setting.)

In the paper we have considered the theory of the pure untyped
$\lambda$-calculus.
Hence, our encodings model the
call-by-name reduction strategy.
A study of the theory induced by process  encodings of the
call-by-value strategy is~\cite{DurierHS22}.

Our definitions and proofs about
  encodings of permeable prefixes using wires
 follows, and is inspired by, encodings of
forms of permeable prefixes in asynchronous and localised variants of
the $\pi$-calculus using forwarders,
e.g.~\cite{MerroSangiorgi04,Yoshida02}. As commented in the main text,
the  technicalities are however  different, both because
our processes are not localised, and because we employ distinct kinds of wires.

We have worked with bisimilarity, as it is the standard behavioural
equivalence in \piI; moreover,  we could then use some powerful proof
techniques  for it (up-to techniques, unique solution of
equations). The results presented also hold
 for other behavioural equivalences (e.g., may testing), since 
processes encoding functions are confluent.
It would be interesting to extend our work to preorders, i.e.,
looking at  preorder relations for $\lambda$-trees and
$\lambda$-models.

In our work, we derived our abstract encoding from Milner's original encoding of functions.
It is unclear how to transport the same methodology to other variants of Milner's encoding in
 the literature, in particular those 
that closely mimics  the  CPS
translations~\cite{Sangiorgi99,Thielecke97}.

We have derived $\lambda$-tree equalities, in parametric manner, by
different instantiations of the abstract wires.
 Van Bakel et al.~\cite{VBDD02} use an intersection type system,
 parametric with respect to the subtyping relations, to (almost) uniformly characterise  \( \lambda \)-tree equalities (the trees
 considered are those  in our paper together with Böhm trees up-to \emph{finite} \( \eta \)-expansion and Beraducci trees).
 Relationships between \emph{non-idempotent} intersection type systems (or relational models) and \( \lambda \)-trees has also been investigated~\cite{ManzonettoRuoppolo14,PaoliniPR17}.
 We believe that non-idempotent intersection types would be easier to compare with our parametric encoding, but we leave a type theoretic study on the different wires as future work.

We would like to investigate the possible relationship between our work and
game semantics.
In particular, we are interested in the `HO/N style' as it is known to be related to process representations (e.g.,~\cite{HylandOng95,HondaYoshida99,CastellanYoshida19,JaberSangiorgi22}).
HO/N game semantics for the three trees considered in this paper have been
proposed~\cite{KerNO02,KerNO03,OngDiGiantonio04}.
The technical differences with our work are substantial. For instance,
the
game semantics are not given in a parametric manner; and 
 the \( \Dinf \) equality is obtained via 
 Nakajima trees rather than  \( \qBTinf{} \).
Nakajima trees are a different `infinite \( \eta \)-expansion' of
Böhm trees, in the sense that  \( \lambda \seq x . \app y {\seq M} \)
is expanded to \( \lambda \seq {x} z_0 z_1 \ldots. \app y {\seq M} \;
z_0 z_1 \cdots \); that is,  trees may be infinitely branching.
In processes, this would mean, for instance, having  input prefixes that  receive infinitely-many names at the same time.
We would like to understand  whether the three kinds of wires we considered are meaningful in game semantics.
The game semantic counterpart of process wires
are the \emph{copycat strategies}, and they intuitively correspond to \IOwires{}, in that they begin with an O-move (i.e., an input action).
This does not change even in concurrent
game semantics~\cite{MelliesMimram07,CastellanCRW17}. We are not
aware of game models that use strategies corresponding to
the \OIwires{} or the \Pwires{} studied in our paper.

Similarly, we would like to investigate  relationships with  call-by-name translations of  the \( \lambda \)-calculus into (pure) proof-nets~\cite{Danos90}.
We think that our encoding could be factorised into the translation from \( \lambda \)-calculus into proof-nets and a
variant of Abramsky translation~\cite{Abramsky94,BellinScott94}.
In this way, a \Pwire{} for location names would correspond to an infinitely \( \eta \)-expanded form of the axiom link for the type \( o \) according to its recursive equation \( o \cong (!o)^\perp \parr o \).
Infinite \( \eta \)-expansions of the identity axioms have also been considered in Girard's ludics~\cite{Girard01}, where they are called faxes.
Faxes are different from \Pwires{} because faxes satisfy an alternation condition on polarity; polarity, as observed by Honda and Laurent~\cite{HondaLaurent10}, corresponds to locality in \( \pi \)-calculus.

There are some general conditions for $\lambda$-models to be fully abstract for \qBTinfs~{}~\cite{DiGianantonioFH99,Manzonetto09,Breuvart16}.
These cover many $\lambda$-models arising from the study of intersection types, game semantics, and (semantics) of linear logic.
As described, all these areas are known to be closely related to \( \pi \)-calculus (and to each other).
Our full abstraction result for \qBTinfs~{} may be explained through these conditions as well.
We would also like to pursue a more conceptual understanding of our full abstraction results for \qLTs{} and \qBTs{}, potentially by investigating the aforementioned connections.

Processes like wires  (often called \emph{links}) 
appear in session-typed process calculi focusing on the Curry-Howard
isomorphism,  as primitive process constructs
used to represent the identity axiom~\cite{Wadler14}.
Some of our assumptions for wires, cf.\ the substitution-like behaviour when an end-point of the wire is restricted, are then given as explicit  rules of the operational
semantics of such links.

We have studied the properties of the concrete wires used in the paper on processes
encoding functions. We would like to   
establish  more general properties, on arbitrary
processes,  possibly subject to constraints on the usage of the names of the wires.
We  would also like to  see if
other kinds of wires are possible, and which properties they yield.

%% file: notation.tex
\section{List of notations}
\label{a:nota}
The following tables summarise the notations (for encodings, equivalence etc.) used in this paper. Some of the notations are only used in Appendix.
\begin{table}[h!]
\textbf{Encodings}\\[2ex]
\begin{tabular}{ccc}
  \( \Qencoabs \) & Abstract encoding & Figure~\ref{fig:abs-enc} \\
  \( \QencoIO \) & \( \Qencoabs \) instantiated with \IOwires{} & Section~\ref{sec:BT-LT} \\
  \( \QencoOI \) & \( \Qencoabs \) instantiated with \OIwires{} & Section~\ref{sec:BT-LT} \\
  \( \QencoP \) & \( \Qencoabs \) instantiated with \Pwires{} & Section~\ref{sec:Dinf} \\
  \( \Qoenco \) & (Abstract) Optimised encoding & Figure~\ref{fig:oenc} \\
  \( \Qoenco^n \) &  `Optimised encoding of arguments' & Figure~\ref{fig:oenc} \\
  \( \QoencoIO \) & \( \Qoenco \) instantiated with \IOwires{} & Section~\ref{sec:BT-LT} \\
  \( \QoencoOI \) & \( \Qoenco \) instantiated with \OIwires{} & Section~\ref{sec:BT-LT} \\
  \( \QoencoP \) & \( \Qoenco \) instantiated with \Pwires{} & Section~\ref{sec:BT-LT} \\
  \( \QencoM \) & Milner's encoding & Section~\ref{sec:abs-enc-def}\\
  \( \Qenco \) & Metavariable for encodings & --- \\
\end{tabular}

\medskip
\textbf{Wires}\\[2ex]
\begin{tabular}{ccc}
  \( \glink a b\) & (Abstract) Wire & Definition~\ref{ax:glink-piI} \\
  \( \link a b \) & \IOwire{} & Section~\ref{sec:wires} \\
  \( \linkO a b \) & \OIwire{} & Section~\ref{sec:wires} \\
  \( \ilink a b \) & \Pwire{} & Section~\ref{sec:wires} \\
\end{tabular}

\medskip
\textbf{Equivalence and Preorders}\\[2ex]
\begin{tabular}{ccc}
  \( \bisim \) & Weak bisimilarity for \piI{} & Definition~\ref{def:bisim} \\
  \( \sbisim \) & Strong bisimilarity for \piI{} & Definition~\ref{def:bisim} \\
  \( \lexp \) & Expansion relation for \piI{} & Definition~\ref{def:expansion} \\
  \( \scong \) & Structural congruence for \piI{} & Definition~\ref{def:scong} \\
  \( \equiv_{\alpha} \) & \( \alpha \)-equivalence & ---
\end{tabular}

\medskip
\textbf{\( \lambda \)-trees}\\[2ex]
\begin{tabular}{ccc}
  \( \qLT \) & Lévy-Longo tree & Section~\ref{sec:lambda} \\
  \( \qBT \) & Böhm tree & Section~\ref{sec:lambda} \\
  \( \qBTinf \) & Böhm tree up-to infinite \( \eta \)-expansion & Section~\ref{sec:lambda}
\end{tabular}

\medskip
\textbf{Reduction of \( \lambda \)-calculus}\\[2ex]
\begin{tabular}{ccc}
  \( \red \) & \( \beta \)-reduction & Section~\ref{sec:lambda} \\
  \( \snred \) & strong call-by-name reduction & Section~\ref{sec:lambda} \\
  \( \hred \) & head reduction & Section~\ref{sec:lambda}
\end{tabular}
\end{table}

%% file: oenc-appx.tex
\section{Proofs for Section~\ref{sec:oenc}}
\label{sec:oenc-appx}
We present  the proofs of properties of the abstract optimised encoding \( \Qoenco \).

\subsection{Proofs for the basic properties of \texorpdfstring{\( \Qoenco \)}{optimised encoding}}
This section first proves the properties of \( \Qoenco \) that are analogous to those for the unoptimised encoding \( \Qencoabs \).
Then, using these properties, we prove that \( \Qoenco \) is indeed an optimisation of \( \Qencoabs \).

Lemma~\ref{lem:glink-subst-oenc} and~\ref{lem:glink-subst-seq-arg} say that \( \oenco M p
\) and \( \oencoN {p_0} p {\oenco {M_1}{} \cdots \oenco {M_n}{}} \) are respectful.
\begin{lem}
\label{lem:glink-subst-oenc}
  \hfill
  \begin{enumerate}
    \item \( \res q (\glink p q \mid \oenco M q) \gexp \oenco M p \).
          \label{it:lem:glink-subst-oenc:cont}
    \item   \( \res x (\glink x y \mid \oenco M p ) \gexp \oenco {M \sub y x} p \).
          \label{it:lem:glink-subst-oenc:var}
  \end{enumerate}
\end{lem}
\begin{proof}
  Similar to Lemma~\ref{lem:glink-subst} %
\end{proof}

\begin{lem}
\label{lem:glink-subst-seq-arg}
For \( n \ge 1\), we have
\begin{align*}
  \res {p_0} \left(\glink {p_0}{q} \mid \oencoN {p_0} p {\oenco {M_1} {}, \ldots, \oenco {M_n} {}} \right)
  &\gexp \oencoN q p {\oenco {M_1} {}, \ldots, \oenco {M_n} {}}.
\end{align*}
\end{lem}
\begin{proof}
  By induction on \( n \).
  First observe that
  \begin{equation}
    \res y (\glink x y \mid \riproc y p {\oenco M p})  \gexp \riproc x p {\oenco M p}
    \label{eq:lem:glink-subst-seq-arg}
  \end{equation}
  under the assumption that \( y \notin \fv M\).
  This is derived from~\ref{it:ax:glink-piI:subst-rep} of Definition~\ref{ax:glink-piI} together with Lemma~\ref{lem:glink-subst-oenc}.

  The base case is the case where \( n = 1 \).
  In this case, we need to show
  \begin{align*}
    &\res {p_0} \left(\glink {p_0}{q} \mid \oprfx {p_0}{x_1,  p_1}  \left( \riproc {x_1} {r_1} {\oenco {M_1} {r_1}} \mid \glink p {p_1} \right) \right) \\
    &\gexp \oprfx q {x_1,  p_1}  \left( \riproc {x_1} {r_1} {\oenco {M_1} {r_1}} \mid \glink p {p_1} \right).
  \end{align*}
  Using \eqref{eq:lem:glink-subst-seq-arg} and the transitivity of
  wires, we derive the expansion by
  applying~\ref{it:ax:glink-piI:subst-dout} of
  Definition~\ref{ax:glink-piI}. 

  Now we consider the case \( n \ge 1 \).
  We apply~\ref{it:ax:glink-piI:subst-dout} of Definition~\ref{ax:glink-piI}.
  The premise of this law is satisfied because of~\eqref{eq:lem:glink-subst-seq-arg} and the induction hypothesis.
\end{proof}

We now prove that the optimised encoding validates \( \beta \)-reduction.
\begin{lem}
\label{lem:gen-oencon-addition}
  Let \( m, n \ge 0 \).
  Then
  \begin{align*}
    &\res q
      (\oencoN[m] p q {\oenco{M_1}, \ldots, \oenco{M_m}{}} \mid \oencoN q r {\oenco{N_1}{}, \ldots, \oenco{N_n}{}}) \\
    &\gexp \oencoN[m + n] p r {\oenco{M_1}{}, \ldots, \oenco{M_m}{}, \oenco{N_1}{}, \ldots, \oenco{N_n}{}}.
  \end{align*}
\end{lem}
\begin{proof}
  Follows from Lemma~\ref{lem:glink-subst-seq-arg}.
\end{proof}

\begin{lem}
  \label{lem:gen-oenc-subst}
  Suppose that \( x \notin \fv N \). Then \( \res x \left( \oenco M p \mid \riproc x q {\oenco N q} \right) \gexp \oenco {M \sub N x} p \)
\end{lem}
\begin{proof}
  By induction on the structure of \( M \).
  The proof is similar to that of the unoptimised case (Lemma~\ref{lem:genc-subst}).
  Indeed the proof for the base case, namely the case \( M = x \), is exactly the same as that of Lemma~\ref{lem:genc-subst} since \( \oenco x p  = \Qencoabs \enco x p \).
  The inductive case follows from the induction hypothesis and  the replication theorems.
\end{proof}

\begin{lem}
  \label{lem:gen-oenc-app-oenco-oencon}
  If \( n \ge 1 \), then
  \begin{align*}
    \res q ( \oenco{M_0} q \mid \oencoN {q} {p} {\oenco {M_1}{}, \ldots, \oenco{M_n}{}}) \gexp  \oenco{\app {M_0} {M_1 \cdots M_n}} p.
  \end{align*}
\end{lem}
\begin{proof}
  \mylabel{Case \( M_0 = x \)}
  This case follows from Lemma~\ref{lem:glink-subst-seq-arg}.
  More precisely, we have
  \begin{align*}
    &\res q ( \oenco{x} q \mid \oencoN q p {\oenco {M_1}{}, \ldots, \oenco{M_n}{}}) \\
    &= \res q(\oprfx x {q'} \glink q {q'} \mid \oencoN q p {\oenco {M_1}{}, \ldots, \oenco{M_n}{}}) \\
    &\scong \oprfx x {q'} \res q(\glink q {q'} \mid \oencoN q p {\oenco {M_1}{}, \ldots, \oenco{M_n}{}}) \\
    &\gexp \oprfx x {q'} \oencoN {q'} {p}{\oenco {M_1}{}, \ldots, \oenco{M_n}{}}  \tag{Lemma~\ref{lem:glink-subst-seq-arg}}. \\
    &= \oenco {\app x {M_1 \cdots M_n}} p.
  \end{align*}

  \mylabel{Case \( M_0 =  \lambda x . M \)}
  By definition,
  \begin{align*}
    &\res q ( \oenco{\lambda x. M } q \mid \oencoN q p {\oenco {M_1}{}, \ldots, \oenco{M_n}{}}) \\
    &\equiv \oenco {\app  {(\lambda x. M)} {M_1 \cdots M_n}} p
  \end{align*}

  \mylabel{Case \( M_0 = \app x {N_1 \cdots N_m} \) with \( m \ge 1\)}
  In this case, we have
  \begin{align*}
    &\res q ( \oenco{M_0} q \mid \oencoN q p {\oenco {M_1}{}, \ldots, \oenco{M_n}{}}) \\
    &= \res q
        (\oprfx {x} {q_0} \oencoN[m] {q_0} q {\oenco {N_1}{}, \ldots, \oenco {N_m}{}} \mid \oencoN q p {\encoN {M_1}{}, \ldots, \encoN{M_n}{}}) \\
    &\scong \oprfx {x} {q_0} \res q
        (\oencoN[m] {q_0} q {\oenco {N_1}{}, \ldots, \oenco {N_m}{}} \mid \oencoN q p  {\oenco {M_1}{}, \ldots, \oenco{M_n}{}}) \\
    &\gexp \oprfx {x} {q_0} \oencoN[m + n] {q_0} {p}{\oenco {N_1}{}, \ldots, \oenco {N_m}{}, \encoN {M_1}{}, \ldots,  \encoN{M_n}{}} \tag{Lemma~\ref{lem:gen-oencon-addition}} \\
    &= \oenco {\app x {N_1 \cdots N_m \, M_1 \cdots M_n}} p
  \end{align*}

  \mylabel{Case \( M_0 = \app {(\lambda x. N_0)} {N_1 \cdots N_m} \) with \( m \ge 1\)}
  Similar to the previous case.
\end{proof}

\genOencValidatesBeta*
\begin{proof}
  It suffices to consider the case where \( M = \app{(\lambda x. M_0)} {M_1 \ldots M_n} \), where \( n \ge 1 \), because the other cases follow from the precongruence of \( \lexp \).
  We only consider the case where \( n \ge 2 \) because the case \( n = 1 \) can be proved as in the case for the unoptimised encoding.
  By definition, \( \oenco M p \) is
  \[
    \res {p_0} \left(\diprfx {p_0} {x, q} \oenco {M_0} q \mid \oencoN {p_0} p {\oenco{M_1}{} \cdots \oenco{M_n}{}} \right)
  \]
  By interaction on \( p_0 \) (Lemma~\ref{lem:permeable-comm}), we have
  \begin{align*}
    \oenco M p
    &\gexp \resb{x, q}
      (\oenco{M_0} q  \mid \riproc x {r_1}{\oenco {M_1 } {r_1}} \mid \oencoN[n - 1] {q} p {\oenco{M_2}{} \cdots \oenco{M_n}{}}).
  \end{align*}
  Note that the assumption of Lemma~\ref{lem:permeable-comm} is satisfied by Lemma~\ref{lem:glink-subst-oenc}.
  The claim follows because
   \begin{align*}
    &\resb{x, q}
        (\oenco{M_0} q  \mid \riproc x {r_1}{\oenco {M_1 } {r_1}} \mid \oencoN[n - 1] {q} p {\oenco{M_2}{} \cdots \oenco{M_n}{}}) \\
     &\gexp \res{q} (\oenco{M_0 \sub {M_1} x} q \mid \oencoN[n - 1] {q} p {\oenco{M_2}{} \cdots \oenco{M_n}{}}) \tag{Lemma~\ref{lem:gen-oenc-subst}} \\
     &\gexp \oenco {\app {M_0 \sub {M_1} x} {M_2 \cdots M_n}} p \tag{Lemma~\ref{lem:gen-oenc-app-oenco-oencon}}
  \end{align*}
\end{proof}

Finally, we prove that \( \Qoenco \) is indeed an optimisation.
\genOencIsOptimisation*
\begin{proof}
  By induction on the structure of \( M \).
  The cases of variables and abstraction are straightforward.
  Consider now \( M = \app y {N_1 \cdots N_n} \).
  We use induction on \( n \).
  For the base case, i.e.~\( M = \app y N_1 \) we have:
\begin{align*}
  &\Qencoabs \encoN {\app y {N_1}} p \\
  &= \res q \left(\Qencoabs \enco y q \mid \oprfx q {x, p'} \left(\riproc x r {\Qencoabs \enco {N_1} r} \mid  \glink {p} {p'} \right) \right)  \\
  &= \res q \left(\oprfx y {q'} {\glink q {q'}} \mid \oprfx q {x, p'} \left(\riproc x r {\Qencoabs  \enco {N_1} r} \mid  \glink {p} {p'} \right) \right)  \\
  &\scong \oprfx y {q'} \res q \left(\glink q {q'} \mid \oprfx q {x, p'} \left(\riproc x r {\Qencoabs \enco {N_1} r} \mid  \glink {p} {p'} \right) \right)  \\
  &\gexp \oprfx y {q'} \res q \left(\glink q {q'} \mid \oprfx q {x, p'} \left(\riproc x r {\oenco {N_1} r} \mid  \glink {p} {p'} \right) \right)  \tag{i.h.} \\
  &\gexp \oprfx y {q'} \oprfx {q'} {x, p'} \left(\riproc x r {\oenco {N_1} r} \mid  \glink {p} {p'} \right)  \tag{Lemma~\ref{lem:glink-subst-seq-arg}} \\
  &\equiv_\alpha \oprfx {y}{p_0} \oprfx {p_0} {x, p_1} \left(\riproc x r {\oenco {N_1} r} \mid  \glink {p} {p_1} \right) \\
  &= \oenco {\app y {N_1}} p
  \end{align*}
  The inductive case for \( n \) can be proved similarly using the induction hypothesis and Lemma~\ref{lem:gen-oencon-addition}.

  The case of \( M = \app {(\lambda x.  N)} {N_1 \cdots N_n}  \) is also handled in the same manner.
\end{proof}
\subsection{Properties about transitions}
Now we prove the properties about the transitions \( \oenco M p \) can do.
The first thing we prove is the operational correspondence for  \( \tau \)-transitions.
\genOencTauHasMatchingRed*
\begin{proof}
  By induction on the structure of \( M \).
  The case where \( M = \app x {\seq M} \), where \( \seq M \) is a possibly empty sequence of terms, is trivial since \( \oenco M p \) cannot make any \( \tau \)-action.
  The case for \( M = \lambda x. M_0 \) is also straightforward: it follows from the induction hypothesis.

  We now consider the remaining case where \( M = \app {(\lambda x. M_0)} {M_1 \cdots M_n} \) and \( n \ge 1\).
  Recall that \( \oenco  {\app {(\lambda x. M_0)} {M_1 \cdots M_n}} p \) is
  \[
    \res {p_0}
    \begin{aligned}[t]
      (&\diprfx {p_0}{x, q} \oenco {M_0} q \mid \oprfx {p_0}{x_1,  p_1} \cdots \oprfx {p_{n - 1}}{x_n, p_n}\\
      &(\riproc {x_1} {r_1} {\oenco {M_1} {r_1}} \mid \cdots \mid \riproc {x_n} {r_1} {\oenco {M_n} {r_n}}  \mid \glink p {p_n} ))
    \end{aligned}
  \]
  There are two cases to consider: (1) the case where the \( \tau \)-action originates from the \( \tau \)-action on \( \oenco {M_0} q \) and (2) the case where the \( \tau \)-action is caused by the interaction at \( p_0 \).
  The former case can be easily proved by using the induction hypothesis.
  The latter case is the most important case.
  Since \( \oenco {M_0} q \) is I-respectful with respect to \( \glink {q'} q \) and O-respectful with respect to \( \glink x {x'} \) (Lemma~\ref{lem:glink-subst-oenc}), we can use the communication law for the permeable prefixes on \( p_0 \).
  Therefore, we have
  \begin{align*}
    &\res {p_0}
    \begin{aligned}[t]
      (&\diprfx {p_0}{x, q} \oenco {M_0} q \mid \oprfx {p_0}{x_1,  p_1} \cdots \oprfx {p_{n - 1}}{x_n, p_n}\\
      &( \riproc {x_1} {r_1} {\oenco {M_1} {r_1}} \mid \cdots \mid \riproc {x_n} {r_1} {\oenco {M_n} {r_n}}  \mid \glink p {p_n} ))
    \end{aligned} \\
    &\gexp \resb{x_1, p_1}
      \begin{aligned}[t]
      (&\oenco {M_0 \sub{x_1} x} q \mid \riproc {x_1} {r_1} {\oenco {M_1} {p_1}} \mid \oprfx {p_1}{x_2,  p_2} \cdots \oprfx {p_{n - 1}}{x_n, p_n}\\
      &( \riproc {x_2} {r_2} {\oenco {M_2} {p_2}} \mid \cdots \mid \riproc {x_n} {r_1} {\oenco {M_n} {r_n}}  \mid \glink p {p_n}))
    \end{aligned} \tag{Lemma~\ref{lem:permeable-comm}} \\
    &\gexp \res {p_1}
      \begin{aligned}[t]
      (&\oenco {M_0 \sub {M_1} x} {p_1} \mid \oprfx {p_1}{x_2,  p_2} \cdots \oprfx {p_{n - 1}}{x_n, p_n}\\
      & (\riproc {x_2} {r_2} {\oenco {M_2} {r_2}} \mid \cdots \mid \riproc {x_n} {r_1} {\oenco {M_n} {r_n}}  \mid \glink p {p_n} ))
    \end{aligned}  \tag{Lemma~\ref{lem:gen-oenc-subst}}
  \end{align*}
  If \( n \ge 2 \), we can apply Lemma~\ref{lem:gen-oenc-app-oenco-oencon} and obtain \( \proc \gexp \oenco {\app{M_0 \sub {M_1} x} {M_2 \cdots M_n}} p \).
  If \( n = 1 \), the subprocess of the form \( \oprfx {p_1}{x_2,  p_2} \cdots \) is simply \( \glink p {p_1} \).
  Hence, by Lemma~\ref{lem:glink-subst-oenc}, we have
  \begin{align*}
    \oenco {\app {(\lambda x. M_0)} {M_1}} p
    &\gexp \res {p_1} \left( \oenco {M_0 \sub {M_1} x} {p_1}  \mid \glink p {p_1}\right) \\
    &\gexp  \oenco {M_0 \sub {M_1} x} {p_1}
  \end{align*}
  as desired.
\end{proof}

The next thing we prove is the property about input actions.
\genOencInputMustBeAtp*
\begin{proof}
  By induction on \( M \) with a case analysis on the shape of \( M \).

  \mylabel{Case \( M = x \)}
  In this case, \( \oenco M p = \oprfx x {p'} \glink p {p'} \).
  Since the only free name that appears in an input occurrence is \( p \) (because of~\ref{it:ax:glink-piI:fn} of   Definition~\ref{ax:glink-piI}), the only possible input action \( \oenco M p \) can do is an input on \( p \).
  (Note that whether the process can do an input on \( p \) will depend on the concrete instantiation of \( \glink p q \).)

  \mylabel{Case \( M = \lambda x . M_0 \)}
  Since \( \oenco M p = \diprfx {p} {x , q} \oenco {M_0} q \), if \( \oenco M p \trans {\act} \proc \) and \( \act \) is an input action, then this action must either be an input on \( p \) or an input that originates from \( \oenco {M_0} q \).
  In the latter case, \( \act \) must be an input on \( q \) by the induction hypothesis.
  Since \( q \) is bound by \( \diprfx p {x, q} \), this action cannot induce an input action of \( \oenco M p \).

  \mylabel{Case \( M = \app x {M_1 \cdots M_n} \)}
  In this case, \( \oenco M p \) is
  \begin{align*}
    &\oprfx x {p_0} \oprfx {p_0}{x_1,  p_1} \cdots \oprfx {p_{n - 1}}{x_n, p_n} \\
    &\left( \riproc {x_1} {r_1} {\oenco {M_1}   {r_1}} \mid \cdots \mid \riproc {x_n} {r_n} {\oenco {M_n} {r_n}}  \mid \glink p {p_n} \right)
  \end{align*}
  Since the only free name that may appear in an input occurrence which is not guarded by a (non-permeable) prefixing is \( p \), the only possible input action \( \oenco M p \) can do is an input on \( p \).

  \mylabel{Case \( M = \app {(\lambda x. M_0)} {M_1 \cdots M_n} \)}
  By combining the argument we made in the previous two cases.
\end{proof}

We now consider the relationship between output actions and head normal forms.
As an auxiliary definition, we introduce a special form of a context.
\begin{defi}
\emph{H-contexts} are contexts defined by the following grammar:
\[
  \hctx \Coloneqq \hole \mid \app {(\lambda x. H)} {M_1 \cdots M_n} \quad (n \ge 0 )
\]
\end{defi}

\begin{lem}
\label{lem:hctx-red}
  Let \( M \defeq \hctx[\app x \seq{M}] \), where \( \seq{M} \) is a possibly empty sequence of terms, and assume that \( x \in \fv M \).
  Then \( M \Hred \lambda \seq{y} . \app x {\seq N} \) for some possibly empty sequences of variables \( \seq{y}\) and terms \( \seq{N} \).
\end{lem}

\begin{lem}
\label{lem:gen-oenc-output-implies-hctx}
  Let \( M \) be a \( \lambda \)-term and suppose that \( \oenco M p \trans \act P \) for an output action \( \act \).
  Then the action \( \act \) must be of the form \( \bar{x}(p)\) for a fresh \( p \) and a variable  \( x \in \fv{M}\), and there exists a H-context \( \hctx \) such that \( H[\app x {\seq M}] = M \) for some possibly empty sequence \(   \seq M = M_1, \ldots M_n\).
\end{lem}
 \begin{proof}
  By induction on \( M \) with a case analysis on the shape of \( M \).

  \mylabel{Case \( M = x \)}
  In this case, \( \oenco M p = \oprfx x {p'} \glink p {p'} \), and the only output action \( \oenco M p \) can do is \( \bout x {p'} \) (cf.~\ref{it:ax:glink-piI:fn} of Definition~\ref{ax:glink-piI}).
  We can take the empty context \( \hole \) for \( \hctx \) and the empty sequence for \( \seq M \).

  \mylabel{Case \( M = \lambda x . M_0 \)}
  Since we have \( \oenco M p = \diprfx {p} {x , q} \oenco {M_0} q \), if \( \oenco M p \trans {\bout y r} \proc \) then this action must originate from \( \oenco {M_0} q \) and we must have \( x \neq y \).
  Hence, by the induction hypothesis, there is a H-context \( \hctx' \) and a sequence of terms \( \seq M \) such that  \( M_0 = \hctx'[\app y {\seq M}] \) with \( y \in \fv {M_0}\).
  We can take \( \hctx \) as \( \lambda x. \hctx' \), and since \( y \neq x \) we also have \( y \in \fv M \).

  \mylabel{Case \( M = \app x {M_1 \cdots M_n} \)}
  In this case, \( \oenco M p \) is
  \begin{align*}
    &\oprfx x {p_0} \oprfx {p_0}{x_1,  p_1} \cdots \oprfx {p_{n - 1}}{x_n, p_n} \\
    &\left( \riproc {x_1} {r_1} {\oenco   {M_1} {r_1}} \mid \cdots \mid \riproc {x_n} {r_n} {\oenco {M_n} {r_n}}  \mid \glink p {p_n} \right).
  \end{align*}
  Obviously, the only output \( \oenco M p \) can do is \( {\bout x {p_0}} \).
  Hence, the claim holds by taking \( \hctx = \hole \) and \( \seq M = M_1, \ldots, M_n \).

  \mylabel{Case \( M = \app {(\lambda x. M_0)} {M_1 \cdots M_n} \)}
  Recall that \( \oenco M p \) is
  \[
    \res {p_0}
    \begin{aligned}[t]
      (&\diprfx {p_0}{x, q} \oenco {M_0} q \mid \oprfx {p_0}{x_1,  p_1} \cdots \oprfx {p_{n - 1}}{x_n, p_n}\\
        &\left( \riproc {x_1} {r_1} {\oenco {M_1} {r_1}} \mid \cdots \mid \riproc {x_n} {r_1} {\oenco {M_n} {r_n}}  \mid  \glink p {p_n} \right) )
      \end{aligned}
  \]
  If \( \oenco M p \) does an output action, then this action must originate from \( \oenco {M_0} q \) and the subject of the action must be different from \( x \).
  Assume that \( \oenco {M_0} q \trans \act \proc' \) for an output action \( \act \) whose subject is not \( x \).
  By the induction hypothesis, \( \act = \bout y {r} \) and there is a H-context \( \hctx' \) and a sequence of terms \( \seq M' \) such that \( M_0 = \hctx'[\app y {\seq M'}] \) with \( y \in \fv {M_0}\).
  We can take \( H = \app {(\lambda x . \hctx')} {M_1 \cdots M_n} \).
  Since \( M = H [\app y {\seq M'}] \) and \( y \in \fv M \) the claim follows.
\end{proof}

\genOencOutputImpliesSolvable*
\begin{proof}
  By Lemma~\ref{lem:hctx-red} and Lemma~\ref{lem:gen-oenc-output-implies-hctx}.
\end{proof}

\genOencUnsolvableNoOutput*
\begin{proof}
  Since \( M \) is an unsolvable term, by Lemma~\ref{lem:gen-oenc-output-implies-solvable}, \( \oenco M p \) cannot do an output action.
  Hence, if \( \oenco M p \wtrans \act \proc \), where \( \act \) is an output action, we must have \( \oenco M p {(\trans \tau)}^n \proc \trans \act \proc' \) for \( n \ge 1 \).
  Assume that such an \( n \) exist.
  Then, by repeatedly applying Lemma~\ref{lem:gen-oenc-tau-has-mathcing-red}, we get \( \proc \gexp \oenco {M'} p \) for some term \( M' \) such that \( M \red^n M' \).
  Note that \( M' \) is also an unsolvable term.
  However, this is a contradiction since \( \proc \trans \act \proc' \) for an output action \( \act \), but \( \oenco {M'} p \not \trans \act \) by Lemma~\ref{lem:gen-oenc-output-implies-solvable}.
\end{proof}

%% file: wires-appx.tex
\section{Supplementary Materials for Section~\ref{sec:wires}}
\label{sec:wires-appx}
In this section, we prove that the three concrete wires we introduced satisfy the properties of Definition~\ref{ax:glink-piI}.
As explained in the main text, we first show that the wires are transitive, and then the other laws are proved by algebraic reasoning exploiting the transitivity of wires.

\subsection{Proofs for transitivity}
We prove that \IOwires{} and  \OIwires{} are transitive.
The reasoning is similar to that of the \Pwires{} which we saw in Section~\ref{sec:wires}; we use  bisimulation up-to context and expansion.
The proofs for the \IOwires{} and \OIwires{} are slightly simpler than that of the \Pwires{} since these wires have fewer permeable prefixes, and the two proofs are essentially the same
because of their `duality'.

Chains of \IOwires{} and \OIwires{}, denoted by \( \chain {\mathtt{IO}} n a b \) and \( \chain {\mathtt{OI}} n a b \), respectively, are defined similar to the chains of \Pwires{}.
(Only for the \( \chain {\mathtt{OI}} n p q \), we assume that \( p \) is an output name and \( q \) is an input name so that the `direction of the chain' becomes from output to input.)

\begin{lem}
  \label{lem:linkI-trans}
  The wires \( \link p q \) and \( \link x y \) are transitive.
  That is, we have \( \res q (\link p q \mid \link q r) \gexp \link p r \) and \( \res y (\link x y \mid \link y z) \gexp \link x z \).
\end{lem}
\begin{proof}
  \renewcommand*{\chain}[3]{\mathrm{chain}_{\mathtt{IO}}^{#1}(#2, #3)}
  We strengthen the statement and prove the transitivity for chains of wires of any length.
  The relations we consider for the  proof are
  \begin{align*}
    {\relR_1} &\defeq \left\{ (\link {p_0} {p_n}, \chain n {p_0} {p_n}) \mid n \ge 2 \right\} \\
    {\relR_2} &\defeq \left\{ (\link {x_0}{x_n}, \chain n {x_0}{x_n})  \mid n \ge 2 \right\}.
  \end{align*}
  We show that \( {\relR_1} \cup {\relR_2}\) is an expansion up-to \( \lexp \) and context %

    We first consider the case for location names.
    Suppose \( \link {p_0} {p_m} \relR_1 \chain m {p_0} {p_m} \) for some \( m \ge 2 \).
    We only consider the case where the process on the right-hand side makes the challenge; the opposite direction can be proved similarly.
    There is only one possible actions the process can do, namely an input at \( p_0 \).

    First we prove the following auxiliary statement by induction on \( n \).
    \begin{quote}
      For any \( n \ge 2 \), if \( \chain n {p_0} {p_n}  \trans{p_0{(x_0, q_0)}} \proc \), then we have \( \proc \gexp \oprfx {p_n} {x_n, q_n} (\chain {2n - 1} {x_n}{x_0} \mid \chain {2n - 1} {q_0}{q_n} ) \)
    \end{quote}
  The base case is \( n = 2 \).
  Recall that \( \link {p_0} {p_1} \) and  \( \link {p_1}{p_0} \) are of the form
  \begin{align*}
    \bin {p_0}{x_0, q_0}. (\res {x_1', q_1'})
    \begin{aligned}[t]
      (&\bout {p_1} {x_1, q_1}.(\link {x_1} {x_1'} \mid \link{q_1'}{q_1} ) \\
      & \mid \link {x_1} {x_0} \mid \link {q_0} {q_1'})
    \end{aligned}
  \end{align*}
  and
  \begin{align*}
  \bin {p_1}{x_1, q_1}. (\res {x_2', q_2'})
    \begin{aligned}[t]
      (&\bout {p_2} {x_2, q_2}.(\link {x_2} {x_2'} \mid \link{q_2'}{q_2} ) \\
      & \mid \link {x_2'} {x_1} \mid \link {q_1} {q_2'})).
    \end{aligned}
   \end{align*}
  Hence the derivative of the transition \( \trans{p_0(x_0, q_0)}\) is
   \begin{align*}
     &
    \begin{aligned}[t]
      &\resb {x_1', q_1', p_1} \\
      &(\bout {p_1} {x_1, q_1}.(\link {x_1} {x_1'} \mid \link{q_1'}{q_1} ) \mid \link {x_1'} {x_0} \mid \link {q_0} {q_1'} \\
      &\mid \bin {p_1}{x_1, q_1}. (\res {x_2', q_2'})(
      \begin{aligned}[t]
      &\bout {p_2} {x_2, q_2}.(\link {x_2} {x_2'} \mid \link{q_2'}{q_2} ) \mid \link {x_2'} {x_1} \mid \link {q_1} {q_2'}))
    \end{aligned}
  \end{aligned} \\
    &\gexp
    \begin{aligned}[t]
      &\resb {x_1', x_1, x_2', q_1', q_1, q_2'}\\
      &(\link {x'_2}  {x_1} \mid \link {x_1} {x_1'}  \mid  \link {x_1'} {x_0} \\
      &\mid \link {q_0} {q_1'} \mid \link{q_1'}{q_1} \mid \link {q_1} {q_2'} \\
      &\mid \bout {p_2} {x_2, q_2}.(\link {x_2} {x_2'} \mid \link{q_2'}{q_2}))
    \end{aligned} \tag{interaction at \( p_1\)} \\
     &\scong \resb {x_2', q_2'}
         (\chain 3 {x_2'} {x_0} \mid \chain 3 {q_0} {q_2'} \mid \bout {p_2} {x_2, q_2}.(\link {x_2} {x_2'} \mid \link{q_2'}{q_2})) \\
     &= \oprfx {p_2} {x_2, q,_2} (\chain 3 {x_2}{x_0} \mid \chain 3 {q_0} {q_2})
   \end{align*}
  as desired.
  The inductive case can be proved similarly.

  Hence if \( \chain m {p_0} {p_{m}} \trans{p_0(x_0, q_0)} \proc \), we have \( \proc \gexp \oprfx {p_m}{x_m, q_m} (\chain {2m - 1}{x_m}{x_0} \mid \chain {2m - 1}{q_0}{q_m} ) \).
  For the matching transition, we pick
  \[
    \link {p_0} {p_m} \trans{p_0(x_0, q_0)} \oprfx {p_m}{x_m, q_m}( \link {x_m} {x_0} \mid \link {q_0} {q_m} ).
  \]
  We can take \( C \defeq \oprfx {p_m}{x_m, q_m} (\hole \mid \hole) \) as the common context and conclude this case because \( \link{q_0} {q_m} \relR_1 \chain {2m - 1}{q_0}{q_m}  \) and \( \link {x_m} {x_0} \relR_2 \chain {2m - 1}{x_m}{x_0} \).

  The case for the variable name is proved similarly.
  Suppose \( \link {x_0} {x_m} \relR_2 \chain m {x_0} {x_m} \) for some \( m \ge 2 \).
  The only action the two processes can do is the input at \( x_0 \).
  As in the case for location names, we can show that
  \begin{quote}
    For any \( n \ge 2 \), if \( \chain n {x_0} {x_n}  \trans{x_0{(p_0)}} \proc \), then \( \proc \gexp  \chain n {x_0} {x_n} \mid \oprfx {x_n} {p_n} \chain {2n - 1} {p_0}{p_n} \)
  \end{quote}
  by induction on \( n \).
  We omit the proof as it is similar to the case for location
  names; instead of the expansion relation for
  interactions among linear names, the proof uses replication theorems
  (the laws~\eqref{it:lem:rep-thm:par},~\eqref{it:lem:rep-thm:gc}
  and~\eqref{it:lem:rep-thm:comm} of  Lemma~\ref{lem:rep-thm}). 
  So if \( \chain m {x_0} {x_m}  \trans {x_0(p_0)} \proc \), we can take \( \link {x_0} {x_m}  \trans{x_0(p_0)}  \link {x_0} {x_m} \mid \oprfx {x_m} {p_m} \link {p_0}{p_m} \) as the matching transition.
  We have \( \proc  \gexp  \chain n {x_0} {x_m} \mid \oprfx {x_m} {p_m} \chain {2n - 1} {p_0}{p_m} \); \( \link{x_0}{x_m} \relR_2 \chain m {x_0} \); and \( \link{p_0}{p_m} \relR_1 \chain {2m - 1} {p_0}{p_m} \).
  We can apply the up-to context technique with the context being \(
  \hole \mid \oprfx {x_m} {p_m} \hole \) to conclude the case.
\end{proof}

Now we prove the transitivity for the \OIwires{}.
Since the proof is almost identical to that of the \IOwires{}, we omit the details and only present the key points.
\begin{lem}
  \label{lem:linkO-trans}
  The wires \( \linkO p q \) and \( \vlinkOI x y \) are transitive.
  That is, we have \( \res q (\linkO p q \mid \linkO q r) \gexp \linkO p r \) and \( \res y (\vlinkOI x y \mid \vlinkOI y z) \gexp \vlinkOI x z \).
\end{lem}
\begin{proof}
\renewcommand*{\chain}[3]{\mathrm{chain}_{\mathtt{OI}}^{#1}(#2, #3)}
As in the case of the \IOwires{}, we consider the following relations.
\begin{align*}
    {\relR_1} &\defeq \left\{ (\linkO {p_0} {p_n}, \chain n {p_0} {p_n}) \mid n \ge 2 \right\} \\
    {\relR_2} &\defeq \left\{ (\vlinkOI {x_0}{x_n}, \chain n {x_0}{x_n})  \mid n \ge 2 \right\}.
\end{align*}
We show that \( {\relR_1} \cup {\relR_2}\) is an expansion up-to \( \lexp \) and context.

Observe that \( \chain n {p_0} {p_n} \) and  \( \chain n {x_0}{x_n} \) can only do an output at \( p_0 \) and input at \( x_0 \), respectively.
We can show that, for any \( n \ge 2 \),
\begin{enumerate}
  \item if \( \chain n {p_0} {p_n}  \trans{\bout {p_0} {x_0, q_0}} \proc \), \( \proc \gexp \diprfx {p_n} {x_n, q_n} (\chain {2n - 1} {x_0}{x_n} \mid \chain {2n - 1} {q_n}{q_0} ) \)
  \item if \( \chain n {x_0} {x_n}  \trans{x_0{(p_0)}} \proc \), then \( \proc \gexp  \chain n {x_0} {x_n} \mid \oprfx {x_n} {p_n} \chain {2n - 1} {p_n}{p_0} \)
\end{enumerate}
by induction on \( n \).

The rest of the proof follows that of Lemma~\ref{lem:linkI-trans}.
\end{proof}

\subsection{Proofs for laws other than transitivity}
We now show the remaining properties holds for all the three concrete wires.
Again, the reasoning is similar in all the three cases, though not identical.

\begin{lem}
\label{lem:link-satifies-ax}
  The \IOwires \( \link p q \) and \( \link x y \) satisfy the laws of Definition~\ref{ax:glink-piI}.
\end{lem}
\begin{proof}

  Requirements~\ref{it:ax:glink-piI:fn}, \ref{it:ax:glink-piI:no-tau}, and~\ref{it:ax:glink-piI:var-link-rep} hold by definition.
Transitivity of the wires has already been proved (Lemma~\ref{lem:linkI-trans}).
  Hence, we only check the remaining laws.

  We start by checking laws~\ref{it:ax:glink-piI:subst-din} and~\ref{it:ax:glink-piI:subst-dout}.
  Law~\ref{it:ax:glink-piI:subst-din} holds because
  \begin{align*}
    &\res q ( \link p q \mid \diprfx q {x, r} \proc )\\
    &\scong \resb {q, x, r}
      \begin{aligned}[t]
        (&\bin p {x'', r''} .\oprfx q {x', r'} (\link {x'} {x''} \mid \link {r''} {r'}) \\
        &\mid \iproc q {x', r'} {(\link x {x'} \mid \link {r'} r)} \mid \proc )
      \end{aligned} \\
    &\sbisim \resb {x, r}
      \begin{aligned}[t]
        (&\bin p {x'', r''} . \res q (\oprfx q {x', r'} (\link {x'} {x''} \mid \link {r''} {r'}) \\
        &\mid \iproc q {x', r'} {(\link x {x'} \mid \link {r'} r)})  \mid \proc )
      \end{aligned} \\
    &\gexp \resb {x, r}
      \begin{aligned}[t]
        (&\bin p {x'', r''} . \resb {x', r'} (\link {x'} {x''} \mid \link x {x'} \\
        &\mid  \link {r''} {r'} \mid \link {r'} r )  \mid \proc)
      \end{aligned} \tag{Lemma~\ref{lem:permeable-comm}} \\
    &\gexp \resb {x, r}  ( \iproc p {x'', r''} {(\link x {x''} \mid  \link {r''} r )}  \mid \proc ) \tag{transitivity of wires} \\
    &= \diprfx p {x, r} \proc
  \end{align*}
  Note that the assumption of Lemma~\ref{lem:permeable-comm} (interaction of permeable prefixes) is fulfilled by the transitivity of wires.
Next we consider law~\ref{it:ax:glink-piI:subst-dout}:
  \begin{align*}
    &\res p ( \link p q \mid \oprfx p {x,r} \proc ) \\
    &\scong \resb {p, x, r}
      \begin{aligned}[t]
        (&\bin p {x', r'} . \oprfx q {x'', r''} (\link {x''} {x'} \mid \link {r'} {r''}) \\
        &\mid \bout p {x', r'}. (\link {x'} x \mid \link r {r'}) \mid \proc)
      \end{aligned} \\
    &\gexp \resb {x, r, x', r'}
      \begin{aligned}[t]
        (&\oprfx q {x'', r''} (\link {x''} {x'} \mid \link {r'} {r''}) \\
        &\mid \link {x'} x \mid \link r {r'} \mid \proc)
      \end{aligned} \tag{Lemma~\ref{lem:permeable-comm}} \\
    &\gexp \resb {x, r, x', r', x'', r''}
      \begin{aligned}[t]
        (&\bout q {x''', r'''} . (\link {x'''} {x''} \mid \link {r''} {r'''}) \\
        &\mid  \link {x''} {x'} \mid \link {x'} x \\
        &\mid \link{r'}{r''} \mid \link r {r'} \mid \proc)
     \end{aligned} \tag{def. of \( \oprfx q {x'', r''}\)} \\
    &\gexp \resb{x, r, x'', r''}
      \begin{aligned}[t]
        (&\bout q {x''', r'''} . (\link {x'''} {x''} \mid \link {r''} {r'''}) \\
        &\mid  \link {x''} {x} \mid \link{r}{r''} \mid \proc )
        \end{aligned} \tag{transitivity of wires} \\
    &\gexp \resb {x'', r''} (\bout q {x''', r'''}. (\link {x'''} {x''} \mid \link {r''} {r'''}) \mid \proc \sub {x'', r''}{x, r})
        \tag{assumption on \( \proc \)} \\
    &\equiv_\alpha \resb {x, r} (\bout q {x', r'}. (\link {x'} x \mid \link r {r'}) \mid \proc ) \\
    &= \oprfx q {x, r} \proc
  \end{align*}

  We conclude by checking  laws~\ref{it:ax:glink-piI:subst-rep} and~\ref{it:ax:glink-piI:subst-dout-var}.
  We have
  \begin{align*}
    \res y ( \link x y \mid \riproc y p \proc )
    &= \res y ( \riproc x {p'} {\oprfx y p \link {p'} p} \mid \riproc y p \proc ) \\
    &\sbisim \riproc x {p'} {\res y (\oprfx y p \link {p'} p \mid \riproc y p \proc )}  \tag{replication theorem} \\
    &\gexp  \riproc x {p'} {\res p (\link p {p'} \mid \proc )} \tag{Lemma~\ref{lem:permeable-comm} and garbage collection on \( y \)} \\
    &\gexp \riproc x {p'} {\proc \sub p {p'}} \tag{assumption on \( \proc \)} \\
    &\equiv_\alpha \riproc x p \proc
  \end{align*}
  and
  \begin{align*}
    &\res x ( \link x y \mid \oprfx x p \proc ) \\
    &\scong \resb {x, p} ( \riproc x {p'} {\oprfx y {p''}  \link {p'} {p''}} \mid \bout x {p'} . \link p {p'} \mid \proc ) \\
    &\gexp \resb {p, p'} ( \oprfx y {p''} \link {p'} {p''} \mid \link p {p'} \mid \proc ) \tag{interaction at \( x \) and garbage collection on} \\
    &\gexp \resb {p, p', p''} ( \bout y {p'''}. \link {p''} {p'''} \mid \link{p'}{p''} \mid \link p {p'} \mid \proc ) \tag{def.~of \( \oprfx y {p''} \)} \\
    &\gexp \resb {p, p''} ( \bout y {p'''}. \link {p''} {p'''} \mid \link{p}{p''} \mid \proc ) \tag{transitivity of wires}  \\
    &\gexp \res {p''} ( \bout y {p'''} . \link {p''} {p'''} \mid \proc \sub {p''} p ) \tag{assumption on \( \proc \)}\\
    &\equiv_\alpha \res p ( \bout y {p'} . \link p {p'} \mid \proc ) \\
    &= \oprfx y p \proc. \tag*{\qedhere}
  \end{align*}
\end{proof}

The proof for the \OIwires{} is `symmetric' to that of the \IOwires{}.
\begin{lem}
\label{lem:linkO-satisfy-ax}
  The \OIwires \( \linkO q p \) and \( \vlinkOI x y \) satisfy the laws of Definition~\ref{ax:glink-piI}.
\end{lem}
\begin{proof}
  We already proved transitivity in Lemma~\ref{lem:linkO-trans}.
  The requirement~\ref{it:ax:glink-piI:fn}
  and~\ref{it:ax:glink-piI:var-link-rep} immediately 
follow from the definition.
Law~\ref{it:ax:glink-piI:subst-din} holds because
  \begin{align*}
    &\res q (\linkO q p \mid \diprfx q {x, r} \proc ) \\
    &= \res q (\linkO q p \mid (\res {x ,r}) ( \iproc q {x', r'} {(\vlinkOI x {x'} \mid  \linkO r {r'})} \mid \proc ) ) \\
    &\scong \resb {q, x, r}
      \begin{aligned}[t]
        (&\bout q {x', r'}.\diprfx p {x'', r''} (\vlinkOI {x'} {x''} \mid \linkO{r'} {r''})  \\
        &\mid \iproc q {x', r'} {(\vlinkOI x {x'} \mid  \linkO r {r'})} \mid \proc)
      \end{aligned} \\
    &\gexp \resb {x, x',  r, r'}
      \begin{aligned}[t]
        (&\diprfx p {x'', r''} (\vlinkOI {x'} {x''} \mid \linkO{r'} {r''})  \\
        &\mid \vlinkOI x {x'} \mid  \linkO r {r'} \mid \proc)
      \end{aligned} \tag{communication on a linear name} \\
    &\gexp
      \begin{aligned}[t]
        &\resb {x, x', x'', r, r', r''} \\
        &(\iproc p {x''', r'''} {(\vlinkOI {x''} {x'''} \mid \linkO{r''} {r'''})}  \\
        &\mid \vlinkOI x {x'} \mid \vlinkOI {x'}{x''} \mid \linkO r {r'} \mid \linkO {r'} {r''} \mid \proc)
      \end{aligned} \tag{def. of \( \diprfx p {x'', r''} \)} \\
    &\gexp \resb {x, x'',  r, r''}
      \begin{aligned}[t]
        (&\iproc p {x''', r'''} {(\vlinkOI {x''} {x'''} \mid \linkO{r''} {r'''})}  \\
        &\mid \vlinkOI x {x''}  \mid \linkO r {r''} \mid \proc)
      \end{aligned} \tag{transitivity of wires} \\
    &\gexp \resb {x'', r''}
      \begin{aligned}[t]
        (&\iproc p {x''', r'''}{(\vlinkOI {x''} {x'''} \mid \linkO{r''} {r'''})}  \\
        &\mid \proc \sub{x'', r''}{x, r} )
      \end{aligned} \tag{assumption on \( \proc \)} \\
    &\equiv_\alpha \resb {x, r} ( \iproc p {x', r'}{(\vlinkOI {x} {x'} \mid \linkO {r} {r'})}  \mid \proc ) \\
    &= \diprfx p {x, r} \proc
  \end{align*}

  The proof of law~\ref{it:ax:glink-piI:subst-dout} is similar to that of~\ref{it:ax:glink-piI:subst-din}, but does not use the respectfulness of \( \proc\) :
  \begin{align*}
    &\res p (\linkO q p \mid \oprfx p {x, r} \proc ) \\
    &\scong (\res {p, x, r})
      \begin{aligned}[t]
        (&\bout q {x'', r''}. \diprfx p {x', r'} (\vlinkOI {x''} {x'} \mid \linkO {r''} {r'}) \\
        &\mid \bout p {x', r'}.(\vlinkOI {x'} {x} \mid \linkO {r'} r ) \mid \proc)
      \end{aligned} \\
    &\sbisim (\res {x, r})
      \begin{aligned}[t]
        (\bout q {x'', r''}. \res p (&\diprfx p {x', r'} (\vlinkOI {x''} {x'} \mid \linkO {r''} {r'}) \\
        &\mid \bout p {x', r'}.(\vlinkOI {x'} {x} \mid \linkO {r'} r )) \mid \proc)
      \end{aligned}\\
    &\gexp (\res {x, r})
      (\bout q {x'', r''}. (\res {x', r'})
      (\vlinkOI {x''} {x'} \mid \vlinkOI {x'} {x} \mid \linkO {r''} {r'} \mid \linkO {r'} r) \mid \proc)
      \tag{Lemma~\ref{lem:permeable-comm}} \\
    &\gexp (\res {x, r}) ( \bout q {x'', r''}. (\vlinkOI {x''} {x} \mid \linkO {r''} {r} ) \mid \proc ) \tag{transitivity of wires} \\
    &= \oprfx q {x, r} \proc
  \end{align*}

  The proofs for laws~\ref{it:ax:glink-piI:subst-rep} and~\ref{it:ax:glink-piI:subst-dout-var} are the same as those for \IOwires.
Law~\ref{it:ax:glink-piI:subst-rep} holds because:
  \begin{align*}
    \res y (\vlinkOI x y \mid \riproc y p \proc)
    &= \res y (\riproc x {p'} {\oprfx y p \linkO p {p'}} \mid \riproc y p \proc) \\
    &\sbisim \riproc x {p'} {\res y (\oprfx y p \linkO p {p'} \mid \riproc y p \proc)} \tag{replication theorem} \\
    &\gexp \riproc x {p'} {\res p (\linkO p {p'} \mid  \proc)} \tag{Lemma~\ref{lem:permeable-comm} and garbage collection} \\
    &\gexp \riproc x {p'} {\proc \sub{p'} p}. \tag{assumption on \( \proc \)}
  \end{align*}
  Finally, law~\ref{it:ax:glink-piI:subst-dout-var} holds because:
  \begin{align*}
    &\res x (\vlinkOI x y \mid \oprfx x {p} \proc) \\
    &\scong  (\res {x, p}) (\riproc x {p'} {\oprfx y  {p''} \linkO{p''} {p'}} \mid \bout x {p'}. \linkO {p'} p  \mid \proc) \\
    &\gexp ({\res p, p'}) (\oprfx y  {p''}  \linkO{p''} {p'} \mid \linkO {p'} p  \mid \proc) \tag{reduction and garbage collection on \( x \)} \\
    &\gexp \resb {p, p', p''} (\bout y  {p'''}.\linkO{p'''}{p''} \mid  \linkO{p''} {p'} \mid \linkO {p'} p  \mid \proc) \tag{def. of \( \oprfx y {p''} \)} \\
    &\gexp \resb {p, p''} (\bout y  {p'''}.\linkO{p'''}{p''} \mid  \linkO{p''} p \mid \proc) \tag{transitivity} \\
    &\gexp \res {p''} (\bout y  {p'''}. \linkO{p'''} {p''} \mid \proc \sub {p''}{p}) \tag{assumption on \( \proc \)} \\
    &\equiv_\alpha  \res p (\bout y  {p'}. \linkO{p'} {p} \mid \proc) \\
    &= \oprfx y {p} \proc \tag*{\qedhere}
  \end{align*}
\end{proof}

We conclude the section by checking that also \Pwires{}  satisfy the desired properties.
\begin{lem}
  The \Pwires{} \( \ilink q p \) and \( \olink x y \) satisfy the laws of Definition~\ref{ax:glink-piI}.
\end{lem}
\begin{proof}
Transitivity has already been proved in
Lemma~\ref{lem:ilink-olink-trans}, and the requirements on free names
and the shape of \( \glink x y \) follow
 from the definition.
  So we only check the remaining laws.

  We now consider law~\ref{it:ax:glink-piI:subst-din}.
  Although the definition of \( \ilink p q \) is  more complex than
  that for the other wires, the proof is similar; we can remove the wires using transitivity and the assumption on \( \proc \).
  \begin{align*}
    &\res q (\ilink p q \mid \diprfx q {x, r} \proc) \\
    &=
      \begin{aligned}[t]
        &(\res {x, y, q, r, s}) \\
        &(\iproc p {x', r'} {(\olink x {x'} \mid \ilink {r'} r)} \\
        &\mid \bout q {y', s'}.(\olink {y'} y \mid \ilink s {s'}) \\
        &\mid \olink y x \mid \ilink r s  \\
        &\mid (\res x, r)( \iproc q {x', r'} {(\olink x {x'} \mid  \ilink {r'} r)} \mid \proc ))
      \end{aligned} \\
    &\gexp (\res {x, y, y', r, s, s'})
      \begin{aligned}[t]
        (&\iproc p {x', r'} {(\olink x {x'} \mid \ilink {r'} r)} \\
        &\mid \olink {y'} y \mid \ilink s {s'} \mid \olink y x \mid \ilink r s  \\
        &\mid (\res x, r) (\olink x {y'} \mid  \ilink {s'} r \mid \proc ))
      \end{aligned} \tag{communication on \( q \)}\\
    &\gexp (\res {x, y', r, s'})
      \begin{aligned}[t]
        (&\iproc p {x', r'} {(\olink x {x'} \mid \ilink {r'} r)} \\
        &\mid \olink {y'} x \mid \ilink r {s'} \\
        &\mid (\res x, r) (\olink x {y'} \mid  \ilink {s'} r \mid \proc ))
      \end{aligned} \tag{transitivity of wires} \\
    &\equiv_\alpha (\res {x'', y', r'', s'})
      \begin{aligned}[t]
        (&\iproc p {x', r'} {(\olink {x''} {x'} \mid \ilink {r'} {r''})} \\
        &\mid \olink {y'} {x''} \mid \ilink {r''} {s'} \\
        &\mid (\res x, r) (\olink x {y'} \mid  \ilink {s'} r \mid \proc ))
      \end{aligned} \\
    &\gexp (\res {x,  x'', r, r''})
      \begin{aligned}[t]
        (&\iproc p {x', r'} {(\olink {x''} {x'} \mid \ilink {r'} {r''})} \\
        &\mid \olink x {x''} \mid  \ilink {r''} r \mid \proc)
    \end{aligned} \tag{transitivity of wires} \\
    &\gexp  (\res {x'', r''}) ( \iproc p {x', r'} {(\olink {x''} {x'} \mid \ilink {r'} {r''})} \mid  \proc \sub {x'', r''} {x, r}) \tag{assumption on \( \proc \)} \\
    &\equiv_\alpha (\res {x, r}) ( \iproc p {x', r'} {(\olink x {x'} \mid \ilink {r'} r)} \mid  \proc ) \\
    &= \diprfx p {x, r} \proc
  \end{align*}

  The proof for law~\ref{it:ax:glink-piI:subst-dout} is the dual of the previous case.

  Now we check  law~\ref{it:ax:glink-piI:subst-rep}.
  The reasoning is identical to the case for \IOwires{} and \OIwires{}.
  \begin{align*}
    \res y (\olink x y \mid \riproc y p {\proc})
    &= \res y ( \riproc x {p'} {\oprfx y {p} \ilink {p'} p} \mid \riproc y p \proc ) \\
    &\sbisim \riproc x {p'} {\res y (\oprfx y {p} \ilink {p'} p \mid \riproc y p \proc )} \tag{replication theorem} \\
    &\gexp \riproc x {p'} {\res p (\ilink {p'} p \mid \proc)} \tag{Lemma~\ref{lem:permeable-comm} and garbage collection} \\
    &\gexp \riproc x {p'} {\proc \sub {p'} p }  \tag{assumption on \( \proc \)} \\
    &\equiv_\alpha \riproc {x} {p} \proc
  \end{align*}

  The last thing to check is  law~\ref{it:ax:glink-piI:subst-dout-var}.
  Again, the reasoning is identical to the case for \IOwires{} and \OIwires{}.
  \begin{align*}
    &\res x (\olink x y \mid \oprfx x {p} \proc) \\
    &= \res x ( \riproc x p {\oprfx y {q} \ilink p q} \mid \res p ( \bout x {p'}. \ilink p {p'}  \mid \proc ) ) \\
    &\gexp (\res {p, p' q}) (\bout y {q'}. \ilink q {q'} \mid \ilink {p'} q \mid \ilink p {p'}  \mid \proc ) \tag{interaction at \( x \) and garbage collection} \\
    &\gexp  (\res {p, q}) (\bout y {q'}. \ilink q {q'} \mid \ilink p q \mid \proc ) \tag{transitivity} \\
    &\gexp  \res q (\bout y {q'}. \ilink q {q'} \mid \proc \sub q p ) \tag{assumption on \( \proc \)} \\
    &\equiv_\alpha \res p (\bout y {p'}. \ilink p {p'} \mid \proc ) \\
    &= \oprfx y {p} \proc \tag*{\qedhere}
  \end{align*}
\end{proof}

%% file: BT-LT-appx.tex
\section{Supplementary Materials for Section~\ref{sec:BT-LT}}
\label{sec:BT-LT-appx}
We present the proofs that were omitted from the main text.

\subsection{Proofs for the properties of encoding of unsolvable terms}
\label{sec:BT-LT-appx:unsolv}
In the main text, we have seen what kind of transition \( \oencoOI M p \) and \( \oencoI M p \) can do when \( M \) is an unsolvable term.
Notably, the only visible action these processes can do is an input at \( p \), and the behaviour of \( \QoencoOI \) differs from that of \( \QoencoIO \).
We give the proofs for these results.

\oencoOZeroTau*
\begin{proof}
If \( M \) is unsolvable of order \( 0 \) then it must be of the form \( \app {(\lambda x. M_0)} {\seq M \)}, with \( \seq M \) non-empty.
Therefore, by  definition of \( \QoencoOI \), if \( M \) is unsolvable of order \( 0 \), \( \oencoO M p \) cannot do an input at \( p \).
This implies that the only transition \( \oencoO M p \) can do is a \( \tau \)-transition.
\end{proof}

\absUnsolvnIn*
\begin{proof}
  By definition of the order of  unsolvables, we have \( M \Hred \lambda x. M' \) for some \( M' \) whose order is \( n - 1\).
  By  validity of \( \beta \)-reduction (Lemma~\ref{lem:gen-oenc-validates-beta}), we have \( \oencoO M p \gexp \oencoO {\lambda x. M'} p \).

  First, we show that \( M \) can do a weak input transition at \( p \).
  Since \( \oencoO {\lambda x. M'} p \trans {p(x, q)} \oencoO {M'} q \), we must have a matching transition  \( \oencoO M p \wtrans{p(x, q)} \proc \) for some \( \proc \).

  We remain to show that for every \( \oencoO M p \wtrans{p(x, q)} \proc \) there is a suitable \( N \) with \( \proc \gexp \oencoO {N} q \).
  Assume that \( \oencoO M p (\trans{\tau})^n \proctwo \trans{p(x, q)} \proc \).
  Then we have \( \oencoO{\lambda x. M'} p (\trans{\hat \tau})^n \proctwo' \) such that \( \proctwo \gexp \proctwo' \).
  By Lemma~\ref{lem:gen-oenc-tau-has-mathcing-red}, we have \( \proctwo' \gexp  \oencoO {\lambda x. M''} p \) for a \( \lambda \)-term \( \lambda x. M'' \) such that \( \lambda x. M' \wred \lambda x. M'' \).
  Since \( M' \) is unsolvable of order \( n - 1\), so is \( M'' \).
  We also have \( \proc \gexp \oenco {M''} q \), because \( \proctwo
  \gexp \oencoO {\lambda x. M''} p \) and \( \oencoO {\lambda x. M''}
  p \trans{p(x, q)} \oencoO {M''} q \) is the only input transition
  that  \( \oencoO {\lambda x. M''} p \) can do.
\end{proof}

\oencoIUnsolvInput*
\begin{proof}
To prove~\ref{it:lem:oencoI-unsolv-input:zero}, first observe that \( M \), an unsolvable term of order \( 0 \), must be of the form \( \app {(\lambda x . M_0)} {M_1 \cdots M_n} \) for \( n \ge 0 \).
Hence, we have
\begin{align*}
  &\oencoI {M} p \\
  &= \res {p_0}
  \begin{aligned}[t]
  (&\diprfx {p_0}{x, r} \oencoI {M_0} r \mid \oprfx {p_0}{x_1,  p_1} \cdots \oprfx {p_{n - 1}}{x_n, p_n} \\
  & ( \riproc {x_1} {r_1} {\oencoI {M_1} {r_1}} \mid \cdots \mid \riproc {x_n} {r_n} {\oencoI {M_n} {r_n}} \mid \link p {p_n}))
  \end{aligned} \\
  &\trans{p(x, q)}\equiv \res {p_0}
    \begin{aligned}[t]
      &(\diprfx {p_0}{x, r} \oencoI {M_0} r  \\
      & \mid \oprfx {p_0}{x_1,  p_1} \cdots \oprfx {p_{n - 1}}{x_n, p_n} \oprfx {p_n}{x_{n + 1}, {p_{n + 1}}} \\
      & (\riproc {x_1} {r_1} {\oencoI {M_1} {r_1}} \mid \cdots \mid \riproc {x_n} {r_n} {\oencoI {M_n} {r_n}} \mid \link {x_{n + 1}} {x} \mid \link q {p_{n + 1}} ))
    \end{aligned} \\
  &= \res {p_0}
    \begin{aligned}[t]
    &( \diprfx {p_0}{x, r} \oencoI {M_0} r \mid \oprfx {p_0}{x_1,  p_1} \cdots \oprfx {p_{n - 1}}{x_n, p_n} \oprfx {p_n}{x_{n + 1}, {p_{n + 1}}} \\
    &(( \riproc {x_1} {r_1} {\oencoI {M_1} {r_1}} \mid \mkern-1.5mu \cdots \mkern-1.5mu \mid \riproc {x_n} {r_n} {\oencoI {M_n} {r_n}} \mid \riproc {x_{n + 1}} {r_{n + 1}} {\oencoI x {r_{n + 1}}} \mid \link q {p_{n + 1}} ))
\end{aligned} \tag{since \( \link x y = \riproc x p {\oencoI{y} p}\)} \\
&= \oencoI {\app {(\lambda x. M_0)} {M_1 \cdots M_n \; x}} p
\end{align*}
as desired.

We prove~\eqref{it:lem:oencoI-unsolv-input:any-ord} by a case analysis on the order of \( M \).
If \( M \) is an unsolvable of order \( 0 \), the claim follows from~\eqref{it:lem:oencoI-unsolv-input:zero} because \( \app M x \) is also an unsolvable of order \( 0 \).

Now assume that \( M \) is an unsolvable of order \( n > 0 \) (\( n \) may be \( \omega \)).
The fact that \( \oencoI M p \) can do an input at \( p \) follows by the definition of \( \QoencoIO \).
We remain to prove the latter claim.
Observe that we have \( M \Hred \lambda x . M' \), and \( M' \) must be an unsolvable.
By Lemma~\ref{lem:gen-oenc-validates-beta}, we have \( \oencoI M p \gexp \oencoI {\lambda x.M'} p \).
Thus, if \( \oencoI M p \trans{p(x, q)} \proc \), we have a matching transition \( \oencoI {\lambda x. M'} p \trans{p(x, q)} \proctwo \) such that \( \proc \gexp \proctwo \).
By  definition of \( \QoencoIO \) and of permeable inputs, and by Lemma~\ref{lem:glink-subst-oenc}, we have \( \proctwo \gexp \oencoI {M'} q \).
Hence, \( \proc \gexp \oencoI {M'} q  \) as desired.
\end{proof}

\subsection{Proof for the inverse context lemma}
\label{sec:BT-LT-full-abst-main}
\inverseContext*
\begin{proof}
   \mylabel{Abstraction context}
   For our encoding, the abstraction context is defined by
   \[
     C^x_{\lambda} \defeq \bind p  {\resb{x,q} ( \inp p {x',q'}.  (\glink {q'} q | \glink x x' ) \mid \appI \hole q )}.
   \]
   We define the inverse context by
   \[
     D \defeq \bout a {b, x}. \iproc b r {\res p  (\appI \hole p \mid  \bout p {x', q'}.(\glink {x'} {x} \mid \glink r  {q'}) )}.
   \]

   Then
   \begin{align*}
     &D[C[F]] \\
     &=
       \begin{aligned}[t]
         &\bout a {b, x}. \iproc b r {\res p  (\resb{x,q} ( \inp p {x',q'}.  (\glink {q'} q \mid \glink x {x'} ) \\
         &\mid \appI F q )  \mid  \bout p {x', q'}.(\glink {x'} x \mid \glink r {q'}) )}
       \end{aligned} \\
     &\equiv_\alpha
       \begin{aligned}[t]
         &\bout a {b, x}. \bin b r. \res p  (\resb{z,q} ( \inp p {x',q'}.  (\glink {q'} q \mid \glink z {x'}  ) \\
         &\mid (\appI F q) \sub z x )  \mid  \bout p {x', q'}.(\glink {x'} x \mid \glink r q) )
       \end{aligned} \tag{\( z \) fresh} \\
     &\gexp
       \begin{aligned}[t]
         &\bout a {b, x}. \bin b r . \resb {z, x', q, q'}  (\glink q {q'} \mid \glink z {x'}  \\
         &\mid (\appI F q) \sub z  x   \mid  \glink {x'} {x}  \mid \glink r q)
         \end{aligned}\tag{communication on \( p \)} \\
     &\gexp \bout a {b, x}. \iproc b r {(\res {z,  q'})  (\glink r {q'} \mid \glink z x  \mid (\appI F {q'})  \sub z x )}  \tag{transitivity of wires} \\
     &\gexp \bout a {b, x}. \iproc b r {\appI F r} \tag{Lemma~\ref{lem:glink-subst-oenc} and \( F = \oenco M {} \)}
   \end{align*}
   as desired.

   \mylabel{Variable context}
   For a variable context obtained by translating \( \app x {\hole_1 \cdots \hole_n} \) the inverse context \( D \) for the \( i \)-th hole can be defined by
  \[
     \begin{aligned}[t]
       &\bout a {x', b}. b(r). (\res {x, p}) (\appI \hole p \\
       &\mid x(p_0).p_0 (x_1, p_1). \ldots p_{n - 1}(x_n,  p_n).(\glink {x'} x \mid \bout {x_i}{r'}. \glink r {r'}))
     \end{aligned}.
   \]
   Now let us consider the process \( D[C[\seq F]] \).
   Since the communication on \( p_i \) is a communication on linear names, we can safely execute theses communications.
   (Note that since permeable input prefixes are encoded using wires, there will be unguarded wires after the reductions)
   With this in mind, we get
   \begin{align*}
     &D[C[\seq F]] \\
     &\gexp
       \begin{aligned}[t]
         &\bout a {x', b}. b(r).(\res {x, x_1, \ldots, x_n, x'_1, \ldots, x'_n, p_n}) \\
         (&\riproc {x_1} {r_1} {\appI {F_1} {r_1}} \mid \cdots \mid \riproc {x_n} {r_n}{\appI {F_n} {r_n}} \mid \\
         &\glink {x'_1} {x_1} \mid \cdots \mid \glink {x'_n}{x_n} \mid \\
         &\glink p {p_n} \mid \glink {x'} x \mid \bout {x_i'}{r'}. \glink r {r'})
         \end{aligned} \\
     &\gexp  \begin{aligned}[t]
       &\bout a {x', b}. b(r).(\res {x, x'_1, \ldots, x'_n, p_n}) \\
       (&\riproc {x'_1} {r_1} {\appI {F_1} {r_1}} \mid \cdots \mid \riproc {x'_n} {r_n}{\appI {F_n} {r_n}} \mid   \\
       &\glink p {p_n} \mid \glink {x'} x \mid \bout {x_i'}{r'}. \glink r {r'} )
      \end{aligned} \tag{\ref{it:ax:glink-piI:subst-rep} of Definition~\ref{ax:glink-piI} and \( F_j =  \oenco   {M_j} {} \)}\\
     &\sbisim  \bout a {x', b}. b(r).(\res {x, x'_i})
       \begin{aligned}[t]
         (&\riproc {x'_i} {r_i} {\appI {F_i} {r_i}} \mid \glink {x'} x \\
         &\mid \bout {x'_i}{r'}. \glink r {r'})
       \end{aligned} \tag{garbage collection} \\
     &\gexp  \bout a {x', b}. b(r).(\res {x, r'})\left(\appI {F_i} {r'} \mid \glink {x'} x \mid \glink r {r'} \right) \tag{communication on \( x'_i \) and garbage collection} \\
     &\gexp   \bout a {x', b}. b(r).(\appI {F_i} {r})\sub {x'} {x} \tag{Lemma~\ref{lem:glink-subst-oenc} and \( F_i = \oenco {M_i} {}\)} \\
     &\equiv_{\alpha}   \bout a {x, b}. b(r).\appI {F_i} {r} \tag*{\qedhere}
   \end{align*}
 \end{proof}